\pdfoutput=1
\documentclass[acmsmall, nonacm, authorversion, noccs]{acmart}
\setcopyright{acmlicensed}
\copyrightyear{2018}
\acmYear{2018}
\acmDOI{XXXXXXX.XXXXXXX}
\acmConference[Conference acronym 'XX]{Make sure to enter the correct
  conference title from your rights confirmation email}{June 03--05,
  2018}{Woodstock, NY}
\acmISBN{978-1-4503-XXXX-X/2018/06}

\usepackage{amsfonts,amsmath,amssymb,amsthm}
\usepackage{bbm}
\usepackage{tikz-cd}
\usepackage{mathtools}
\usepackage{thmtools}
\usepackage{hyperref}
\usepackage{float}
\usepackage{quiver}
\usepackage{tcolorbox}
\usepackage[graphicx]{realboxes}
\usepackage{adjustbox}
\usepackage{rotating}
\usepackage[capitalise]{cleveref}
\usepackage{xspace}
\usepackage{multicol}      
\usepackage{enumitem}

\newcommand{\BA}[1]{\textcolor{blue}{BA: #1}}
\newcommand{\PN}[1]{\textcolor{purple}{PN: #1}}

\newcommand{\pinkduplicates}{1}

\newcommand{\dup}[1]{%
  \if\pinkduplicates1
    \textcolor{magenta}{#1}%
  \else
    #1%
  \fi
}

\newcommand{\LongOrShort}[2]{#2}
\newcommand{\LongOrShortPoPl}[2]{#2}
\newcommand{\BigOrSmallRules}[2]{#2}
\newcommand{\TablesOrNot}[2]{#2}

\newcommand{\emphasize}[1]{\textbf{#1}}

\newcommand{\type}{\ensuremath{\mathsf{~type}}}

\newcommand{\ctx}{\ensuremath{\mathsf{~ctx}}}

\newcommand{\set}{\ensuremath{\mathsf{Set}}}
\newcommand{\Set}{\set}

\newcommand{\id}{\ensuremath{\mathsf{Id}}}

\newcommand{\emptyctx}{\diamond}
\newcommand{\iso}{\cong}

\newcommand{\dom}{\text{dom}}
\newcommand{\cod}{\text{cod}}
\newcommand{\Tau}{\ensuremath{\mathcal{T}}}
\newcommand{\TT}{\Tau}
\newcommand{\CC}{\ensuremath{\mathcal{C}}}
\newcommand{\DD}{\ensuremath{\mathcal{D}}}

\newcommand{\EE}{\ensuremath{\mathcal{E}}}

\newcommand{\app}{\ensuremath{{\mathsf{app}}}}

\newcommand{\Cat}{\ensuremath{{\mathsf{Cat}}}}
\newcommand{\EM}{\ensuremath{{\mathsf{EM}}}}

\newcommand{\AWFS}{AWFS\xspace}
\newcommand{\AWFSs}{AWFSs\xspace}

\newcommand{\PiT}[2]{\ensuremath{\Pi(#1,#2)}}
\newcommand{\SigmaT}[2]{\ensuremath{\Sigma(#1,#2)}}

\newcommand{\compiso}[3]{\ensuremath{i^{\mathsf{comp}}_{#1, #2, #3}}}
\newcommand{\compisoinv}[3]{{\ensuremath{i^{\mathsf{comp}^{-1}}_{#1, #2, #3}}}}
\newcommand{\idiso}[1]{\ensuremath{i^{\mathsf{id}}_{#1}}}
\newcommand{\idisoinv}[1]{\ensuremath{i^{\mathsf{id}^{-1}}_{#1}}}
\newcommand{\subiso}[3]{\ensuremath{i^\mathsf{sub}_{#1, #2, #3}}}
\newcommand{\subisoinv}[3]{{\ensuremath{i^{\mathsf{sub}^{-1}}_{#1, #2, #3}}}}
\newcommand{\pairmor}[2]{\ensuremath{\mathsf{pair}_{\SigmaT{#1}{#2}}}}
\newcommand{\projmor}[2]{\ensuremath{\mathsf{proj}_{\SigmaT{#1}{#2}}}}

\newcommand{\appmor}[2]{\ensuremath{\app_{{\PiT{#1}{#2}}}}}
\newcommand{\lambdamor}[2]{\ensuremath{\lambda_{\PiT{#1}{#2}}}}
\newcommand{\reflmor}[1]{\ensuremath{r_{#1}}}
\newcommand{\jmor}[3]{\ensuremath{j_{#1, #2, #3}}}

\newcommand{\subtypesigma}[2]{{\ensuremath{\Sigma(#1,#2)}}}
\newcommand{\subtypepi}[2]{{\ensuremath{\Pi(#1,#2)}}}
\newcommand{\subtypeid}[1]{{\ensuremath{\id_{#1}}}}

\newcommand{\interp}[1]{\ensuremath{{\llbracket {#1} \rrbracket}}}

\newcommand{\CTT}{\textsf{CCTT}\xspace}
\newcommand{\emphCTT}{\emphasize{\CTT}\xspace}
\newcommand{\CTTsplit}{\CTT_{\textsf{split}}\xspace}

\newcommand{\MLTTcoe}{\(\text{MLTT}_{\text{coe}}\)\xspace}
\newcommand{\NatModDO}{\(\text{NatMod}_{\text{DO}}\)\xspace}

\newcommand{\vdashiso}{%
\>
\mathbin{\vbox{\offinterlineskip\ialign{%
\hfil##\hfil\cr
$\scriptstyle\sim$\cr
$\vdash$\cr
}}}
\>
}

\usepackage{bussproofs}

\newcommand{\proofskip}{0.5em}
\newenvironment{bproof}
{\leavevmode\hbox\bgroup}
{\DisplayProof\egroup}

\newcommand{\unaryRule}[3]{\begin{bproof}
\AxiomC{\ensuremath{#1}}
\RightLabel{\rulelabel{#3}}
\UnaryInfC{\ensuremath{#2}}
\end{bproof}
\BigOrSmallRules{
\vspace{\proofskip}
}
{}
}

\newcommand{\unaryRuleBinaryConcl}[4]{\begin{bproof}
\AxiomC{\ensuremath{#1}}
\RightLabel{\rulelabel{#4}}
\UnaryInfC{\ensuremath{#2}}
\def\extraVskip{0.5pt}
\noLine
\UnaryInfC{\ensuremath{#3}}
\end{bproof}
\BigOrSmallRules{
\vspace{\proofskip}
}
{}
}

\newcommand{\binaryRule}[4]{\begin{bproof}
\AxiomC{\ensuremath{#1}}
\AxiomC{\ensuremath{#2}}
\RightLabel{\rulelabel{#4}}
\BinaryInfC{\ensuremath{#3}}
\end{bproof}
\BigOrSmallRules{
\vspace{\proofskip}
}
{}
}

\newcommand{\ternaryRule}[5]{\begin{bproof}
\AxiomC{\ensuremath{#1}}
\AxiomC{\ensuremath{#2}}
\AxiomC{\ensuremath{#3}}
\RightLabel{\rulelabel{#5}}
\TrinaryInfC{\ensuremath{#4}}
\end{bproof}
\BigOrSmallRules{
\vspace{\proofskip}
}
{}
}

\newcommand{\quadRule}[6]{\begin{bproof}
\AxiomC{\ensuremath{#1}}
\AxiomC{\ensuremath{#2}}
\AxiomC{\ensuremath{#3}}
\AxiomC{\ensuremath{#4}}
\RightLabel{\rulelabel{#6}}
\QuaternaryInfC{\ensuremath{#5}}
\end{bproof}
\BigOrSmallRules{
\vspace{\proofskip}
}
{}
}

\newcommand{\senRule}[8]{\begin{bproof}
\AxiomC{\ensuremath{#1}}
\AxiomC{\ensuremath{#2}}
\AxiomC{\ensuremath{#3}}
\def\extraVskip{0.5pt}
\noLine
\TrinaryInfC{\ensuremath{#4} \quad \ensuremath{#5} \quad \ensuremath{#6}}
\def\extraVskip{2pt}
\RightLabel{\rulelabel{#8}}
\UnaryInfC{\ensuremath{#7}}
\end{bproof}
\BigOrSmallRules{
\vspace{\proofskip}
}
{}
}

\newcommand{\binaryRuleBinaryConcl}[5]{\begin{bproof}
\AxiomC{\ensuremath{#1}}
\AxiomC{\ensuremath{#2}}
\RightLabel{\rulelabel{#5}}
\BinaryInfC{\ensuremath{#3}}
\def\extraVskip{0.5pt}
\noLine
\UnaryInfC{\ensuremath{#4}}
\end{bproof}
\BigOrSmallRules{
\vspace{\proofskip}
}
{}
}

\newcommand{\binaryRuleTernaryConcl}[6]{\begin{bproof}
\AxiomC{\ensuremath{#1}}
\AxiomC{\ensuremath{#2}}
\RightLabel{\rulelabel{#6}}
\BinaryInfC{\ensuremath{#3}}
\def\extraVskip{0.5pt}
\noLine
\UnaryInfC{\ensuremath{#4}}
\def\extraVskip{0.5pt}
\noLine
\UnaryInfC{\ensuremath{#5}}
\end{bproof}
\BigOrSmallRules{
\vspace{\proofskip}
}
{}
}

\newcommand{\ternaryRuleBinaryConcl}[6]{\begin{bproof}
\AxiomC{\ensuremath{#1}}
\AxiomC{\ensuremath{#2}}
\AxiomC{\ensuremath{#3}}
\RightLabel{\rulelabel{#6}}
\TrinaryInfC{\ensuremath{#4}}
\def\extraVskip{0.5pt}
\noLine
\UnaryInfC{\ensuremath{#5}}
\end{bproof}
\BigOrSmallRules{
\vspace{\proofskip}
}
{}
}

\newcommand{\ternaryRuleTernaryConcl}[7]{\begin{bproof}
  \AxiomC{\ensuremath{#1}}
  \AxiomC{\ensuremath{#2}}
  \AxiomC{\ensuremath{#3}}
  \RightLabel{\rulelabel{#7}}
  \TrinaryInfC{\ensuremath{#4}}
  \def\extraVskip{0.5pt}
  \noLine
  \UnaryInfC{\ensuremath{#5}}
  \def\extraVskip{0.5pt}
  \noLine
  \UnaryInfC{\ensuremath{#6}}
  \end{bproof}
\BigOrSmallRules{
\vspace{\proofskip}
}
{}
}

\newcommand{\quadRuleBinaryConcl}[7]{\begin{bproof}
\AxiomC{\ensuremath{#1}}
\AxiomC{\ensuremath{#2}}
\AxiomC{\ensuremath{#3}}
\AxiomC{\ensuremath{#4}}
\RightLabel{\rulelabel{#7}}
\QuaternaryInfC{\ensuremath{#5}}
\def\extraVskip{0.5pt}
\noLine
\UnaryInfC{\ensuremath{#6}}
\end{bproof}
\BigOrSmallRules{
\vspace{\proofskip}
}
{}
}

\newcommand{\quadRuleTernaryConcl}[8]{\begin{bproof}
\AxiomC{\ensuremath{#1}}
\AxiomC{\ensuremath{#2}}
\AxiomC{\ensuremath{#3}}
\AxiomC{\ensuremath{#4}}
\RightLabel{\rulelabel{#8}}
\QuaternaryInfC{\ensuremath{#5}}
\def\extraVskip{0.5pt}
\noLine
\UnaryInfC{\ensuremath{#6}}
\def\extraVskip{0.5pt}
\noLine
\UnaryInfC{\ensuremath{#7}}
\end{bproof}
\BigOrSmallRules{
\vspace{\proofskip}
}
{}
}

\newcommand{\rulelabel}[1]{\hypertarget{#1}{\textsf{\scriptsize #1}}}
\newcommand{\ruleref}[1]{\hyperlink{#1}{\textsf{#1}}}

\crefname{relatedwork}{Related Work}{Related Work}

\AtEndPreamble{%
\theoremstyle{acmdefinition}
\newtheorem{remark}[theorem]{Remark}
\newtheorem{notation}[theorem]{Notation}
\newtheorem{construction}[theorem]{Construction}
\newtheorem{relatedwork}[theorem]{Related Work}
}

\begin{document}

\title{From Semantics to Syntax: A Type Theory for Comprehension Categories}
\author{Niyousha Najmaei}
\email{najmaei@lix.polytechnique.fr}
\orcid{0009-0001-2892-5537}
\affiliation{%
  \institution{\'Ecole Polytechnique}
  \country{France}
}

\author{Niels van der Weide}
\email{nweide@cs.ru.nl}
\orcid{0000-0003-1146-4161}
\affiliation{%
  \institution{Radboud University Nijmegen}
  \country{Netherlands}}

\author{Benedikt Ahrens}
\email{B.P.Ahrens@tudelft.nl}
\orcid{0000-0002-6786-4538}
\affiliation{%
  \institution{Delft University of Technology}
  \country{Netherlands}
}

\author{Paige Randall North}
\email{p.r.north@uu.nl}
\orcid{0000-0001-7876-0956}
\affiliation{%
 \institution{Utrecht University}
 \country{Netherlands}}


\begin{abstract}
  Recent models of intensional type theory have been constructed in algebraic weak factorization systems (\AWFSs).
  \AWFSs give rise to comprehension categories that feature non-trivial morphisms between types;
  these morphisms are not used in the standard interpretation of Martin-Löf type theory in comprehension categories.

  We develop a type theory that internalizes morphisms between types, reflecting this semantic feature back into syntax.
  Our type theory comes with $\Pi$-, $\Sigma$-, and identity types.
  We discuss how it can be viewed as an extension of Martin-Löf type theory with coercive subtyping, as sketched by Coraglia and Emmenegger.
  We furthermore define semantic structure that interprets our type theory and prove a soundness result.
  Finally, we exhibit many examples of the semantic structure, yielding a plethora of interpretations.

\end{abstract}



\maketitle

\section{Introduction}

There is a fruitful synergy between programming languages and mathematics, and more precisely between the study of dependent type theories and category theory.
Category theory provides ways to give semantics to type theory, enabling proofs of consistency, normalization, and other desirable properties \cite{Kapulkin2021,DBLP:conf/lics/Gratzer22}.
On the other hand, type theories provide internal languages for categories, paving the way for concise and computer-verifiable reasoning about the mathematics in a category \cite{curien14,cohen:2017,fiore:2022,DBLP:journals/lmcs/GratzerKNB21,north:2019}.
In summary, type theory can distill the salient features of semantic structures, and categorical structure can help in the study and development of type theories.

In this paper, we extend the synergy between type theory and category theory by giving a syntax for comprehension categories~\cite{jacobs93,jacobs99} -- one of the most common categorical structures for interpreting Martin-Löf type theory~\cite{mltt}. This syntax adds structure to Martin-Löf type theory.

Furthermore, we give a full analysis of the categorical structure that underlies subtyping, and we recognize this structure in many intensional models of Martin-Löf type theory.

\subsection{Interpretation of Type  Theories in Comprehension Categories}

We sketch the interpretation of type theories in categorical structures, contrasting that of MLTT with that of other type theories.

\subsubsection{Interpretation of MLTT}

Martin-Löf type theory (MLTT), and variations thereof, are typically given semantics in categories with structure, such as categories with families~\cite{dybjer96}, contextual categories~\cite{cartmell:1978,Cartmell86}, and type categories~\cite{pitts:2000}.
All of these categorical structures can be viewed as a \emph{full (or discrete) comprehension category} \cite{DBLP:conf/aplas/AhrensLN24}. 

A comprehension category consists (among other things) of two categories: one whose objects interpret the contexts of the type theory and one which interprets its types. Thus, at first glance, a comprehension category can express both morphisms between contexts and morphisms between types. In MLTT, however, there is only one such notion in the sense that the morphisms between both contexts and types are generated by the terms of the theory. Hence, morphisms between types can in turn be recovered from the context morphisms. The requirements that a comprehension category is discrete or that it is full then represent the two universal ways of deriving the type morphisms from the context morphisms: discreteness assumes that there are as few type morphisms as possible,
whereas fullness assumes there are as many type morphisms as possible.
Thus, both the requirements of discreteness and fullness are used to `kill off' this `extra dimension' of morphisms.

\subsubsection{Interpretation of Intensional Type Theories}

To give semantics to intensional type theories,
such as homotopy type theory~\cite{hottbook} and cubical type theories~\cite{cohen:2017},
one generally uses more refined ideas, in particular from higher category theory and homotopy theory.
Specifically, such type theories are often given semantics in an \emph{algebraic weak factorization system} (or \emph{\AWFS} for short)~\cite{gambino:2023,hofmann:1998}.
These \AWFSs are usually interesting in their own right,
and they endow the category with higher dimensional structure
that we can understand synthetically via the identity type.
Just like other categorical structures such as categories with families, \AWFSs give rise to comprehension categories.
However, these comprehension categories are typically neither full nor discrete, nor are they split\footnote{A comprehension category is split if the chosen lifts of identities are identities and the chosen lift of any composite is the composition of the individual lifts (see \cref{def:split}).}.
In other words, these comprehension categories have morphisms between types that do \emph{not} arise from context morphisms, and their substitution structure is more intricate than that of Martin-Löf type theory itself. We argue that this extra structure in the semantics should not be ignored, and it is captured by the type theory presented here.


\subsection{New Type Theories: Reflecting Semantic Features Into the Syntax}

We develop a type theory that allows for synthetic reasoning about the comprehension categories arising from \AWFSs --- thus, in particular, about morphisms of types and about non-split substitution.
To this end, we develop a new type theory, which we call \emph{comprehension category type theory} or \CTT for short.
This type theory is obtained by ``reflecting'' some semantic features back into traditional type theory.

We are not the first to develop type theory by reflecting semantic features back into the syntax.
For instance, a type theory for non-split substitution has been developed by \citet{curien14}, who developed a type theory with explicit substitutions to reflect the non-split substitution structure of comprehension categories into type theory.
Similarly, \citet{coraglia_et_al:LIPIcs.TYPES.2023.3} studied generalized categories with families for the semantics of type theory; they discovered that, syntactically, such generalized structures give rise to a notion of coercive subtyping.
In our work, we encounter both of these syntactic features, and, in particular, give a systematic and extensive account of the syntax and semantics of subtyping sketched by \citet{coraglia_et_al:LIPIcs.TYPES.2023.3}.




Specifically, by not imposing the restrictions of discreteness or fullness, our syntax and semantics can capture coercive subtyping \cite{coraglia_et_al:LIPIcs.TYPES.2023.3}. As argued there, every comprehension category is equipped with a notion of subtyping,
which is given by morphisms between types.
Fullness expresses that coercions and terms are the same, making the subtyping, in one sense, trivial.
Discreteness expresses that all coercions arise from identities of types,
making the subtyping, in another sense, trivial.
To faithfully model coercive subtyping,
one thus need to use comprehension categories that are neither full nor discrete.

Moreover, from a practical point of view, having both terms and type morphisms in a type theory allows for tight control over definitional equalities.
Specifically, when considering semantics of type theory in \AWFSs, types and type morphisms are interpreted as algebras for some monad, where the algebra structure models the transport, $(a = a') \to B(a) \to B(a')$, of structure along a term of an identity type, and morphisms of algebras preserve such transport \emph{strictly}, up to definitional equality; see also \cref{sec:awfs}.
Semantics in \AWFSs justify adding syntactic rules expressing such definitional equalities, which only hold up to \emph{propositional} identity in MLTT.

\subsection{Contributions and Synopsis}

The present paper is organized as follows.
\begin{enumerate}
\item In \cref{sec:from-comp-cat-judgements-rules}, we design judgements and structural rules for \CTT, our language for comprehension categories. We prove soundness for our rules with respect to comprehension categories.
  We show how our rules can be interpreted as the rules for a type theory with subtyping, extending a sketch by Coraglia and Emmenegger \cite{coraglia_et_al:LIPIcs.TYPES.2023.3}.
\item In \cref{sec:type-formers}, we design rules for type formers --- dependent pairs, dependent functions, and identity types --- for our structural rules. We prove soundness for these rules in suitably structured comprehension categories. We show our type formers can be interpreted with respect to subtyping, again extending a sketch by Coraglia and Emmenegger \cite{coraglia_et_al:LIPIcs.TYPES.2023.3}.
\item In \cref{sec:split}, we discuss a variant of \CTT with strictly functorial substitution that integrates \emph{splitness}. This syntax is easier to work with, at the cost of having fewer models.
\item In \cref{sec:rel-work}, we discuss related work.
\end{enumerate}

\section{Review: Comprehension Categories from Algebraic Weak Factorization Systems}
\label{sec:review}

In this section we briefly review the categorical notions used in the remainder of the paper.

\subsection{Fibrations}
\label{subsec:fibrations}

\emph{Fibrations} can model dependent types in contexts, and substitutions.
Intuitively,
a fibration consists of a category $\CC$ of contexts and context morphisms,
a category $\DD$ of types depending on contexts,
and a substitution operation on types.
As a reference on fibrations, see the notes by \citet{streicher18}.


\begin{definition} \label{definition-cartesian-morphism}
  Let $p : \DD \to \CC$ be a functor. A morphism $\varphi : Y \to X$ in $\DD$ is called \emphasize{cartesian} if and only if for all $v : \Theta \to \Delta$ in $\CC$ and $\theta : Z \to X$ with $p(\theta) = p (\varphi) \circ v$ there exists a unique morphism $\psi : Z \to Y$ with $p(\psi) = v$ and $\theta = \varphi \circ \psi$.
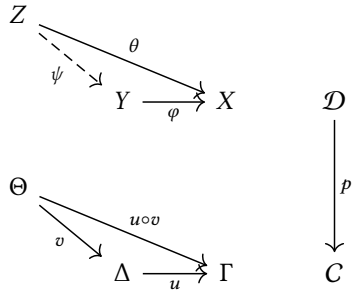
  \[\begin{tikzcd}[row sep=0.5em]
  Z \\
  & Y & X & \DD \\
  \Theta \\
  & \Delta & \Gamma & \CC
  \arrow["\psi"', dashed, from=1-1, to=2-2]
  \arrow["\theta", from=1-1, to=2-3,bend left=10]
  \arrow["\varphi"', from=2-2, to=2-3]
  \arrow["p", from=2-4, to=4-4]
  \arrow["v"', from=3-1, to=4-2]
  \arrow["{u \circ v}", from=3-1, to=4-3,bend left=10]
  \arrow["u"', from=4-2, to=4-3]
  \end{tikzcd}\]
  A morphism $\alpha$ is called \emphasize{vertical} if and only if $p(\alpha)$ is an identity morphism in $\CC$. For $\Gamma \in \CC$, we write $\DD_\Gamma$ for the subcategory of $\DD$ consisting of objects $X$ with $p (X) = \Gamma$ and morphisms $\alpha$ with $p(\alpha) = \id_\Gamma$. The category $\DD_\Gamma$ is called the \emphasize{fiber of $p$ over $\Gamma$}.
\end{definition}

\begin{definition} \label{definition-fibration}
  A functor $p : \DD \to \CC$ is called a \emphasize{(cloven) Grothendieck fibration} if and only if for all $u : \Delta \to \Gamma$ in $\CC$ and $X \in \DD_\Gamma$ we have a chosen cartesian arrow $\varphi : Y \to X$ with $p (\varphi) = u$ called a \emphasize{cartesian lifting} of $u$ to $X$.

\end{definition}

The adjective ``cloven'' in \cref{definition-fibration} refers to the fact that we assume the cartesian lifts to be \emph{chosen} rather than merely \emph{existing}, as in some other sources.
 Throughout this paper, all fibrations will be cloven.
For brevity, we omit both ``cloven'' and ``Grothendieck'' when referring to \emph{fibrations}.

\begin{definition} \label{def-split-fibration}
  \label{def:split}
  A fibration $p : \DD \to \CC$ is \emphasize{split} if
      the chosen lifts of identities are identities; and
      the chosen lift of any composite is the composite of the individual lifts.
\end{definition}

\begin{definition}
 Given a category $\CC$, its \emphasize{arrow category} $\CC^{\to}$ has, as objects, morphisms of $\CC$.
 A morphism in $\CC^{\to}$ from $f : a \to b$ to $g : c \to d$ is given by a pair $(k,l)$ of morphisms in $\CC$, with $k : a \to c$ and $l : b \to d$, such that $l \circ f = g \circ k$.
\end{definition}

\begin{example}
 Let $\CC$ be a category.
 \begin{itemize}
  \item The forgetful functor $\dom : \CC^{\to} \to \CC$ sending a morphism to its domain is a fibration.
  \item The forgetful functor $\cod : \CC^{\to} \to \CC$ sending a morphism to its codomain is a fibration if and only if the category $\CC$ has chosen pullbacks.
  \end{itemize}
\end{example}

\subsection{Comprehension Categories}

In this section we recall the definition of comprehension categories.

\begin{definition}[{\cite[Definition~4.1]{jacobs93}}] \label{definition-comprehension-category} \label{definition-full-split-comprehension-category}
A \emphasize{comprehension category} consists of a category $\CC$, a fibration $p : \TT \to \CC$, and a functor $\chi : \TT \to \CC^\to$ preserving cartesian arrows, such that the following diagram commutes.
\[
\begin{tikzcd}
\TT \arrow[rr, "\chi"] \arrow[dr, "p" swap] && \CC^{\to} \arrow[dl, "\cod"] \\
& \CC
\end{tikzcd}
\]
Here, $\chi$ is called the \emphasize{comprehension} and $\dom \circ \chi$ is denoted by $\chi_0$.

A comprehension category is called \emphasize{full} if $\chi: \TT \to \CC^{\to}$ is fully faithful and is called \emphasize{split} if $p : \TT \to \CC$ is a split fibration.     
\end{definition}

\citet{jacobs93,jacobs99} uses full split comprehension categories to model type dependency.
In the following example, we build intuition about the interpretation of dependent type theories in comprehension categories.
\begin{example}[{\cite[Section~10.3]{jacobs99}},{\cite[Example~2.1.2]{lumsdaine2015}}]
  \label{exa:mltt-term-model}
In the syntactic comprehension category built from MLTT, the base category $\CC$ is the category of contexts where objects are contexts and morphisms are substitutions between contexts.
The fiber $\TT_\Gamma$ over a context $\Gamma$ is the category of types in context $\Gamma$.
The morphisms in the fibers are such that $\chi$ is fully faithful, i.e. the comprehension category is \emph{full}.
The cloven fibration $p : \TT \to \CC$ maps each type to its context and the chosen lifts are determined such that reindexing gives type substitution.
In particular, for a context morphism (substitution) $s : \Gamma \to \Delta$ and $ A \in \TT_\Delta$ (type $A$ in context $\Delta$), $s^* A$ is substitution by $s$ in $A$.
Since substitution in MLTT is strictly functorial, the resulting comprehension category is \emph{split}.
The functor $\chi$ maps $ A \in \TT_\Delta$ (type $A$ in context $\Delta$) to the projection from an extended context to the original context $\pi_A : \Gamma.A \to \Gamma$.
In this example, the sections of such projections are the terms of type $A$ in context $\Gamma$.
\end{example}

\subsection{Comprehension Categories from Algebraic Weak Factorization Systems}
\label{sec:awfs}

As discussed in \cref{exa:mltt-term-model}, in a comprehension category,
we have a notion of morphism between types,
namely morphisms in $\TT_\Gamma$,
and we have a notion of terms,
which are sections of projections $\chi(A) : \Gamma.A \to \Gamma$.
Often, the comprehension $\chi$ is fully faithful, and thus (up to equivalence) the inclusion of a full subcategory.
In this case, the objects of $\TT$ can be seen as objects of $\CC^{\to}$, i.e., morphisms of $\CC$, with some property.
This is also the case for \cref{exa:mltt-term-model}. 
One might want, however, to consider a category of morphisms of $\CC$, each equipped with some kind of \emph{structure}, not just property.

That is, one could consider a monad $R$ on $\CC^{\to}$ and take its Eilenberg-Moore category $\EM(R)$ together with the forgetful functor $U: \EM(R) \to \CC^{\to}$.
\[
\begin{tikzcd}
\EM(R) \arrow[rr, "U"] \arrow[dr, "\cod" swap] && \CC^{\to} \arrow[dl, "\cod"] \\
& \CC
\end{tikzcd}
\]
The composition $\EM(R) \xrightarrow{U} \CC^{\to} \xrightarrow{\cod} \CC$ (also denoted by $\cod$ above) is a fibration if $\CC$ has all pullbacks (so that $\CC^{\to} \xrightarrow{\cod} \CC$ is a fibration) and $U$ is a \emph{discrete pullback-fibration} \cite{bourke:2016}. The functor $U$ is a \emph{discrete pullback-fibration} \cite{bourke:2016} if the pullback of the underlying morphism of an $R$-algebra can itself uniquely be given the structure of an $R$-algebra making the pullback square an $R$-algebra morphism. \citet[Prop.~8]{bourke:2016} show that $U$ is a discrete pullback-fibration if and only if $R$ is isomorphic to a monad over $\cod$, meaning that $R$ lifts to a monad on $\CC^{\to} \xrightarrow{\cod} \CC$ in the slice $\Cat / \CC$.

Given such a monad, the components of its unit assemble into an endofunctor on $\CC^{\to}$ with a counit -- in fact it is a copointed endofunctor on $\CC^{\to}$ over $\dom: \CC^{\to} \to \CC$.

There are many examples of such monads coming from an \emph{algebraic weak factorization system} (\AWFS)\footnote{Note that not all non-full comprehension categories arise from \AWFSs. In particular, comprehension categories where $\chi$ is not faithful do not arise from \AWFSs as described above.}.
In brief, an \AWFS~\cite{grandis:2006,garner:2009b} is a monad on $\CC^{\to}$ over $\cod : \CC^{\to} \to \CC$ together with structure making the associated copointed endofunctor on $\CC^{\to}$ a comonad (and sometimes a distributive law of the comonad over the monad).

\AWFSs are fundamental when constructively studying the semantics of intensional type theory. They have been applied to construct models of both homotopy and cubical type theory \cite{AWODEY20181270,swan2018identity,cavallo2020unifying,EMMENEGGER2022103103,gambino:2023,van2022effective,awodey2024equivariant}.
Very roughly, the algebras of the constituent monad of one of these \AWFSs are objects with \emph{transport} and associated structure, in the terminology of homotopy type theory.
Morphisms of algebras, which are morphisms of types in the associated comprehension category, are morphisms that preserve the transport \emph{strictly}, commuting with the transport up to definitional equality.

\begin{example}[{\cite[Ex.~29]{bourke:2016}},~cf.~\cite{hofmann:1998,north:2019}]
\label{exa:grpd}
There is a monad on the category of groupoids whose algebras are groupoids $\Gamma$ together with a split (iso)fibration $T : \EE \rightarrow \Gamma$.
Morphisms from $(\Gamma_1 , T_1)$ to $(\Gamma_2 , T_2)$
consist of functors $F : \Gamma_1 \rightarrow \Gamma_2$ and $G : \EE_1 \rightarrow \EE_2$
such that $T_2 \circ F = G \circ T_1$
and such that $G$ preserves the chosen lifts up to equality.
The comprehension $\chi$ is given by taking the underlying functor.
In the resulting comprehension category, $\CC$ is the category of groupoids, the objects in the fiber over a groupoid $\Gamma$ are split (iso)fibrations, and $\chi$ is an inclusion. 
This comprehension category is not full
because functors do not preserve chosen lifts in general.

We can restrict the \AWFS on categories to one on the full subcategory of groupoids.
Again, the resulting comprehension category is not full.
\end{example}

The groupoid model is fundamental in the semantics of type theory,
because it refutes the uniqueness of identity proofs~\cite{hofmann:1998}.
Currently, the groupoid model is being formalized in Lean~\cite{hua2025hottlean}.



Next we look at a different class of examples,
where types are interpreted as formulas.
These examples are in nature closer to refinement types than to dependent types.
Whereas dependent types generally are used in proof relevant settings,
refinement types are used in proof irrelevant settings.
In our first example, we look at formulas as subobjects.

\begin{example}
\label{exa:topos-subobj}
Let $\EE$ be a topos with subobject classifier $\Omega$. 
Every topos can be equipped with an orthogonal factorization system
whose left class is given by the epimorphisms
and whose right class is given by the monomorphisms.
Since every orthogonal factorization system also is an algebraic factorization system,
we get  a comprehension category where $\CC$ is $\EE$
and where objects of $\TT$ are given by an object $x \in \EE$ with a morphism $p : x \rightarrow \Omega$.
The functor $\chi$ sends a morphism $p : x \rightarrow \Omega$ to the subobject classified by $p$.
The resulting comprehension category is not full in general \cite{coraglia:2024,jacobs99}.
\end{example}

Next we specialize \Cref{exa:topos-subobj} to get a comprehension category whose contexts are sets
and whose types are predicates valued in a given Heyting algebra.

\begin{example}
\label{exa:cha-pred}
Let $H$ be a Heyting algebra.
Note that we have a topos $\EE$ of sheaves over $H$,
which can equivalently be described as partial equivalence relations valued in $H$ \cite{higgs:1984}.
We have a fully faithful functor $F$ from $\Set$ to $\EE$ sending every set $X$ to the partial equivalence relation given by equality.
Objects in the image of $F$ are called discrete.
From this topos we obtain a comprehension category,
which we restrict  to the contexts that are discrete.
In the resulting comprehension category,
$\CC$ is the category of sets
and objects of $\TT$ consist of a set $\Gamma$ together with a map $p : \Gamma \rightarrow H$.
Morphisms in $\TT$ from $(\Gamma_1 , p_1)$ to $(\Gamma_2 , p_2)$ are given by
functions $f : \Gamma_1 \rightarrow \Gamma_2$ such that for each $x \in \Gamma_1$ we have $p_1(x) \le p_2(f(x))$.
The functor $\chi$ maps every $(\Gamma , p)$ to the set $\{ x \in \Gamma \mid \top \le p(x) \}$.
This comprehension category is not full in general.
In fact, if $H$ is the collection of open subsets of some topological space $X$,
then this comprehension category only is full if $X$ is indiscrete.
\end{example}



\section{Syntax from Comprehension Categories} \label{sec:from-comp-cat-judgements-rules}

We present the judgements and structural rules of a type theory with explicit substitution corresponding to a comprehension category. To be clear, we do not impose any requirements of fullness or splitness on the comprehension category. 

\subsection{Judgements} \label{sec:from-comp-cat-judgements}
The judgements of the type theory are as follows:
\begin{enumerate}
    \item $\Gamma \ctx$, which is read as `$\Gamma$ is a context'; \label{judge-one}
    \item $\Gamma \vdash  s : \Delta$, which is read as `$s$ is a substitution from $\Gamma$ to $\Delta$', where $\Gamma, \Delta \ctx$; \label{judge-two}
    \item $\Gamma \vdash s \equiv s' : \Delta$, which is read as `$s$ is equal to $s'$', where $\Gamma \vdash s, s' : \Delta$;
    \item $\Gamma \vdash A \type$, which is read as `$A$ is a type in context $\Gamma$', where $\Gamma \ctx$; \label{judge-four}
    \item $\Gamma ~|~ A \vdash t : B$, which is read as `$t$ is a type morphism from $A$ to $B$ in context $\Gamma$', where $\Gamma \vdash A , B \type$; \label{judge-five}
    \item $\Gamma ~|~ A \vdash t \equiv t' : B $, which reads as `$t$ and $t'$ are equal', where $\Gamma ~|~ A \vdash t, t' : B$.
\end{enumerate}

The judgement $\Gamma \vdash  s : \Delta$ can also be read as `$s$ is a context morphism from $\Gamma$ to $\Delta$', since this judgement is interpreted as a morphism $\interp{s} : \interp{\Gamma} \to \interp{\Delta}$  in the category of contexts $\CC$ (see \cref{sec:from-comp-cat-soundness}). Similarly, the judgement $\Gamma ~|~ A \vdash t : B$ is read as `$t$ is a type morphism from $A$ to $B$', since $t$ is interpreted as a morphism from $\interp{A}$ to $\interp{B}$ in $\TT_{\interp{\Gamma}}$ (see \cref{sec:from-comp-cat-soundness}). 

Unlike Martin-Löf type theory, this type theory has explicit substitution in the sense of \citet{explicit-substitution}. Another difference between this type theory and Martin-Löf type theory is the terms. As we will see in \cref{sec:from-comp-cat-soundness}, the terms of this type theory, which correspond to judgements of the form $\Gamma ~|~ A \vdash t :  B$, are interpreted as morphisms in $\TT$. The terms of Martin-Löf type theory, however, are interpreted as certain morphisms in $\CC$ --- in particular, as sections of the projection context morphisms. We will refer to these sections as \emphasize{MLTT terms}. 
In \cref{notation:terms}, we define a notation for MLTT terms. 

\begin{relatedwork}[\citet{ANW23}] \label{relwork:judgements}
Judgements \ref{judge-one}, \ref{judge-two}, \ref{judge-four} and \ref{judge-five} are the same as the corresponding judgements in the type theory for comprehension \emphasize{bi}categories in the work of \citet{ANW23}. They read Judgement \ref{judge-five} as `$t$ is a term of type $B$ depending on $A$ in context $\Gamma$'. 
In the present work, we read Judgement \ref{judge-five} as `$t$ is a type morphism from $A$ to $B$ in context $\Gamma$' for generality.
\end{relatedwork}

\subsection{Rules for Context and Type Morphisms} \label{sec:from-comp-cat-c-and-t}

We want contexts and substitutions (context morphisms) to form a category --- the category $\CC$ in a comprehension category $(\CC, \TT, p, \chi)$. Hence, the rules of the type theory concerning substitution  follow the usual axioms for a category.

\[
\unaryRule{\Gamma \ctx}{\Gamma \vdash 1_\Gamma: \Gamma}{ctx-mor-id}
\binaryRule
{\Gamma \vdash s : \Delta}{\Delta \vdash s' : \Theta}{\Gamma \vdash s' \circ s : \Theta}{ctx-mor-comp}
\]
\[
\unaryRuleBinaryConcl{\Gamma \vdash s : \Delta}{\Gamma \vdash s \circ 1_\Gamma \equiv s: \Delta}
{\Gamma \vdash 1_\Delta \circ s \equiv s : \Delta}
{ctx-id-unit}
\ternaryRule{\Gamma \vdash s : \Delta}{\Delta \vdash s' : \Theta}{\Theta \vdash s'' : \Phi}{\Gamma \vdash s'' \circ (s' \circ s) \equiv (s'' \circ s') \circ s : \Phi}{ctx-comp-assoc}
\]

Similarly, we want types and type morphisms to form a category --- the category $\TT$ in a comprehension category $(\CC, \TT, p, \chi)$. Hence, we postulate the following rules, which follow the usual axioms for a category.
\[
\unaryRule
{\Gamma \vdash A \type}{\Gamma ~|~ A \vdash 1_A: A}{ty-mor-id}
\binaryRule
{\Gamma ~|~ A \vdash t : B}{\Gamma ~|~B \vdash t' : C}{\Gamma ~|~A  \vdash t' \circ t : C}{ty-mor-comp}
\]
\[
\unaryRuleBinaryConcl
{\Gamma ~|~ A \vdash t : B}{\Gamma ~|~ A \vdash t \circ 1_A \equiv t: B}
{\Gamma ~|~ A \vdash 1_B \circ t \equiv t: B}{ty-id-unit}
\ternaryRule
{\Gamma ~|~ A \vdash t : B}{\Gamma ~|~ B \vdash t' : C}{\Gamma ~|~ C \vdash t'' : D}{\Gamma ~|~ A \vdash t'' \circ (t' \circ t) \equiv (t'' \circ t') \circ t : D}{ty-comp-assoc}
\]

\begin{notation} \label{notation:entiso}
  Similar to \citet[Section~8]{ANW23}, we define the following notations, which each stand for four judgements.
  \vspace{-\baselineskip}
  \begin{multicols}{2}
  \begin{enumerate}
    \item $\Gamma \vdashiso s : \Delta$ stands for the following four judgements.
  \begin{itemize}
      \item $\Gamma \vdash s : \Delta$
      \item $\Delta \vdash s' : \Gamma$
      \item $\Delta \vdash s \circ s' \equiv 1_\Delta: \Delta$
      \item $\Gamma \vdash s' \circ s \equiv 1_\Gamma: \Gamma$
  \end{itemize}
  \end{enumerate}
  \columnbreak
  \begin{enumerate}[resume]
  \item $\Gamma ~|~ A \vdashiso t : B$ stands for the following four judgements.
  \begin{itemize}
      \item $\Gamma ~|~A \vdash t : B$
      \item $\Gamma ~|~ B \vdash t' : A$
      \item $\Gamma ~|~ B \vdash t \circ t' \equiv 1_B: B$
      \item $\Gamma ~|~ A \vdash t' \circ t \equiv 1_A: A$
  \end{itemize} 
  \end{enumerate} 
  \end{multicols}
  \vspace{-\baselineskip}
Given $\Gamma \vdashiso s : \Delta$, we write $\Delta \vdash s^{-1} : \Gamma$ for the inverse of $s$. Similarly, $t^{-1}$ denotes the inverse of $t$.
\end{notation}

\begin{remark} \label{rem:empty-ctx}
We can postulate an empty context and a unique substitution from each context to it, which corresponds to having a terminal object in the category $\CC$. For this, the following rules can be added to the type theory. We do not discuss this further in this paper.
\[
\unaryRule{}{\emptyctx \ctx}{empty-ctx}
\unaryRule{\Gamma \ctx}{\Gamma \vdash \langle\rangle_\Gamma : \emptyctx}{empty-ctx-mor}
\unaryRule
{\Gamma \vdash s : \emptyctx}{\Gamma \vdash s \equiv \langle\rangle_\Gamma : \emptyctx}{empty-ctx-mor-unique}
\]

\end{remark}

\subsection{Rules for Context Extension} \label{sec:from-comp-cat-ext}
The rules for context extension mirror the action of the comprehension functor $\chi : \TT \to \CC^{\to}$ in a comprehension category $(\CC,\TT, p, \chi)$. Particularly, the rules correspond to the functoriality of $\chi$ restricted to $\TT_\Gamma \to \CC/\Gamma$ for each $\Gamma \in \CC$.
\[
\unaryRule
{\Gamma \vdash A \type}{\Gamma.A \ctx}{ext-ty}
\unaryRule
{\Gamma ~|~ A \vdash t : B}{\Gamma . A \vdash \Gamma . t : \Gamma.B }{ext-tm}
\unaryRule
{\Gamma \vdash A \type}{\Gamma . A \vdash \Gamma . 1_A \equiv 1_{\Gamma.A} : \Gamma. A}{ext-id}
\]
\[
\binaryRule
{\Gamma ~|~ A \vdash t : B}{\Gamma ~|~ B \vdash t' : C}{\Gamma . A \vdash \Gamma . (t' \circ t) \equiv \Gamma. t' \circ \Gamma . t : \Gamma . B}{ext-comp}
\unaryRule
{\Gamma \vdash A \type}{\Gamma.A \vdash \pi_A : \Gamma}{ext-proj}
\unaryRule
{\Gamma ~|~ A \vdash t : B}{\Gamma . A \vdash \pi_B \circ \Gamma . t \equiv \pi_A : \Gamma}{ext-c}
\]



Before moving on to the rules for substitution, we introduce the following notation for MLTT terms.
\begin{notation}
  \label{notation:terms}
  We define the notation $\Gamma \vdash a : A$ for MLTT terms, which stands for two judgements $\Gamma \vdash a : \Gamma.A$ and $\Gamma \vdash \pi_A \circ a \equiv 1_\Gamma : \Gamma$.
\end{notation}

\subsection{Rules for Substitution} \label{sec:from-comp-cat-subst}
Unlike standard Martin-Löf type theory, this type theory features explicit substitution in the syntax. By the Grothendieck construction, we know that the fibration $p : \TT \to \CC$, in a comprehension category $(\CC, \TT, p, \chi)$, can be thought of as a pseudofunctor of the form $\CC^{\mathsf{op}} \to \mathsf{Cat}$. The rules regarding substitution will be interpreted by this pseudofunctor. Specifically, the fibration $p : \TT \to \CC$ gives rise to reindexing functors of the form $s^* : \TT_\Delta \to \TT_\Gamma$ for each $s : \Gamma \to \Delta$ in $\CC$, and two natural isomorphisms corresponding to composition of reindexing functors and reindexing along identity morphisms. Namely, for each $s : \Gamma \to \Delta$ and $s' : \Delta \to \Theta$ in $\CC$, we have a natural isomorphism $i^{\mathsf{comp}} : (s' \circ s)^* \iso s^*  \circ s'^* $, and for each $A \in \TT_\Gamma$ we have an isomorphism $\idiso{A} : 1^*_\Gamma A \iso A$.

The rules for substitution are as follows. 
\[
\binaryRule
{\Gamma \vdash s : \Delta}{\Delta \vdash A \type}{\Gamma \vdash A[s] \type}{sub-ty}
\binaryRule
{\Gamma \vdash s : \Delta}{\Delta ~|~ A \vdash t: B}{\Gamma ~|~ A[s] \vdash t[s] : B[s]}{sub-tm}
\]
\[
\binaryRule
{\Gamma \vdash s : \Delta}{\Delta \vdash A \type}{\Gamma ~|~ A[s] \vdash 1_A[s] \equiv 1_{A[s]}: A[s]}{sub-prs-id}
\ternaryRule
{\Gamma \vdash s : \Delta}
{\Delta ~|~A \vdash t : B}{\Delta ~|~ B \vdash t' : C}{\Gamma ~|~ A[s] \vdash (t' \circ t)[s] \equiv t'[s] \circ t[s] : C[s]}{sub-prs-comp}
\]
\[
\unaryRule
{\Gamma \vdash A \type}{\Gamma ~|~ A[1_\Gamma] \vdashiso \idiso{A} : A}{sub-id}
\ternaryRule
{\Gamma \vdash s : \Delta}{\Delta \vdash s' : \Theta}{\Theta \vdash A \type}{\Gamma ~|~ A[s' \circ s] \vdashiso \compiso{A}{s'}{s}  : A[s'][s]}{sub-comp}
\]
\[
\unaryRule
{\Gamma ~|~ A \vdash t: B}{\Gamma ~|~ A[1_\Gamma] \vdash t[1_\Gamma] \equiv \idisoinv{B} \circ t \circ \idiso{A} : B[1_\Gamma]}{sub-tm-id}
\]
\[
\ternaryRule
{\Gamma \vdash s : \Delta}{\Delta \vdash s' : \Theta}{\Theta ~|~ A \vdash t: B}
{\Gamma ~|~ A[s' \circ s] \vdash t[s' \circ s] \equiv \compisoinv{B}{s'}{s} \circ t[s'][s] \circ \compiso{A}{s'}{s} : B[s' \circ s] }{sub-tm-comp}
\]

\begin{remark}\label{rem:split-subst}
 In our rules for substitution, composition and identity of context morphisms are only preserved up to isomorphism rather than up to equality. The reason behind this choice is that in many comprehension categories, in particular those arising from \AWFSs, substitution laws only hold up to isomorphism. To guarantee that our syntax can be interpreted in such comprehension categories, it is necessary to relax the substitution laws.
 In \cref{sec:split} we present a split variant of our syntax, where we consider these rules up to equality.
\end{remark}

In a comprehension category $(\CC, \TT, p, \chi)$, the comprehension of each cartesian lift $s_A : s^* A \to A$ is a pullback square in $\CC$, since $\chi$ preserves cartesian morphisms. This means that for each morphism $s : \Gamma \to \Delta$ in $\CC$, morphisms $s' : \Gamma \to (\chi_0 A)$ in $\CC$ such that $\chi(A) \circ s' = s$ correspond to sections $t$ of $\chi (s^* A)$ such that $\chi_0 (s_A) \circ t = s'$. 
\[\begin{tikzcd}
   [column sep = large, row sep = normal]
	\Gamma \\
	& {\chi_0 (s^* A)} & {\chi_0(A)} \\
	& \Gamma & \Delta
	\arrow["t", dashed, from=1-1, to=2-2]
	\arrow["{s'}", curve={height=-12pt}, from=1-1, to=2-3]
	\arrow[equal, curve={height=12pt}, from=1-1, to=3-2]
	\arrow["{\chi_0 (s_A)}", from=2-2, to=2-3]
	\arrow["{\chi (s^*A)}" description, from=2-2, to=3-2]
	\arrow["\lrcorner"{anchor=center, pos=0.125}, draw=none, from=2-2, to=3-3]
	\arrow["{\chi (A)}" description, from=2-3, to=3-3]
	\arrow["s"', from=3-2, to=3-3]
\end{tikzcd}\]
The following rules capture the idea that context morphisms are built out of MLTT style terms, that is sections of projection context morphisms $\Gamma.A \vdash \pi_A : \Gamma$.
\[
\ternaryRule
{\Delta \vdash A \type}
{\Gamma \vdash s : \Delta}
{\Gamma \vdash t : A[s]}
{\Gamma \vdash (s,t) : \Delta.A}{\scriptsize{sub-ext}}
\unaryRule
{\Gamma \vdash s : \Delta.A}
{\Gamma \vdash p_2(s) : A[\pi_A \circ s]}
{\scriptsize{sub-proj}}
\]
\[
\ternaryRuleBinaryConcl
{\Delta \vdash A \type}{\Gamma \vdash s : \Delta}
{\Gamma \vdash t : A[s]}
{\Gamma \vdash \pi_A \circ (s,t) \equiv s : \Delta}{\Gamma \vdash p_2(s,t) \equiv t : \Gamma.A[s]}{\scriptsize{sub-beta}}
\unaryRule
{\Gamma \vdash s : \Delta.A}{\Gamma \vdash (\pi_A \circ s ,p_2(s)) \equiv s : \Delta.A}{\scriptsize{sub-eta}}
\]
We also have rules for functoriality of $p_2(-)$ (Rules \ruleref{sub-proj-id} and \ruleref{sub-proj-comp} in \cref{sec:structural-rules}).

Before concluding the rules, we discuss the following derived rule which is frequently used in the rest of the paper.
\begin{proposition} \label{lemma:s.A}
From the rules in \cref{fig:from-comp-cat,fig:from-comp-cat-congruence}, we can derive the following rule. 
\[
\binaryRule
{\Delta \vdash A \type}{\Gamma \vdash s : \Delta}{\Gamma.A[s] \vdash s.A : \Delta.A}{\scriptsize{ctx-mor-lift}}
\]
\end{proposition}
\begin{proof}
  The context morphism $s.A$ is $(s \circ \pi_{A[s]}, \Gamma.(\compiso{A}{s}{\pi_{A[s]}} [\pi_{A[s]}]) \circ p_2(1_{\Gamma.A[s]}))$ and Rule \ruleref{ctx-mor-lift} can be derived using Rules \ruleref{ctx-mor-comp}, \ruleref{ext-proj}, \ruleref{sub-comp}, \ruleref{sub-tm}, \ruleref{ext-tm}, \ruleref{sub-proj} and \ruleref{sub-ext}.
\end{proof}

Lastly, we have the following rule which characterizes the behavior of $\Gamma.t[s]$.
In this rule, we use the notation introduced in \cref{lemma:s.A}.
\[
\binaryRule
{\Delta ~|~ A \vdash t : B}{\Gamma \vdash s : \Delta}
{\Gamma.A[s] \vdash s.B \circ \Gamma.t[s] \equiv \Delta.t \circ s.A : \Delta.B}{\scriptsize{tm-sub-coh}}
\]

\LongOrShort{
Rules \ruleref{sub-ty} and \ruleref{sub-tm} are derived from the action of the reindexing functor $s^* : \TT_\Delta \to \TT_\Gamma$ on objects and morphisms respectively. Rules \ruleref{sub-prs-id} and \ruleref{sub-prs-comp} are derived from the functoriality of the reindexing functor. Rule \ruleref{sub-id} is derived from the isomorphism $\idiso{A} : 1^*_\Gamma A \iso A$ and Rule \ruleref{sub-comp} from the isomorphism $\compiso{A}{s}{s'} : (s' \circ s)^* A \iso s^* ( s'^* A)$.

Rules \ruleref{sub-ty} and \ruleref{sub-tm} are the usual substitution rules. Rules \ruleref{sub-id} and \ruleref{sub-comp} state that substitution is functorial up to isomorphism, unlike Martin-Löf type theory which has strictly functorial substitution as stated in \cref{lemma-substitution-comp-associative,lemma-substitution-id}. As discussed in \cref{sect-dtt-substitution-upto-iso}, this is in line with there being no assumption on the splitness of the comprehension category.
Rules \ruleref{sub-prs-id} and \ruleref{sub-prs-comp} state that identity and composition of type morphisms are preserved under substitution.

For each $\Gamma, \Delta \in \CC$, $A \in \TT_\Delta$ and $s : \Gamma \to \Delta$, the comprehension of the cartesian lift $s_A : s^* A \to A$ is a pullback square in $\CC$. Rules \ruleref{sub-ext}, \ruleref{sub-proj}, \ruleref{sub-beta} and \ruleref{sub-eta} are derived from the universal property of such pullback squares:
\[
\begin{tikzcd}
	{\Gamma.A[s]} & {\Delta.A} \\
	\Gamma & \Delta.
	\arrow["{s.A}", from=1-1, to=1-2]
	\arrow["{\pi_{A[s]}}"{description}, from=1-1, to=2-1]
	\arrow["\lrcorner"{anchor=center, pos=0.125}, draw=none, from=1-1, to=2-2]
	\arrow["{\pi_A}"{description}, from=1-2, to=2-2]
	\arrow["s"', from=2-1, to=2-2]
\end{tikzcd}
\]
}
{}
The rules of the type theory are summarized in \cref{fig:from-comp-cat} in \cref{sec:structural-rules}. In addition to these, we also have the rules related to $\equiv$ being a congruence for all the judgements, which are listed in \cref{fig:from-comp-cat-congruence} in \cref{sec:structural-rules}. 

\begin{remark} \label{rem:sub-cong}
  In a comprehension category $(\CC,\TT,p ,\chi)$, for each object $A$ in $\CC$ and equal morphisms $s$ and $s'$ in $\CC$, we have $s^* A = (s')^* A$. As we do not have a judgement for equality of types, particularly a judgement of the form $\Gamma \vdash A[s] \equiv A[s']$, we can not express this idea syntactically. Instead, we add Rule \ruleref{sub-cong} (see \cref{fig:from-comp-cat-congruence}) to the type theory.
  \[
  \binaryRule{\Delta \vdash A \type}{\Gamma \vdash s \equiv s' : \Delta}{\Gamma ~|~ A[s] \vdashiso \subiso{A}{s}{s'} : A[s']}{\scriptsize{sub-cong}}
  \]
  This rule states that given a type $A$ in context $\Delta$ and two equal context morphisms $s$ and $s'$ from $\Gamma$ to $\Delta$, types $A[s]$ and $A[s']$ are isomorphic in the sense that there are two context morphism $\subiso{A}{s}{s'} : A[s] \to A[s']$ and $\subiso{A}{s'}{s} : A[s'] \to A[s]$ and their compositions are equal to the identity context morphisms.

  The coherence rules regarding $\subiso{A}{s}{s'}$ are stated in \cref{fig:from-comp-cat-congruence} in \cref{sec:structural-rules}.
\end{remark}

\begin{definition} \label{def:CTT}
  We define $\emphCTT$ to be the judgements described in \cref{sec:from-comp-cat-judgements} together with the rules in \cref{fig:from-comp-cat,fig:from-comp-cat-congruence}.
\end{definition}

\begin{remark}
Note that our syntax does not contain any coherence rules
that express the commutativity of diagrams built out of $\compiso{A}{s_1}{s_2}$ and $\idiso{A}$.
While it is possible to add such coherence equations,
we refrain from doing so.
This approach is similar to the work by Curien, Garner, and Hofmann who also consider a non-split syntax for type theory~\cite{curien14}.
In their syntax,
there is an additional operation on terms,
which they call explicit coercion.
This operation is used to apply the morphisms $\compiso{A}{s_1}{s_2}$ and $\idiso{A}$ to terms.
Instead of adding coherence rules to syntax,
they show that every diagram built out of $\compiso{A}{s_1}{s_2}$ and $\idiso{A}$ commutes in every model,
and that whenever two terms are equal after removing all coercions,
they have the same denotation in every model.
We conjecture that their methods can be applied to our setting; hence, we do not dicsuss the coherence rules in this paper.
\end{remark}

\subsection{Soundness: Interpretation in a Comprehension Category} \label{sec:from-comp-cat-soundness}
We establish soundness of the rules of the type theory by giving an interpretation of the type theory in every comprehension category. Note that there is no assumption of fullness or splitness on the comprehension category. 

\begin{theorem}[Soundness of Structural Rules] \label{thm:soundness}
Every comprehension category models the rules of $\CTT$.
\end{theorem}

Let $(\CC, \TT, p, \chi)$ be a comprehension category. The judgements are interpreted as follows:
\begin{enumerate}
    \item $\Gamma \ctx$ is interpreted as an object $\interp{\Gamma}$ in $\CC$;
    \item $\Gamma \vdash  s : \Delta$ is interpreted as a morphism $\interp{s} : \interp{\Gamma} \to \interp{\Delta}$ in $\CC$;
    \item $\Gamma \vdash s \equiv s' : \Delta$ is interpreted as $\interp{s} = \interp{s'}$;
    \item $\Gamma \vdash A \type$ is interpreted as an object $\interp{A}$ in $\TT_\interp{\Gamma}$;
    \item $\Gamma ~|~ A \vdash t : B$ is interpreted as a morphism $\interp{t} : \interp{A} \to \interp{B}$ in $\TT_\interp{\Gamma}$;
    \item $\Gamma ~|~ A \vdash t \equiv t' : B $ is interpreted as $\interp{t} = \interp{t'}$.
\end{enumerate}
The interpretation of the rules is deferred to \cref{sec:interp-struc-rules}.

\subsection{Subtyping} \label{sec:subty-structural-rules} 

One can regard type morphisms as witnesses of coercive subtyping. Coraglia and Emmenegger explore this in \emphasize{generalized categories with families}, a structure equivalent to (not necessarily full) comprehension categories \cite{coraglia_et_al:LIPIcs.TYPES.2023.3}. In this view, a judgement $\Gamma ~|~ A \vdash t:  B$ can be seen as expressing that $t$ is a witness for $A$ being a subtype of $B$. In Coraglia and Emmenegger’s notation, this is expressed as $\Gamma \vdash A \leq_t B$. We focus exclusively on proof-relevant subtyping. A comparison between proof-relevant and proof-irrelevant subtyping is given by \citet[Section~2.3]{coraglia_et_al:LIPIcs.TYPES.2023.3}.

\begin{proposition}[Subsumption rule] \label{prop:subtyping}
  From the rules of $\CTT$, we can derive the following rule.
  \[
  \binaryRule
  {\Gamma ~|~ A \vdash t:  B}
  {\Gamma \vdash a : A}
  {\Gamma \vdash \Gamma.t \circ a : B}
  {}
  \]
\end{proposition}
\begin{proof}
  The rule can be derived using Rules \ruleref{ext-tm}, \ruleref{ctx-mor-comp} and \ruleref{ext-c}.
\end{proof}

\cref{prop:subtyping} corresponds to subsumption in coercive subtyping.
The rule states that if $A$ is a subtype of $B$, then a (MLTT) term of type $A$ can be coerced to a term of type $B$. 

\begin{proposition}[Coercions commute with substitution] \label{prop:coer-comm-subst}
  From the rules of $\CTT$, we can derive the following rule.
  \[
  \ternaryRule
  {\Delta ~|~ A \vdash t:  B}
  {\Delta \vdash a : A}
  {\Gamma \vdash s : \Delta }
  {\Gamma \vdash (\Delta.t \circ a)[s] \equiv \Gamma.(t[s]) \circ a[s] : \Gamma.B[s]}{}
  \]
  where $a[s] \coloneqq p_2(s \circ a)$.
\end{proposition}

\cref{prop:coer-comm-subst} expresses that substitution and coercion commute. In practice, it allows us to compute substitutions in terms that contain coercions.

\begin{table}[h]
  \caption{Meaning of the rules of $\CTT$ when the judgement $\Gamma ~|~ A \vdash t : B$ expresses subtyping and relation to the rules by \citet{coraglia_et_al:LIPIcs.TYPES.2023.3}.}
  \label{tab:subtyping}
  \begin{tabular}{l p{7cm} l}
  Rule of $\CTT$ & Meaning under Subtyping & Rule in \cite{coraglia_et_al:LIPIcs.TYPES.2023.3} \\
  \hline
  \ruleref{ty-mor-id} & Reflexivity of subtyping witnessed by $1_A$ & - \\
  \ruleref{ty-mor-comp} & $A \leq_f B$ and $B \leq_g C$ give $A \leq_{g \circ f} C$. & \textit{Trans} and \textit{Sbsm} \\
  \ruleref{ty-id-unit} & Each $1_A$ is an identity for witness composition. & - \\
  \ruleref{ty-comp-assoc} & Composition of witnesses is associative. & - \\
  \ruleref{ext-tm} & $A \leq_t B$ gives a context morphism $\Gamma.A \vdash \Gamma.t : \Gamma.B$. & - \\
  \ruleref{ext-id} & Context morphism $\Gamma.1_A$ is equal to the identity. & - \\
  \ruleref{ext-comp} & For witnesses $f$ and $g$, $\Gamma.g \circ f$ is equal to $\Gamma.f \circ \Gamma.g$. & - \\
  \ruleref{sub-tm} & Substitution preserves subtyping. & \textit{Wkn} and \textit{Sbst} \\
  \ruleref{sub-prs-id} & Substitution preserves the identity witness. & - \\
  \ruleref{sub-prs-comp} & Substitution preserves composition of witnesses. & - \\
  \end{tabular}
\end{table}

In \cref{tab:subtyping}, we discuss the meaning of the rules of $\CTT$ that involve a judgement of the form $\Gamma ~|~ A \vdash t : B$ from the subtyping perspective. We also show how these rules relate to the rules discussed by \citet{coraglia_et_al:LIPIcs.TYPES.2023.3}.

We now discuss the interpretation of $\CTT$ and the subtyping structure for some of the examples from \cref{sec:awfs}.
\begin{example}[\cref{exa:grpd} ctd.]
  \label{exa:interp-grpd}
  In this example, a context $\Gamma$ is interpreted as a groupoid $\interp{\Gamma}$.
  A type $A$ in context $\Gamma$ is interpreted as a split isofibration $\interp{A} : \interp{\Gamma.A} \to \interp{\Gamma}$. 
  Type morphisms from $A$ to $A'$ in context $\Gamma$, i.e. witnesses for a subtyping relation $A \leq A'$, are interpreted as morphisms of split fibrations of the form $\interp{A} \to \interp{A'}$ preserving the chosen lifts up to equality. 
  A context morphism $\Gamma \vdash s : \Delta$ is interpreted as a functor of the form $\interp{s} : \interp{\Gamma} \to \interp{\Delta}$. 
  Given $\Delta \vdash A \type$ and $\Gamma \vdash s : \Delta$, the type $A[s]$ in context $\Gamma$ is interpreted as the pullback of $\interp{A}$ along $\interp{s}$.
\end{example}

\begin{example}[\cref{exa:cha-pred} ctd.]
  \label{exa:interp-cha-pred}
  Let $H$ be a Heyting algebra.
  A context $\Gamma$ is interpreted as a set $\interp{\Gamma}$.
  A type $A$ in context $\Gamma$ is interpreted as an $H$-values predicate $\interp{A} : \interp{\Gamma} \to H$. 
  We have a type morphism from $A$ to $A'$ in context $\Gamma$ if $\interp{A}$ entails $\interp{A'}$, i.e. for all $x \in \Gamma$ we have $\interp{A} (x) \leq \interp{A'}(x)$. 
  The subtyping relation interpreted in this example is proof-irrelevant in the sense that each hom-set in the fibers has at most one element.
  Given $\Delta \vdash A \type$ and $\Gamma \vdash s : \Delta$, the type $A[s]$ in context $\Gamma$ is interpreted as $\interp{A} \circ \interp{s}$.
  Context extension $\Gamma.A$ is interpreted as comprehension of $\interp{A}$.
\end{example}

\section{Extending $\CTT$ with Type Formers}
\label{sec:type-formers}

In this section, we develop syntactic rules --- on top of \CTT --- and semantic structures for interpreting these rules --- on top of non-full comprehension categories---, for $\Pi$-, $\Sigma$-, and identity types, respectively.

In detail, we define the semantic structures for these type formers in \cref{def:pi-non-full,def:sigma-non-full,def:id-non-full}. 
We then discuss functoriality conditions on those structures that allow us to use type morphisms to interpret subtyping (\cref{def:comp-cat-with-pi,def:comp-cat-with-sigma,def:comp-cat-with-id}).
Subsequently, we extend $\CTT$ with functorial $\Pi$-, $\Sigma$- and identity types (\cref{def:syntax-with-pi-for-terms,def:syntax-with-sigma-for-terms,def:syntax-with-id-for-terms}).  
We also prove soundness by giving an interpretation of the rules in any comprehension category with suitable structure for each case (\cref{thm:soundness-pi-for-terms,thm:soundness-sigma-for-terms,thm:soundness-id-for-terms}). 
Finally, we briefly discuss how $\CTT$ with functorial $\Pi$-, $\Sigma$- and identity types supports subtyping (\cref{rem:subty-pi,rem:subty-sigma,rem:subt-id}).

Proofs of the soundness theorems (\cref{thm:soundness-pi-for-terms,thm:soundness-sigma-for-terms,thm:soundness-id-for-terms}) are in \cref{sec:interp-struc-rules}.



\begin{notation}
  \label{notation:global}
  In the remainder of this paper, we also use the following notations in comprehension categories to highlight how the syntax relates to the semantics.
  \begin{enumerate}
  \item We use $\pi_{A}$ to denote $\chi A $ for $A \in \TT_{\Gamma}$.
  \item We use $\Gamma.A$ to denote $\chi_0 A$ for $\Gamma \in \CC$ and $A \in \TT_{\Gamma}$. \label{item:notation-chi0}
  \item We use $A[s]$ to denote $s^* A$ for $s : \Gamma \to \Delta$ and $A \in \TT_{\Delta}$.
  \item We use $t[s]$ to denote the morphism given by the universal property of the following pullback square:
  \[\begin{tikzcd}
     [row sep = normal, column sep = large]
	\Gamma & \Delta \\
	& {\Gamma.s^* A} & {\Delta.A} \\
	& \Gamma & \Delta
	\arrow["s", from=1-1, to=1-2]
	\arrow["{{t[s]}}"{description}, dashed, from=1-1, to=2-2]
	\arrow[curve={height=12pt}, equals, from=1-1, to=3-2]
	\arrow["t", from=1-2, to=2-3]
	\arrow["{{s.A}}", from=2-2, to=2-3]
	\arrow["{{\chi (s^* A)}}"{description}, from=2-2, to=3-2]
	\arrow["\lrcorner"{anchor=center, pos=0.125}, draw=none, from=2-2, to=3-3]
	\arrow["{{\chi A}}"{description}, from=2-3, to=3-3]
	\arrow["s"', from=3-2, to=3-3]
  \end{tikzcd}\]
  for a section $t$ of $\chi A$, $A \in \TT_\Delta$ and $s : \Gamma \to \Delta$.
  \end{enumerate}
\end{notation}

\subsection{Functorial $\Pi$-types} \label{sec:pi-for-terms}
In this section, we define semantic structure for $\Pi$-types in non-full comprehension categories. 
We then discuss the necessary functoriality conditions that allow us to use type morphisms to interpret subtyping. 
We extend $\CTT$ with functorial $\Pi$-types and prove soundness by giving an interpretation of the rules in any comprehension category with functorial $\Pi$-types.
We also discuss how $\CTT$ with functorial $\Pi$-types supports subtyping. 



\begin{definition}[{\cite[Definition~3.2.2.3]{lindgren21}}] \label{def:pi-non-full}
    Let $(\CC, \TT, p, \chi)$ be a comprehension category. Given $\Gamma \in \CC$, $A \in \TT_\Gamma$ and $B \in \TT_{\Gamma.A}$, a \emphasize{dependent product} for $\Gamma, A$ and $B$ consists of: 
    \begin{enumerate}
    \item an object $\PiT{A}{B} \in \TT_\Gamma$;
    \item a morphism $\appmor{A}{B} : \Gamma.A.\PiT{A}{B}[\pi_A] \to \Gamma.A.B$ making the following diagram commute;
    \[\begin{tikzcd}
    {\Gamma.A.\PiT{A}{B}[\pi_A]} & {} & {\Gamma.A.B} \\
    & {\Gamma.A}
    \arrow["{\appmor{A}{B}}", from=1-1, to=1-3]
    \arrow["{\pi_{\PiT{A}{B}[\pi_A]}}"', from=1-1, to=2-2]
    \arrow["{\pi_B}", from=1-3, to=2-2]
    \end{tikzcd}\]
    \item a function $\lambdamor{A}{B}$ giving for each section $b: \Gamma.A \to \Gamma.A.B$ of $\pi_B$, a section $\lambdamor{A}{B} (b) : \Gamma \to \Gamma.\PiT{A}{B}$ of $\pi_{\PiT{A}{B}}$ such that $\lambdamor{A}{B}(-)$ and $\appmor{A}{B} \circ (-)[\pi_A]$ establish a bijection between sections of $\pi_B$ and sections of $\pi_{\PiT{A}{B}}$.
    \end{enumerate}
\end{definition}

We briefly draw the connection between \cref{def:pi-non-full} and $\Pi$-types in syntax. The object $\PiT{A}{B}$ in $\TT_\Gamma$ corresponds  to a type in context $\Gamma$ and is the dependent product of types $A$ and $B$. The morphism $\appmor{A}{B} : \Gamma.A.\PiT{A}{B}[\pi_A] \to \Gamma.A.B$ in $\CC$ gives the application of a dependent function \LongOrShortPoPl{(see \cref{rem:pi-app}).}{.} 
The function $\lambdamor{A}{B} $ is the usual $\lambda$-abstraction. The bijection between sections of $\pi_B$ and sections of $\pi_{\PiT{A}{B}}$ given by $\appmor{}{}$ and $\lambdamor{}{}$ corresponds to the usual $\beta$- and $\eta$-conversion for dependent products.

\begin{relatedwork}[\citet{jacobs93}] \label{rel:pi-jacobs}
  Jacobs interprets $\Pi$-types in a full comprehension category with right adjoints to weakening functors and an extra condition that the comprehension preserves products \cite[Subsection~5.1]{jacobs93}. 
  In a full comprehension category, \cref{def:pi-non-full} gives structure equivalent to Jacobs' definition.

  In a full comprehension category, morphisms in $\TT_\Gamma$ can be conflated with morphisms in $\CC/\Gamma$, for each $\Gamma$.
  We, however, do not assume fullness. 
  Hence, it is particularly important to make a distinction between the structure added to $\CC$ and the structure added to $\TT$.
  In Jacobs' definition, the structure of $\mathsf{app}$ and $\mathsf{lam}$ is added to the category $\TT$ in a comprehension category $(\CC, \TT, p, \chi)$. 
  Since we work with non-full comprehension categories and since we do not want $\mathsf{app}$ and $\mathsf{lam}$ to be type morphisms, we add them as a morphism in $\CC$ and a function on terms respectively.
  The structure related to subtyping is added to $\TT$. 
  This is why we use \cref{def:pi-non-full}.
\end{relatedwork}

\begin{relatedwork}[\citet{lumsdaine2015}] \label{rel:pi-lw}
  Lumsdaine and Warren define dependent products in a full comprehension category and take $\appmor{A}{B}$ to be a morphism $\appmor{A}{B}: \PiT{A}{B}[\pi_A] \to B$ in $\TT_{\Gamma.A}$ \cite[Definition~3.4.2.1]{lumsdaine2015}. In \cref{def:pi-non-full}, $\appmor{A}{B}$ is a morphism from $\Gamma.A.\PiT{A}{B}[\pi_A]$ to $\Gamma.A.B$ in $\CC$, as we do not assume fullness (see \cref{rel:pi-jacobs}). 
  In a full comprehension category,  \cref{def:pi-non-full} is equivalent to the definition of Lumsdaine and Warren.
\end{relatedwork}

\begin{relatedwork}[\citet{lindgren21}] \label{rel:pi-lindgren}
  Lindgren uses the term \emphasize{strong} dependent products to refer to what we call dependent products \cite[Definition~3.2.2.3]{lindgren21}. We do not make this distinction, as we only consider the strong case. 

  Lindgren shows that \cref{def:pi-non-full} can equivalently be expressed in terms of relative adjoints \cite[Propositions~3.2.4.2 and 3.2.4.4]{lindgren21}.
\end{relatedwork}

To be able to use type morphisms to interpret subtyping, we need to add certain functoriality conditions which formalize the intuition
that $\Pi$-types act contravariantly on the first argument, and covariantly on the second in the context of subtyping.
In particular, given subtyping relations $A' \leq_f A$ and $B[\Gamma.f] \leq_g B'$ we expect to have $\PiT{A}{B} \leq_{\subtypepi{f}{g}} \PiT{A'}{B'}$, since $\Pi$ acts contravariantly on the first argument, and covariantly on the second one.
The coercion function for $\PiT{A}{B} \leq_{\subtypepi{f}{g}} \PiT{A'}{B'}$ acts as follows.
Given a dependent function $t : \PiT{A}{B}$, the coerced dependent function $t' : \PiT{A'}{B'}$ takes a term $a' : A'$, coerces it to a term $a : A$ using $f$, applies $t$ to $a$ to get a term of type $B$, coerces it to a term of type $B'$ using $g$, and finally applies a $\lambda$-abstraction to get a term of type $\PiT{A'}{B'}$.




Now we define what it means for a comprehension category to have functorial $\Pi$-types. 
For this, we add the structure defined in \cref{def:pi-non-full} to a comprehension category, postulate a Beck-Chevalley condition, i.e. that this structure is preserved under substitution, and postulate suitable functoriality conditions.
\begin{definition} 
  \label{def:comp-cat-with-pi}
    A comprehension category $(\CC, \TT, p, \chi)$ \emphasize{has functorial $\Pi$-types} if it is equipped with a function giving for each $\Gamma \in \CC$, $A \in \TT_\Gamma$ and $B \in \TT_{\Gamma.A}$, $\PiT{A}{B}$, $\appmor{A}{B}$ and $\lambdamor{A}{B}$ as \cref{def:pi-non-full} such that: 
    \begin{enumerate}
    \item \label{item:pi-bc} for each $s : \Gamma \to \Delta$ we have an isomorphism $i_{\PiT{A}{B}, s} : \PiT{A[s]}{B[s.A]} \iso \PiT{A}{B}[s]$ in $\TT_\Gamma$;
    \item \label{item:pi-bc-lam} for each section $b : \Gamma.A \to \Gamma.A.B$ of $\pi_B$ the following diagrams commute
    \[
    \begin{tikzcd}
    [column sep = huge, row sep = large]
    {\Gamma.\PiT{A[s]}{B[s.A]} } & {\Delta.\PiT{A}{B}} \\
    \Gamma & \Delta
    \arrow["{s.\Pi(s_A,s_B)}", from=1-1, to=1-2]
    \arrow["{\lambdamor{A[s]}{B[s.A]}(b[s.A])}"{description}, from=2-1, to=1-1]
    \arrow["s"', from=2-1, to=2-2]
    \arrow["{\lambdamor{A}{B}(b)}"{description}, from=2-2, to=1-2]
    \end{tikzcd}
    \]
    where $s.\Pi(s_A,s_B) \coloneqq{s.\PiT{A}{B} \circ \Gamma.i_{\PiT{A}{B}, s}}$;
    \item \label{item:pi-subt} the comprehension category is equipped with a function giving for each $f : A' \to A$ in $\TT_{\Gamma}$ and $g : B[\Gamma.f] \to B'$ in $\TT_{\Gamma.A'}$, a morphism 
    $\subtypepi{f}{g}: \PiT{A}{B} \to \PiT{A'}{B'}$
    in $\TT_\Gamma$;
    \item \label{item:pi-subt-chi} $\Gamma.\subtypepi{f}{g}$ is the following composition in $\CC$;
    \[\begin{tikzcd}[column sep = 8em]
      {\Gamma.\PiT{A}{B}} & {\Gamma.\PiT{A}{B}.\PiT{A'}{B'} [\pi_{\PiT{A}{B}}]} & {\Gamma.\PiT{A'}{B'}}
      \arrow["{{\lambda(\Gamma.A.g[\appmor{A}{B}[\Gamma.f]])}}", from=1-1, to=1-2]
      \arrow["{\pi_{\PiT{A}{B}}.\PiT{A'}{B'}}", from=1-2, to=1-3]
    \end{tikzcd}\]
    (see \cref{sec:constructions} for more detail);
    \item \label{item:pi-subt-id} $\subtypepi{-}{-}$ preserves identity, i.e. we have 
      $\subtypepi{1_{A}}{i^{\mathsf{id}}_{B}} = 1_{\PiT{A}{B}}$
    for each suitable $A$ and $B$, where $i^{\mathsf{id}}_{B} : B[1_{\Gamma.A}] \iso B$;
    \item \label{item:pi-subt-comp} $\subtypepi{-}{-}$ preserves composition, i.e. we have 
    $\subtypepi{f \circ f'}{g' \circ g[\Gamma.f']} = \subtypepi{f'}{g'} \circ \subtypepi{f}{g}$
    for each suitable $f'$ and $g'$;
    \item \label{item:pi-sub-congs}
    $i_{\PiT{A}{B},-}$ is functorial in that it preserves $i^{\mathsf{iso}}$ and $i^{\mathsf{comp}}$ (see \cref{sec:iso-coherences} for more detail).
  \end{enumerate} 
\end{definition}
\cref{item:pi-bc,item:pi-bc-lam} of \cref{def:comp-cat-with-pi} state that $\Pi$ and $\lambda$ are preserved under substitution respectively. 
Consequently, $\mathsf{app}$ is also preserved under substitution. 
\cref{item:pi-subt,item:pi-subt-id,item:pi-subt-comp,item:pi-subt-chi,item:pi-sub-congs} give the functoriality conditions which formalize the variance of $\Pi$-types on the arguments for subtyping.
\cref{item:pi-subt-chi} expresses compatibility of the type morphism structure on the category $\TT$ with the type former structure on the category $\CC$.


The following proposition states that \cref{def:comp-cat-with-pi} is compatible with Jacobs' definition of comprehension categories with products.
\begin{proposition}[Relation to \citet{jacobs99}]
  \label{rel:pi-subty-jacobs}
  Every \emphasize{full} comprehension category with products where the comprehension functor preserves products in the sense of \citet[Section~5.1]{jacobs93} has functorial $\Pi$-types in the sense of \cref{def:comp-cat-with-pi}.
\end{proposition}

\begin{relatedwork}[\citet{coraglia_et_al:LIPIcs.TYPES.2023.3} and \citet{coraglia2024contextjudgementdeduction}]
  \label{rel:pi-coraglia}
  \citet{coraglia_et_al:LIPIcs.TYPES.2023.3} define $\Pi$-types for generalized categories with families, a structure equivalent to comprehension categories \cite{coraglia:2024}.
  This definition is similar to \cref{def:comp-cat-with-pi} regarding variance on the first and the second argument. 
  In another work, \citet{coraglia2024contextjudgementdeduction} present an alternative definition for $\Pi$-types in generalized categories with families (there called DTT), which is covariant in both arguments. 
\end{relatedwork}

\begin{relatedwork}[\citet{gambino:2023}]
  \label{rel:pi-gambino}
  \citet{gambino:2023} also define $\Pi$-types for comprehension categories,
  and there are a couple of differences to note.
  First, \citet{gambino:2023} do not require their $\Pi$-types to be functorial,
  whereas we do (\Cref{item:pi-subt,item:pi-subt-chi,item:pi-subt-id,item:pi-subt-comp} in \Cref{def:comp-cat-with-pi}).
  Second, while \citet{gambino:2023} phrase (pseudo) stability under substitution by postulating suitable Cartesian morphisms,
  we use explicit isomorphisms.
  These different ways of phrasing preservation are equivalent,
  and the advantage of our way is that it directly gives us a derivation rule in type theory.
\end{relatedwork}

\LongOrShortPoPl{
In the following remark, we give a more detailed explanation of how the morphism $\appmor{A}{B} : \Gamma.A.\PiT{A}{B} [\pi_A] \to \Gamma.A.B$ corresponds to the application of dependent products.
\begin{remark} \label{rem:pi-app}
  Let $(\CC, \TT, p, \chi)$ be a comprehension category with $\Pi$-types and $f : \Gamma \to \Gamma.\PiT{A}{B}$ a section of a projection $\pi_{\PiT{A}{B}}$ in $\CC$. The morphism $\appmor{A}{B} \circ f[\pi_A] : \Gamma.A \to \Gamma.A.B$ corresponds to the application of the dependent function $f$, in the sense that given a section $a : \Gamma \to \Gamma.A$ of $\pi_A$, the morphism 
  \[ (\appmor{A}{B} \circ f[\pi_A])[a]: \Gamma \to \Gamma.B[a]\] 
  can be thought of as the application of the function $f$ to the term $a$. Recall that  $(\appmor{A}{B} \circ f[\pi_A])[a]$ is the morphism given by the universal property of the following pullback square in $\CC$.
\[\begin{tikzcd}
	\Gamma & {\Gamma.A} \\
	& {\Gamma.B[a]} & {\Gamma.A.B} \\
	& \Gamma & {\Gamma.A}
	\arrow["a", from=1-1, to=1-2]
	\arrow[dashed, from=1-1, to=2-2]
	\arrow[equal, curve={height=12pt}, from=1-1, to=3-2]
	\arrow["{\appmor{A}{B} \circ f[\pi_A]}", from=1-2, to=2-3]
	\arrow[from=2-2, to=2-3]
	\arrow["{\pi_{B[a]}}"{description}, from=2-2, to=3-2]
	\arrow["\lrcorner"{anchor=center, pos=0.125}, draw=none, from=2-2, to=3-3]
	\arrow["{\pi_B}"{description}, from=2-3, to=3-3]
	\arrow["a"', from=3-2, to=3-3]
\end{tikzcd}\]
\end{remark}
}{}

We now discuss examples of comprehension categories with functorial $\Pi$-types, including those that arise from \AWFSs (see \cref{sec:awfs}).
In a large class of \AWFSs, $\Pi$-types can be interpreted.
More specifically,
if an \AWFS satisfies the \textbf{exponentiability property}
and comes equipped with a \textbf{functorial Frobenius structure},
then its associated comprehension category has $\Pi$-types~\cite[Proposition 4.6]{gambino:2023}.
The functoriality condition discussed in \cref{def:comp-cat-with-pi} is satisfied by the universal property of exponentials.
These conditions are satisfied, for instance, by models of cubical type theory.
They are also satisfied by the examples in \cref{sec:awfs}.
In what follows, we discuss functorial $\Pi$-types in the these examples.

\begin{example}[\cref{exa:grpd,exa:interp-grpd} ctd.]
  \label{exa:pi-grpd}
  As explained in \cref{exa:interp-grpd}, in this example, a type $A$ in context $\Gamma$ is interpreted as a split isofibration $\interp{A} : \interp{\Gamma.A} \to \interp{\Gamma}$. 
  For each groupoid $\Gamma$ and split isofibrations $A : \Gamma.A \to \Gamma$ and $B : \Gamma.A.B \to \Gamma.A$, we have a split isofibration $\PiT{A}{B} : \Gamma.\PiT{A}{B} \to \Gamma$, where for each $x \in \Gamma$, the set of objects in the fiber over $x$ is $\Pi({x': A_x}),B_{x'}$, i.e. the elements in the fibers are dependent functions.
  The functor $\app$ is $\dom (\alpha)$, where $\alpha: \PiT{A}{B}[\pi_A] \to B$ is the morphism of split fibrations in the fiber over $\Gamma.A$ that maps $(x \in \Gamma, a \in A_x, f \in \Pi ({x' : {A_x}}), B_{x'})$ to $(x \in \Gamma, a  \in A_x, f(a) \in B_a)$ for each $x, a$ and $f$.
  For a section $b : \Gamma.A \to \Gamma.A.B$ of the fibration $B$ over $\Gamma.A$ that maps $(x \in \Gamma, a \in A_x)$ to $(x \in \Gamma, a \in A_x, b(a) \in B_a)$ for each $x$ and $a$, $\lambda b$ is a section of the fibration $\PiT{A}{B}$ over $\Gamma$ that maps $x \in \Gamma$ to $(x \in \Gamma, \lambda (x' : A_x). b(x'))$.

  Now for the functoriality condition, let $f : A' \to A$ be a morphism of split fibrations in the fiber over $\Gamma$ that maps $(x \in \Gamma, a' \in A'_x)$ to $(x \in \Gamma, f(a') \in A_x)$ for each $x$ and $a'$ and let $g : B[f] \to B'$ be a morphism of split fibrations over $\Gamma.A'$ that maps $(x \in \Gamma, a' \in A'_x, b \in B_{f(a')})$ to $(x \in \Gamma, a' \in A'_x, g(b)  B'_{a'})$ for each $x,a'$ and $b$.
  The morphism $\subtypepi{f}{g} : \PiT{A}{B} \to \PiT{A'}{B'}$ maps $(x \in \Gamma, \in \Pi ({x' : A_x}), B_{x'})$ to $(x \in \Gamma, \Pi ({x' : A'_x}), ghf(x'))$. 
  It is easy to see that this assignation satisfies the functoriality conditions discussed in \cref{def:comp-cat-with-pi}.
\end{example}

\begin{example}[\cref{exa:cha-pred,exa:interp-cha-pred} ctd.]
  \label{exa:pi-cha-pred}
  As explained in \cref{exa:interp-cha-pred}, in this example, a type $A$ in context $\Gamma$ is interpreted as an $H$-valued predicate $\interp{A} : \interp{\Gamma} \to H$ and $\Gamma.A$ is interpreted as the comprehension.
  This example has $\Pi$-types.
  For each set $\Gamma$ and predicates $A : \Gamma \to H$ and $B : \Gamma.A \to H$, where $\Gamma.A$ is the comprehension of A, the predicate $\PiT{A}{B} : \Gamma \to H$ is defined to be $B(x)$ if $\top \leq A(x)$ and $\top$ otherwise.
  We have $\Gamma.A.\PiT{A}{B}[\pi_A] = \Gamma.A.B = \{x \in \Gamma | \top \leq A(x) \wedge \top \leq B(x) \}$ and the function $\app : \Gamma.A.\PiT{A}{B}[\pi_A] \to \Gamma.A.B$ is identity. 

  Now for the functoriality condition, recall that in this example the hom-sets in the fibers have at most one element. 
  Let $\Gamma$ be a set, and $A,A' : \Gamma \to H$ be two $H$-valued predicates. 
  We have a morphism in the fiber over $\Gamma$ of the form $A' \to A$. 
  This means that for all $x \in \Gamma$ we have $A'(x) \leq A(x)$. 
  Hence, $\Gamma.A' \subseteq \Gamma.A$. 
  We have a morphism in the fiber over $\Gamma.A'$ of the form $B (x) \to B'(x)$, which means that for all $x \in \Gamma$ such that $\top \leq A'(x)$ we have  $B(x) \leq B'(x)$.
  Now we verify that there is a morphism in the fiber over $\Gamma$ of the form $\PiT{A}{B} \to \PiT{A'}{B'}$. 
  For this we need that for each $x \in \Gamma$, $\PiT{A}{B}(x) \leq \PiT{A'}{B'}$. 
  If $\top \leq A'(x)$, then we have $\top \leq A(x)$ and $\PiT{A}{B}(x) = B(x)$. 
  Since $\top \leq A'(x)$, we have $\PiT{A}{B}(x) = B(x) \leq B'(x) = \PiT{A'}{B'}(x)$.
  If $A'(x) < \top$, then $\PiT{A'}{B'} (x) = \top$ and we have $\PiT{A}{B}(x) \leq \PiT{A'}{B'}$ trivially.
\end{example}
Note that in \cref{exa:pi-cha-pred}, $\Pi$-types do not give universal quantification the way that one would expect in first-order predicate logic. 
The derivation rules for universal quantification are expressed using proofs, whereas the derivation rules for $\Pi$-types in this example are expressed using terms.
In this example, proofs are represented by morphisms in the fiber categories, whereas terms are represented as sections of projections, which explains the discrepancy.

\begin{figure}[!htb]
  \centering
  \tcbset{colframe=black, colback=white, width=\textwidth, boxrule=0.1mm, arc=0mm, auto outer arc}
  \begin{tcolorbox}
  \centering
  \scriptsize
  \[
  \binaryRule{\Gamma \vdash A \type}{\Gamma.A \vdash B \type}{\Gamma \vdash \PiT{A}{B} \type}{\scriptsize{pi-form}}
  \unaryRule
  {\Gamma.A \vdash b : B}
  {\Gamma \vdash \lambda b : \PiT{A}{B}}
  {\scriptsize{pi-intro}}
  \]
  \[
  \binaryRuleBinaryConcl{\qquad \qquad \quad \Gamma \vdash A \type}{\Gamma.A \vdash B \type \qquad \qquad \quad}{\Gamma.A.\PiT{A}{B}[\pi_A] \vdash \appmor{A}{B} : \Gamma.A.B}{\Gamma.A.\PiT{A}{B}[\pi_A] \vdash \pi_B \circ \appmor{A}{B} \equiv \pi_{\PiT{A}{B}[\pi_A]} : \Gamma.A}{\scriptsize{pi-elim}}
  \]
  \[
  \unaryRule
  {\Gamma.A \vdash b : B}
  {\Gamma. A \vdash \appmor{A}{B} \circ p_2(\lambda b \circ \pi_A) \equiv b : \Gamma.A.B}{\scriptsize{pi-beta}}
  \unaryRule
  {\Gamma \vdash f : \PiT{A}{B}}
  {\Gamma \vdash \lambda (\appmor{A}{B} \circ p_2(f \circ \pi_A)) \equiv f : \Gamma.\PiT{A}{B}}{\scriptsize{pi-eta}}
  \]
  \[
  \ternaryRule{\Delta \vdash A \type}{\Delta.A \vdash B \type}{\Gamma \vdash s : \Delta}{\Gamma ~|~ \PiT{A[s]}{B[s.A]} \vdashiso i_{\PiT{A}{B}, s} : \PiT{A}{B}[s]}{\scriptsize{pi-sub}}
  \]
  \[
  \binaryRule
  {\Gamma \vdash s : \Delta}
  {\Gamma \vdash b : \PiT{A}{B}}
  {\Gamma \vdash \lambdamor{A}{B}(b) \circ s \equiv s.\PiT{A}{B} \circ \Gamma.i_{\PiT{A}{B}, s} \circ \lambdamor{A[s]}{B[s.A]} (p_2(b \circ s.A)): \Delta.\PiT{A}{B}}{\scriptsize{sub-lam}}
  \]
  \[
  \quadRuleTernaryConcl
  {\Gamma.A \vdash B \type}{\Gamma.A' \vdash B' \type}{\Gamma ~|~ A' \vdash f : A}{\Gamma.A' ~|~ B[\Gamma.f] \vdash g : B'}
  {\Gamma ~|~ \PiT{A}{B} \vdash \subtypepi{f}{g} : \PiT{A'}{B'}}
  {\Gamma. \PiT{A}{B} \vdash \Gamma.\subtypepi{f}{g} \equiv \Gamma.A'.g \circ \lambda(p_2(\Gamma.A'.g \circ (\pi_{\PiT{A}{B}[\pi_{A'}]}.B[\Gamma.f] \circ p_2 (\appmor{A}{B} \circ}
  { (\Gamma.f).\PiT{A}{B}[\pi_A] \circ \compiso{\PiT{A}{B}[\pi_{A}]}{\pi_{A}}{\Gamma.f} \circ \subiso{\PiT{A}{B}[\pi_{A}]}{\pi_{A'}}{\pi_{A} \circ \Gamma.f})))) : \Gamma. \PiT{A'}{B'}}
  {\scriptsize{subt-pi}}
  \]
  \[
  \binaryRule
  {\Gamma \vdash A \type}{\Gamma.A \vdash B \type}
  {\Gamma ~|~ \PiT{A}{B} \vdash \subtypepi{1_{A}}{i^{\mathsf{id}}_{B}} \equiv 1_{\PiT{A}{B}} : \PiT{A}{B}}
  {\scriptsize{subt-pi-id}}
  \]
  \[
  \senRule
  {\Gamma.A \vdash B \type}{\Gamma.A' \vdash B' \type}
  {\Gamma ~|~ A' \vdash f : A}
  {\Gamma.A ~|~ B[\Gamma.f] \vdash g : B'}
  {\Gamma ~|~ A'' \vdash f' : A'}
  {\Gamma.A' ~|~ B'[\Gamma.f'] \vdash g' : B''}
  {\Gamma ~|~ \PiT{A}{B} \vdash \subtypepi{f \circ f'}{g' \circ g[\Gamma.f']} = \subtypepi{f'}{g'} \circ \subtypepi{f}{g} : \PiT{A''}{B''}}
  {\scriptsize{subt-pi-comp}}
  \]
  \[
  \binaryRule{\Gamma \vdash A \type}{\Gamma.A \vdash B \type}{\Gamma ~|~ \PiT{A[1_\Gamma]}{B[1.A]} \vdash \idiso{\PiT{A}{B}} \circ i_{\PiT{A}{B}, 1_\Gamma} \equiv \subtypepi{\idisoinv{A}}{\idiso{B} \circ \subiso{B}{{1_\Gamma.A} \circ {\Gamma.\idisoinv{A}}}{1_{\Gamma.A}} \circ \compisoinv{B}{1_\Gamma.A}{\Gamma.\idisoinv{A}}} : \PiT{A}{B}}{\scriptsize{pi-sub-id}}
  \]
  \[
  \quadRuleBinaryConcl{\qquad \qquad \qquad \qquad \Theta \vdash A \type}{\Theta.A \vdash B \type}{\Gamma \vdash s' : \Delta}{\Delta \vdash s : \Theta \qquad \qquad \qquad \qquad}{\Gamma ~|~ \PiT{A[s \circ s']}{B[(s \circ s').A]} \vdash \compiso{\PiT{A}{B}}{s}{s'} \circ i_{\PiT{A}{B},s \circ s'} \equiv (i_{\PiT{A}{B},s})[s'] \circ i_{\PiT{{A[s']}}{B[s'.A]}} \circ }
  {\subtypepi{\compisoinv{A}{s}{s'}}{\compiso{B}{s.A}{s'.A[s]} \circ \subiso{B}{(s \circ s').A \circ \Gamma.\compisoinv{A}{s}{s'}}{s.A \circ s'.A[s]} \circ \compisoinv{B}{(s \circ s').A}{\Gamma.\compisoinv{A}{s}{s'}}} : {\PiT{A}{B}[s][s']}}{\scriptsize{pi-sub-comp}}
  \]
\end{tcolorbox}
\caption{Rules for functorial $\Pi$-types. Rules \protect\ruleref{pi-sub}, \protect\ruleref{sub-lam},\protect\ruleref{p-sub-id} and \protect\ruleref{p-sub-comp} use the notation introduced in \cref{lemma:s.A}. 
\LongOrShortPoPl{
For example, in Rule \protect\ruleref{pi-sub}, $s.A$ is $(s \circ \pi_{A[s]}, \Gamma.(\compiso{A}{s}{\pi_{A[s]}} [\pi_{A[s]}]) \circ p_2(1_{\Gamma.A[s]}))$.
}{}
}
\label{fig:pi-for-terms}
\end{figure}

\begin{definition} \label{def:syntax-with-pi-for-terms}
  We define the \emphasize{extension of $\CTT$ by functorial $\Pi$-types} to consist of the rules in \cref{fig:pi-for-terms}.
\end{definition}

\begin{theorem}[Soundness of Rules for $\Pi$-types] \label{thm:soundness-pi-for-terms}
  Any comprehension category with functorial $\Pi$-types models the rules of $\CTT$ and the rules for functorial $\Pi$-types in \cref{fig:pi-for-terms}.
\end{theorem}
 
We now see how elimination and computation rules for $\Pi$-types similar to those in MLTT can be derived from the rules in \cref{fig:pi-for-terms}.
\begin{proposition} \label{prop:pi-for-terms-usual-rules}
  From the rules of $\CTT$ and the rules in \cref{fig:pi-for-terms}, we can derive the following rules. 
  \[
  \binaryRule
  {\Gamma \vdash f : \PiT{A}{B}}
  {\Gamma \vdash a : A}
  {\Gamma \vdash \appmor{A}{B}(f,a) : B[a]}
  {\scriptsize{}}
  \binaryRule
  {\Gamma.A \vdash b : B}
  {\Gamma \vdash a : A}
  {\Gamma \vdash \appmor{A}{B}(\lambda b, a) \equiv p_2(b \circ a) : \Gamma.B[a]}{\scriptsize{}}
  \]
\end{proposition}
\begin{proof}
  The context morphism $\appmor{A}{B}(f,a)$ is $p_2(\appmor{A}{B} \circ p_2(f \circ \pi_A) \circ a)$. 
  \LongOrShortPoPl{
  The elimination rule can be derived using Rules \ruleref{pi-elim}, \ruleref{ctx-mor-comp} and \ruleref{sub-proj}. The computation rule can be derived using Rules \ruleref{pi-beta}, \ruleref{comp-cong} and \ruleref{sub-proj-cong}.
  }{}
\end{proof}

\begin{relatedwork}[\citet{coraglia_et_al:LIPIcs.TYPES.2023.3}] \label{rem:subty-pi}
  \TablesOrNot{
  In \cref{sec:subty-structural-rules}, we discussed the meaning of the structural rules of $\CTT$ from the subtyping point of view, where type morphisms are seen as witnesses of subtyping relations.
  In \cref{tab:subty-pi}, we discuss the meaning of the rules in \cref{fig:subty-for-pi} from this point of view and show how they relate to the rules of presented by Coraglia and Emmenegger in \cite{coraglia_et_al:LIPIcs.TYPES.2023.3}.
  }
  {
  Rule \ruleref{subt-pi} corresponds to the rules in Proposition 20 of Coraglia and Emmenegger's paper \cite{coraglia_et_al:LIPIcs.TYPES.2023.3}. 
  Under the subtyping point of view, this rule states that $A' \leq_f A$ and $B[\Gamma.f] \leq_g B'$ give $\PiT{A}{B} \leq_{\subtypepi{f}{g}} \PiT{A'}{B'}$.
  They do not however, explicitly present rules corresponding to Rules \ruleref{subt-pi-id} and \ruleref{subt-pi-comp}, which state that $\subtypepi{-}{-}$ preserves identity and composition of subtyping witnesses, respectively.
  }
  Note that Coraglia and Emmenegger write the rules in \cite[Proposition~20]{coraglia_et_al:LIPIcs.TYPES.2023.3} as if the fibrations involved were split, for simplicity.
  In our case, this equates to removing the $i^\mathsf{comp}$, $i^\mathsf{id}$ and $i^\mathsf{sub}$ terms from the rules. 
  
  \TablesOrNot{
  \begin{table}[h]
    \caption{Meaning of the rules regarding subtyping for $\Pi$-types when the judgement $\Gamma ~|~ A \vdash t : B$ expresses subtyping and relation to the rules of Coraglia and Emmenegger in \cite{coraglia_et_al:LIPIcs.TYPES.2023.3}.}
    \label{tab:subty-pi}
    \begin{tabular}{l p{6.5cm} p{3cm}}
    Rule in \cref{fig:subty-for-pi} & Meaning under Subtyping & Rule in \cite{coraglia_et_al:LIPIcs.TYPES.2023.3} \\
    \hline
    \ruleref{subt-pi} & $A' \leq_f A$ and $B[\Gamma.f] \leq_g B'$ give $\PiT{A}{B} \leq_{\subtypepi{f}{g}} \PiT{A'}{B'}$. & Rules in \cite[Proposition~20]{coraglia_et_al:LIPIcs.TYPES.2023.3} \\
    \ruleref{subt-pi-id} & $\subtypepi{-}{-}$ preserves identity witnesses. & - \\
    \ruleref{subt-pi-comp} & $\subtypepi{-}{-}$ preserves composition of witnesses. & - \\
    \end{tabular}
  \end{table}
  }
  {}
\end{relatedwork}

\subsection{Functorial $\Sigma$-types} \label{sec:sigma-for-terms}
In this section, we define semantic structure for $\Sigma$-types in non-full comprehension categories. 
We then discuss the necessary functoriality conditions that allow us to use type morphisms to interpret subtyping. 
We extend $\CTT$ with functorial $\Sigma$-types and prove soundness by giving an interpretation of the rules in any comprehension category with functorial $\Sigma$-types.
We also discuss how $\CTT$ with functorial $\Sigma$-types supports subtyping. 



\begin{definition} \label{def:sigma-non-full}
    Let $(\CC, \TT, p, \chi)$ be a comprehension category. Given $\Gamma \in \CC$, $A \in \TT_\Gamma$ and $B \in \TT_{\Gamma.A}$, \emphasize{a (strong) dependent sum} for $\Gamma, A$ and $B$ consists of: 
    \begin{enumerate}
      \item an object $\SigmaT{A}{B} \in \TT_\Gamma$;
      \item an isomorphism $\pairmor{A}{B} : \Gamma.A.B \to \Gamma.\SigmaT{A}{B}$ in $\CC$ with inverse $\projmor{A}{B}$ making the following diagram commute. 
      \[\begin{tikzcd}
      {\Gamma.A.B} && {\Gamma.\SigmaT{A}{B}} \\
      & \Gamma
      \arrow["{\pairmor{A}{B}}", from=1-1, to=1-3]
      \arrow["{\pi_A \circ \pi_B}"', from=1-1, to=2-2]
      \arrow["{\pi_{\SigmaT{A}{B}}}", from=1-3, to=2-2]
    \end{tikzcd}\]
    \end{enumerate}  
\end{definition}

Just as in the case of $\Pi$-types, \cref{def:sigma-non-full} is closely connected to the syntactic representation of $\Sigma$-types. 
The object $\SigmaT{A}{B}$ in $\TT_\Gamma$ corresponds to the dependent sum of types $A$ and $B$ in context $\Gamma$. 
The morphism $\pairmor{A}{B} : \Gamma.A.B \to \Gamma.\SigmaT{A}{B}$ in $\CC$ corresponds to pairing of a term of type $A$ with a term of type $B$ that depends on $A$. 
The morphism $\projmor{A}{B} : \Gamma.\SigmaT{A}{B} \to \Gamma.A.B$ corresponds to the second projection and $\pi_B \circ \projmor{A}{B}$ gives the first projection. 
$\projmor{A}{B}$ being an inverse of $\pairmor{A}{B}$ gives the $\beta$- and $\eta$-rules.

\begin{relatedwork}[\citet{jacobs93}] \label{rel:sigma-jacobs}
  Jacobs defines dependent sums in a (full) comprehension category as left adjoints to certain reindexing functors --- the weakening functors of the form $\pi_A^*$ \cite[Definition~3.10~(i)]{jacobs93}.
  In a full comprehension category, this definition is equivalent to \cref{def:sigma-non-full}.

  As explained in \cref{rel:pi-jacobs}, we do not assume fullness; hence, we make a distinction between the structure added to $\CC$ and the structure added to $\TT$.
  In particular, we are interested in having the structure of $\mathsf{pair}$ and $\mathsf{proj}$ on $\CC$ and the structure related to subtyping on $\TT$.
  We use \cref{def:sigma-non-full}, where the structure of $\mathsf{pair}$ and $\mathsf{proj}$ is on the category $\CC$.
  In Jacobs' definition this structure is on the category $\TT$.
\end{relatedwork}

\begin{relatedwork}[\citet{lumsdaine2015}] \label{rel:sigma-lw}
  Lumsdaine and Warren define $\Sigma$-types with an induction rule \cite[Definition~3.4.4.1]{lumsdaine2015}.
  In contrast, we describe $\Sigma$-types via projections, as this description is simpler.
\end{relatedwork}

To be able to use type morphisms to interpret subtyping, we need to add certain functoriality conditions which formalize the intuition that $\Sigma$-types act covariantly on both arguments in the context of subtyping.
In particular, given subtyping relations $A \leq_f A'$ and $B \leq_g B'[\Gamma.f]$ we have $\SigmaT{A}{B} \leq_{\subtypesigma{f}{g}} \SigmaT{A'}{B'}$, since $\Sigma$ acts covariantly on both the first and the second arguments. The coercion function for $\SigmaT{A}{B} \leq_{\subtypesigma{f}{g}} \SigmaT{A'}{B'}$ takes a dependent pair $(a,b) : \SigmaT{A}{B}$ to the coerced dependent pair $(a', b') : \PiT{A'}{B'}$ as follows. The term $a' : A'$ is obtained by coercing $a$ to $A'$ using $f$, and the term $b' : B'$ is obtained by coercing $b$ to $B'[\Gamma.f]$ using $g$.

Now we define what it means for a comprehension category to have functorial $\Sigma$-types. 
For this, we add the structure defined in \cref{def:sigma-non-full} to a comprehension category, postulate a Beck-Chevalley condition, i.e. that this structure is preserved under substitution, and postulate suitable functoriality conditions.
\begin{definition} \label{def:comp-cat-with-sigma}
    A comprehension category $(\CC, \TT, p, \chi)$ \emphasize{has functorial $\Sigma$-types} if it is equipped with a function giving $\SigmaT{A}{B}$, $\pairmor{A}{B}$ and $\projmor{A}{B}$ for each suitable $\Gamma, A, B$ such that:
    \begin{enumerate}
        \item \label{item:sigma-bc} for each $s : \Gamma \to \Delta$ in $\CC$ we have $i_{\SigmaT{A}{B}, s} : \SigmaT{A[s]}{B[s.A]} \iso \SigmaT{A}{B}[s]$ in $\TT_\Gamma$;
        \item \label{item:sigma-bc-pair} for each $A \in \TT_\Delta$, $B \in \TT_{\Delta.A}$ the following diagram commutes
        \[\begin{tikzcd}
          [column sep = large]
        	{\Gamma.\SigmaT{A[s]}{B[s.A]}} & {\Delta.\SigmaT{A}{B}} \\
        	{\Gamma.A[s].B[s.A]} & {\Delta.A.B}
        	\arrow["{s.\Sigma(s_A,s_B)}", from=1-1, to=1-2]
        	\arrow["{\pairmor{A[s]}{B[s.A]}}", from=2-1, to=1-1]
        	\arrow["{\pairmor{A}{B}}"', from=2-2, to=1-2]
        	\arrow["{s.A.B}"', from=2-1, to=2-2]
        \end{tikzcd}\]
        where $s.\Sigma(s_A,s_B) \coloneqq s.\SigmaT{A}{B} \circ \Gamma.i_{\SigmaT{A}{B}, s}$;
    \item \label{item:sigma-subt} the comprehension category is equipped with a function giving for each  $f : A \to A'$ in $\TT_{\Gamma}$ and $g : B \to B';[\Gamma.f]$ in $\TT_{\Gamma.A}$, a morphism 
    $ \subtypesigma{f}{g} : \SigmaT{A}{B} \to \SigmaT{A'}{B'} $
    in $\TT_\Gamma$;
    \item \label{item:sigma-subt-chi} $\Gamma.\subtypesigma{f}{g}$ is the following composite,
    \[\begin{tikzcd}[column sep=large]
      {\Gamma.{\SigmaT{A}{B}}} & {\Gamma.A.B} & {\Gamma.A.B'[\Gamma.f]} & {\Gamma.A'.B'} & {\Gamma.\SigmaT{A'}{B'}}
      \arrow["{\projmor{A}{B}}", from=1-1, to=1-2]
      \arrow["{\Gamma.A.g}", from=1-2, to=1-3]
      \arrow["{\Gamma.f.B'}", from=1-3, to=1-4]
      \arrow["{\pairmor{A'}{B'}}", from=1-4, to=1-5]
    \end{tikzcd}\]
    (see \cref{sec:constructions} for more detail);
    \item \label{item:sigma-subt-id} $\subtypesigma{-}{-}$ preserves identities, i.e. we have 
      $\subtypesigma{1_{A}}{i^{\mathsf{id}^{-1}}_{B}} = 1_{\SigmaT{A}{B}}$
    for each suitable $A$ and $B$, where $i^{\mathsf{id}}_{B} : B[1_{\Gamma.A}] \iso B$;
    \item \label{item:sigma-subt-comp} $\subtypesigma{-}{-}$ preserves composition, i.e. we have  
    $\subtypesigma{f' \circ f}{g'[\Gamma.f] \circ g} = \subtypesigma{f'}{g'} \circ \subtypesigma{f}{g}$
    for each suitable $f'$ and $g'$;
    \item \label{item:sigma-sub-congs}
    $i_{\SigmaT{A}{B},-}$ is functorial in that it preserves $i^{\mathsf{iso}}$ and $i^{\mathsf{comp}}$ (see \cref{sec:iso-coherences} for more detail).
  \end{enumerate}   
\end{definition}

\cref{item:sigma-bc,item:sigma-bc-pair} of \cref{def:comp-cat-with-sigma} state that $\Sigma$ and $\mathsf{pair}$ are preserved under substitution respectively.
Consequently, $\mathsf{proj}$ is also preserves under substitution.
\cref{item:sigma-subt,item:sigma-subt-comp,item:sigma-subt-id,item:sigma-subt-chi,item:sigma-sub-congs} give the functoriality conditions which formalize the variance of $\Sigma$-types on the arguments for subtyping.
\cref{item:sigma-subt-chi} expresses compatibility of the type morphism structure on the category $\TT$ with the type former structure on the category $\CC$.

The following proposition states that \cref{def:comp-cat-with-sigma} is compatible with Jacobs' definition of comprehension categories with sums.
\begin{proposition}[Relation to \citet{jacobs99}]
  \label{rel:sigma-subty-jacobs}
  Every \emphasize{full} comprehension category with sums in the sense of Jacobs \cite{jacobs93} has functorial $\Sigma$-types in the sense of \cref{def:comp-cat-with-sigma}.
\end{proposition}

\begin{relatedwork}[\citet{coraglia_et_al:LIPIcs.TYPES.2023.3}]
  \label{rel:sigma-subty-ce}
  \citet{coraglia_et_al:LIPIcs.TYPES.2023.3} define $\Sigma$-types for generalized categories with families, a structure equivalent to comprehension categories \cite{coraglia:2024}.
  This definition is similar to \cref{def:comp-cat-with-pi} regarding variance on the arguments. 
\end{relatedwork}

\begin{relatedwork}[\citet{gambino:2023}]
  \label{rel:sig-gambino}
  \citet{gambino:2023} also define $\Sigma$-types for comprehension categories,
  but in a different way than we do.
  As with $\Pi$-types, 
  we require $\Sigma$-types to be functorial, while they do not, and we phrase stability using isomorphisms instead of Cartesian morphisms.
  A more notable difference lies in the fact that \citet{gambino:2023} express $\Sigma$-types by giving the usual introduction and elimination rule (i.e., $\Sigma$-induction),
  whereas we express $\Sigma$-types equivalently via pairing and projections.
\end{relatedwork}

We now discuss examples of comprehension categories with functorial $\Sigma$-types, including those that arise from \AWFSs (see \cref{sec:awfs}).
Whereas for $\Pi$-types we need to require additional assumptions,
every \AWFS supports $\Sigma$-types~\cite[Proposition 4.3]{gambino:2023}.
This is because the right class of maps is closed under composition.
Since $\Sigma$-types are interpreted using composition, the functoriality condition of \cref{def:comp-cat-with-sigma} is also satisfied.
This means that in particular, the examples in \cref{sec:awfs} have functorial $\Sigma$-types. 
In what follows, we discuss functorial $\Sigma$-types in the these examples.

\begin{example}[\cref{exa:grpd,exa:interp-grpd} ctd.]
  \label{exa:sigma-grpd}
  As explained in \cref{exa:interp-grpd}, in this example, a type $A$ in context $\Gamma$ is interpreted as a split isofibration $\interp{A} : \interp{\Gamma.A} \to \interp{\Gamma}$. 
  This example has $\Sigma$-types.
  For each groupoid $\Gamma$ and split isofibrations $A : \Gamma.A \to \Gamma$ and $B : \Gamma.A.B \to \Gamma.A$, we have a split isofibration $\SigmaT{A}{B} : \Gamma.\SigmaT{A}{B} \to \Gamma$, where for each $x \in \Gamma$, the set of objects in the fiber over $x$ is $\Sigma ({x': A_x}), B_{x'}$, i.e. the elements in the fibers are dependent pairs.
  The functor $\mathsf{pair} : \Gamma.A.B \to \Gamma.\SigmaT{A}{B}$ is $\dom (\alpha)$, where $\alpha : (A \circ B) \to \SigmaT{A}{B}$ is the morphism of split fibrations over $\Gamma$ that maps $(x \in \Gamma, a \in A_x, b \in B_a)$ to $(x \in \Gamma, (a,b) \in \Sigma ({x' : A_x}),B_{x'})$.

  Now for the functoriality condition, let $f : A \to A'$ be a morphism of split fibrations in the fiber over $\Gamma$ that maps $(x \in \Gamma, a \in A_x)$ to $(x \in \Gamma, f(a) \in A'_x)$ and let $g : B \to B'[f]$ be a morphism of split fibrations over $\Gamma.A$ that maps $(x \in \Gamma, a \in A_x, b \in B_{a})$ to $(x \in \Gamma, f(a) \in A'_x, g(b) : B'_{f(a)})$.
  The morphism $\subtypesigma{f}{g} : \SigmaT{A}{B} \to \SigmaT{A'}{B'}$ maps $(x \in \Gamma, (h_1, h_2) \in \Sigma ({x' : A_x}), B_{x'})$ to $(x \in \Gamma, (f(h_1), g(h_2)) \in \Sigma ({x': A'_x}), B'_{x'})$. 
  It is easy to see that this assignation satisfies the properties from \cref{def:comp-cat-with-sigma}.
\end{example}

\begin{example}[\cref{exa:cha-pred,exa:interp-cha-pred} ctd.]
  \label{exa:sigma-cha-pred}
  As explained in \cref{exa:interp-cha-pred}, in this example, a type $A$ in context $\Gamma$ is interpreted as an $H$-valued predicate $\interp{A} : \interp{\Gamma} \to H$ and $\Gamma.A$ is interpreted as the comprehension.
  This example has $\Sigma$-types.
  For each set $\Gamma$ and predicates $A : \Gamma \to H$ and $B : \Gamma.A \to H$, where $\Gamma.A$ is the comprehension of A, the predicate $\SigmaT{A}{B} : \Gamma \to H$ is defined to be $B(x)$ if $A(x) = \top$ and $\bot$ otherwise.
  We have $\Gamma.A.B = \Gamma.\SigmaT{A}{B} = \{x \in \Gamma | \top \leq A(x) \wedge \top \leq B(x)\}$ and the function $\mathsf{pair} : \Gamma.A.B \to \Gamma.\SigmaT{A}{B}$ is identity.

  Now for the functoriality condition, recall that in this example the hom-sets in the fibers have at most one element. 
  Let $\Gamma$ be a set, and $A,A' : \Gamma \to H$ be two $H$-valued predicates. 
  We have a morphism in the fiber over $\Gamma$ of the form $A \to A'$. 
  This means that for all $x \in \Gamma$ we have $A(x) \leq A'(x)$. 
  Hence, $\Gamma.A \subseteq \Gamma.A'$. 
  We have a morphism in the fiber over $\Gamma.A$ of the form $B (x) \to B' (x)$, which means that for all $x \in \Gamma$ such that $\top \leq A(x)$ we have  $B(x) \leq B'(x)$.
  Now we verify that there is a morphism in the fiber over $\Gamma$ of the form $\SigmaT{A}{B} \to \SigmaT{A'}{B'}$. 
  For this we need that for each $x \in \Gamma$, $\SigmaT{A}{B}(x) \leq \SigmaT{A'}{B'}$. 
  If $\top \leq A(x)$, then we have $\top \leq A'(x)$.
  Hence, $\SigmaT{A}{B}(x) = B(x) \leq B'(x) = \SigmaT{A'}{B'}$.
  If $A(x) < \top$, then $\SigmaT{A}{B} (x) = \bot$ and we have $\SigmaT{A}{B}(x) \leq \SigmaT{A'}{B'}$ trivially.
\end{example}
Recall that $\Pi$-types in \cref{exa:pi-cha-pred} do not correspond to universal quantification. 
For the same reason, $\Sigma$-types in \cref{exa:sigma-cha-pred} do not correspond to existential quantification.

\begin{figure}[!htb]
  \centering
  \tcbset{colframe=black, colback=white, width=\textwidth, boxrule=0.1mm, arc=0mm, auto outer arc}
  \begin{tcolorbox}
  \centering
  \scriptsize
  \begin{tabular}{cc}
    $\binaryRule{\Gamma \vdash A \type}{\Gamma.A \vdash B \type}{\Gamma \vdash \SigmaT{A}{B} \type}{\scriptsize{sigma-form}}$ &
    $ \binaryRuleBinaryConcl{\Gamma \vdash A \type}{\Gamma.A \vdash B \type}{\Gamma.A.B \vdash \pairmor{A}{B} : \Gamma.\SigmaT{A}{B}}{\Gamma \vdash \pi_{\SigmaT{A}{B}} \circ \pairmor{A}{B} \equiv \pi_A \circ \pi_B : \Gamma}{\scriptsize{sigma-intro}} $ 
  \end{tabular}
  \begin{tabular}{cc}
    $ \binaryRule{\Gamma \vdash A \type}{\Gamma.A \vdash B \type}{\Gamma.\SigmaT{A}{B} \vdash \projmor{A}{B} : \Gamma.A.B}{\scriptsize{sigma-elim}}$ &
    $ \binaryRule{\Gamma \vdash A \type}{\Gamma.A \vdash B \type}{\Gamma.A.B \vdash \projmor{A}{B} \circ \pairmor{A}{B} \equiv 1_{\Gamma.A.B} : \Gamma.A.B}{\scriptsize{sigma-beta}} $
  \end{tabular}
  \[ 
  \binaryRule{\Gamma \vdash A \type}{\Gamma.A \vdash B \type}{\Gamma.\SigmaT{A}{B} \vdash \pairmor{A}{B} \circ \projmor{A}{B} \equiv 1_{\Gamma.\SigmaT{A}{B}} : \Gamma.\SigmaT{A}{B}}{\scriptsize{sigma-eta}} \]
  \[ \ternaryRule{\Delta \vdash A \type}{\Delta.A \vdash B \type}{\Gamma \vdash s : \Delta}{\Gamma ~|~ \SigmaT{A[s]}{B[s.A]} \vdashiso i_{\SigmaT{A}{B}, s} : \SigmaT{A}{B}[s]}{\scriptsize{sigma-sub}}
  \]
  \[ \ternaryRule{\Delta \vdash A \type}{\Delta.A \vdash B \type}{\Gamma \vdash s : \Delta}{\Gamma.A[s].B[s.A] \vdash s.\SigmaT{A}{B} \circ \Gamma.i_{\SigmaT{A}{B}, s} \circ \pairmor{A[s]}{B[s.A]} \equiv \pairmor{A}{B} \circ s.A.B : \Delta.\SigmaT{A}{B}}{\scriptsize{sub-pair}} \]
  \[
  \quadRuleBinaryConcl
  {\Gamma.A \vdash B \type}{\Gamma.A' \vdash B' \type}{\Gamma ~|~ A \vdash f : A'}{\Gamma.A ~|~ B \vdash g : B'[\Gamma.f]}
  {\Gamma ~|~ \SigmaT{A}{B} \vdash \subtypesigma{f}{g} : \SigmaT{A'}{B'}}
  {\Gamma.\SigmaT{A}{B} \vdash \Gamma.\subtypesigma{f}{g} \equiv \pairmor{A'}{B'} \circ (\Gamma.f).B' \circ \Gamma.A.g \circ \projmor{A}{B} : \Gamma.\SigmaT{A'}{B'}}
  {\scriptsize{subt-sigma}}
  \]
  \[
  \binaryRule
  {\Gamma \vdash A \type}{\Gamma.A \vdash B \type}
  {\Gamma ~|~ \SigmaT{A}{B} \vdash \subtypesigma{1_{A}}{i^{\mathsf{id}^{-1}}_{B}} \equiv 1_{\SigmaT{A}{B}} : \SigmaT{A}{B}}
  {\scriptsize{subt-sigma-id}}
  \]
  \[
  \senRule{\Gamma.A \vdash B \type}{\Gamma.A' \vdash B' \type}{\Gamma ~|~ A \vdash f : A'}{\Gamma.A ~|~ B \vdash g : B'[\Gamma.f]}{\Gamma ~|~ A' \vdash f' : A''}{\Gamma.A' ~|~ B' \vdash g' : B''[\Gamma.f']}
  {\Gamma ~|~ \SigmaT{A}{B} \vdash \subtypesigma{f' \circ f}{g'[\Gamma.f] \circ g} \equiv \subtypesigma{f'}{g'} \circ \subtypesigma{f}{g} : \SigmaT{A''}{B''}}
  {\scriptsize{subt-sigma-comp}}
  \]
  \[
  \binaryRule{\Gamma \vdash A \type}{\Gamma.A \vdash B \type}{\Gamma ~|~ \SigmaT{A[1_\Gamma]}{B[1.A]} \vdash \idiso{\SigmaT{A}{B}} \circ i_{\SigmaT{A}{B}, 1_\Gamma} \equiv \subtypesigma{\idiso{A}}{\idiso{B} \circ \subiso{B}{{1_\Gamma.A} \circ {\Gamma.\idiso{A}}}{1_{\Gamma.A}} \circ \compisoinv{B}{1_\Gamma.A}{\Gamma.\idiso{A}}} : \SigmaT{A}{B}}{\scriptsize{sigma-sub-id}}
  \]
  \[
  \quadRuleBinaryConcl{\quad \qquad \qquad \qquad \Theta \vdash A \type}{\Theta.A \vdash B \type}{\Gamma \vdash s' : \Delta}{\Delta \vdash s : \Theta \qquad \qquad \qquad \quad}{\Gamma ~|~ \SigmaT{A[s \circ s']}{B[(s \circ s').A]} \vdash \compiso{\SigmaT{A}{B}}{s}{s'} \circ i_{\SigmaT{A}{B},s \circ s'} \equiv (i_{\SigmaT{A}{B},s})[s'] \circ i_{\SigmaT{A[s']}{B[s'.A]}} \circ }
  {\subtypesigma{\compiso{A}{s}{s'}}{\compiso{B}{s.A}{s'.A[s]} \circ \subiso{B}{(s \circ s').A \circ \Gamma.\compiso{A}{s}{s'}}{s.A \circ s'.A[s]} \circ \compisoinv{B}{(s \circ s').A}{\Gamma.\compiso{A}{s}{s'}}} : {\SigmaT{A}{B}[s][s']}}{\scriptsize{sigma-sub-comp}}
  \]
\end{tcolorbox}
\caption{Rules for functorial $\Sigma$-types. Rules \protect\ruleref{sigma-sub}, \protect\ruleref{sub-pair},\protect\ruleref{sigma-sub-id} and \protect\ruleref{sigma-sub-comp} use the notation introduced in \cref{lemma:s.A}. For example, in Rule sigma-sub, $s.A$ is $(s \circ \pi_{A[s]}, \Gamma.(\compiso{A}{s}{\pi_{A[s]}} [\pi_{A[s]}]) \circ p_2(1_{\Gamma.A[s]}))$.}
\label{fig:sigma-for-terms}
\end{figure}

\begin{definition} \label{def:syntax-with-sigma-for-terms}
  We define the \emphasize{extension of $\CTT$ by functorial $\Sigma$-types} to consist of the rules in \cref{fig:sigma-for-terms}.
\end{definition}

\begin{theorem}[Soundness of Rules for $\Sigma$-types] \label{thm:soundness-sigma-for-terms}
  Any comprehension category with functorial $\Sigma$-types models the rules of $\CTT$ and the rules for functorial $\Sigma$-types in \cref{fig:sigma-for-terms}.
\end{theorem}

In the following proposition, we see that we can derive first and second projection rules similar to those of Martin-Löf type theory from the rules in \cref{fig:sigma-for-terms}.
\begin{proposition} \label{prop:sigma-for-terms-usual-rules}
  From the rules of $\CTT$ and the rules in \cref{fig:sigma-for-terms}, we can derive the following rules.
  \[
  \unaryRuleBinaryConcl
  {\qquad \Gamma \vdash p : \SigmaT{A}{B} \qquad}
  {\Gamma \vdash \mathsf{proj}_1 p : A}
  {\Gamma \vdash \mathsf{proj}_2 p : B[ \mathsf{proj}_1 p]}
  {}
  \]
\end{proposition}
\begin{proof}
  The context morphism $\mathsf{proj}_1 p$ is $\pi_B \circ \projmor{A}{B} \circ p$ and $\mathsf{proj}_2 p$ is $p_2(\projmor{A}{B} \circ p)$. 
  \LongOrShortPoPl{The rules can be derived using Rules \ruleref{sigma-elim}, \ruleref{ctx-mor-comp} and \ruleref{sub-proj}.}
  {}
\end{proof}

\begin{relatedwork}[\citet{coraglia_et_al:LIPIcs.TYPES.2023.3}]
  \label{rem:subty-sigma}
  \TablesOrNot{
  In \cref{tab:subty-sigma}, we discuss the meaning of the rules in \cref{fig:subty-for-sigma} from the subtyping point of view, where type morphisms are seen as witnesses of subtyping relations and show how they relate to the rules discussed by Coraglia and Emmenegger in \cite{coraglia_et_al:LIPIcs.TYPES.2023.3}. 
  }
  {
  Rule \ruleref{subt-sigma} corresponds to the rules in Proposition 21 of Coraglia and Emmenegger's paper \cite{coraglia_et_al:LIPIcs.TYPES.2023.3}.
  Under the subtyping point of view, this rule states that $A \leq_f A'$ and $B \leq_g B'[\Gamma.f]$ give $\SigmaT{A}{B} \leq_{\subtypesigma{f}{g}} \SigmaT{A'}{B'}$.
  They do not however, explicitly present rules corresponding to Rules \ruleref{subt-sigma-id} and \ruleref{subt-sigma-comp}, which state that $\subtypesigma{-}{-}$ preserves identity and composition subtyping witnesses, respectively.
  }
  Note that Coraglia and Emmenegger write the rules in \cite[Proposition~21]{coraglia_et_al:LIPIcs.TYPES.2023.3} as if the fibrations involved were split, for simplicity.
  In our case, this equates to removing the $i^\mathsf{comp}$ and $i^\mathsf{id}$ from the rules\LongOrShort{, for example in the definition of the term $(\Gamma.f).B'$ in Rule \ruleref{subt-sigma}.}{.}
  \TablesOrNot{
  \begin{table}[h]
    \caption{Meaning of the rules regarding subtyping for $\Sigma$-types when the judgement $\Gamma ~|~ A \vdash t : B$ expresses subtyping and relation to the rules of Coraglia and Emmenegger in \cite{coraglia_et_al:LIPIcs.TYPES.2023.3}.}
    \label{tab:subty-sigma}
    \begin{tabular}{p{3cm} p{6.5cm} p{2.5cm}}
    Rule in \cref{fig:subty-for-sigma} & Meaning under Subtyping & Rule in \cite{coraglia_et_al:LIPIcs.TYPES.2023.3} \\
    \hline
    \ruleref{subt-sigma} & $A \leq_f A'$ and $B \leq_g B'[\Gamma.f]$ give $\SigmaT{A}{B} \leq_{\subtypesigma{f}{g}} \SigmaT{A'}{B'}$. & Rules in \cite[Proposition~21]{coraglia_et_al:LIPIcs.TYPES.2023.3} \\
    \ruleref{subt-sigma-id} & $\subtypesigma{-}{-}$ preserves identity witnesses. & - \\
    \ruleref{subt-sigma-comp} & $\subtypesigma{-}{-}$ preserves composition of witnesses. & - \\
    \end{tabular}
  \end{table}
  }{}
\end{relatedwork}

\subsection{Functorial $\id$-types} \label{sec:id-for-terms}
In this section, we define semantic structure for $\id$-types in non-full comprehension categories. 
We then discuss the necessary functoriality conditions that allow us to use type morphisms to interpret subtyping. 
We extend $\CTT$ with functorial $\id$-types and prove soundness by giving an interpretation of the rules in any comprehension category with functorial $\id$-types.
We also discuss how $\CTT$ with functorial $\id$-types supports subtyping. 

\begin{definition}[{\cite[Definition~2.3.1]{lumsdaine2015}}] \label{def:id-non-full}
  Let $(\CC, \TT, p, \chi)$ be a comprehension category. Given $\Gamma \in \CC$, $A \in \TT_\Gamma$, \emphasize{an identity type} for $\Gamma$ and $A$ consists of: 
  \begin{enumerate}
      \item an object $\id_A \in \TT_{\Gamma.A.A[\pi_A]}$;
      \item a morphism $\reflmor{A} : \Gamma.A \to \Gamma.A.A[\pi_A].\id_A$ in $\CC$ making the following diagram commute,
      \[\begin{tikzcd}
        {\Gamma.A} && {\Gamma.A.A[\pi_A].\id_A} \\
        & {\Gamma.A.A[\pi_A]}
        \arrow["{\reflmor{A}}", from=1-1, to=1-3]
        \arrow["{\Delta_A}"', from=1-1, to=2-2]
        \arrow["{\pi_A}", from=1-3, to=2-2]
      \end{tikzcd}\]
      where $\Delta_A$ is the diagonal morphism of the form $\Gamma.A \to \Gamma.A.A[\pi_A]$; 
      \item for each $C \in \TT_{\Gamma.A.A[\pi_A].\id_A}$ and $d : \Gamma.A \to \Gamma.A.A[\pi_A].\id_A.C$ making the outer square commute, a section $\jmor{A}{C}{d}: \Gamma.A.A[\pi_A].\id_A \to \Gamma.A.A[\pi_A].\id_A.C$ of $\pi_C$ making the following two triangles commute.
      \[\begin{tikzcd}
        {\Gamma.A} & {\Gamma.A.A[\pi_A].\id_A.C} \\
        {\Gamma.A.A[\pi_A].\id_A} & {\Gamma.A.A[\pi_A].\id_A}
        \arrow["d", from=1-1, to=1-2]
        \arrow["{\reflmor{A}}"', from=1-1, to=2-1]
        \arrow["{\pi_C}", from=1-2, to=2-2]
        \arrow["{\jmor{A}{C}{d}}"{description}, from=2-1, to=1-2]
        \arrow[equals, from=2-1, to=2-2]
      \end{tikzcd}\]
  \end{enumerate}   
\end{definition}

Similar to the case of $\Pi$- and $\Sigma$-types, we tie \cref{def:id-non-full} to identity types in syntax. The object $\id_A$ in $\TT_\Gamma$ corresponds the identity type for terms of type $A$ in context $\Gamma.A.A$. 
The morphism $\reflmor{} : \Gamma.A \to \Gamma.A.A[\pi_A].\id_A$ in $\CC$ gives the reflexivity proof. 
The morphism $\jmor{A}{B} : \Gamma.A.A[\pi_A].\id_A \to \Gamma.A.A[\pi_A].\id_A.C$ in $\CC$ gives the elimination rule for identity types. 
This rule states that to construct a term of type $C$ in the context $\Gamma.A.A.\id_A$, it suffices to provide a term of type $C$ when the second and third variables are replaced by the first one and the reflexivity proof, in the context $\Gamma.A$.
 
\begin{relatedwork}[\citet{jacobs99}] \label{rel:id-jacobs}
  Jacobs defines identity types in a (full) comprehension category as left adjoints to certain reindexing functors --- the contraction functors of the form $\Delta_A^*$, where $\Delta_A : \Gamma.A \to \Gamma.A.A$ is a diagonal morphism \cite[Definition~10.5.1]{jacobs99}. This definition gives an extensional identity type, whereas \cref{def:id-non-full} gives an intensional identity type.
\end{relatedwork}

To be able to use type morphisms to interpret subtyping, we need to add certain functoriality conditions which formalize the intuition of how $\id$-types interact with subtyping.
In particular, given a subtyping relation $A \leq_t B$ we have $\id_A \leq_{\subtypeid{t}} \id_B[\Gamma.t.t]$, since $\id$-types preserve subtyping. 
In the semantics, this means that for each morphism $t : A \to B$ in $\TT_\Gamma$, we have a morphism $t' : \id_A \to \id_B[\Gamma.t.t]$ in $\TT_{\Gamma.A.A[\pi_A]}$, which is equivalent to having a morphism $\subtypeid{t} : \id_A \to \id_B$ in $\TT$ with $p (\subtypeid{t}) = \Gamma.t.t$. 

Now we define what it means for a comprehension category to have functorial $\id$-types. 
For this, we add the structure defined in \cref{def:id-non-full} to a comprehension category, postulate a Beck-Chevalley condition, i.e. that this structure is preserved under substitution, and postulate suitable functoriality conditions.

\begin{definition} \label{def:comp-cat-with-id}
  A comprehension category $(\CC, \TT, p, \chi)$ \emphasize{has functorial identity types} if it is equipped with a function giving $\id_A$, $\reflmor{A}$ and $\jmor{A}{C}{d}$ for each suitable $\Gamma, A, C, d$ such that:
  \begin{enumerate}
    \item \label{item:id-bc} for each $s : \Gamma \to \Delta$ in $\CC$, we have an isomorphism $i_{\id_A , s} : \id_{A[s]} \iso \id_A [s.A.A[\pi_A] \circ \chi_0 i^{\mathsf{comp}}_A]$ in $\TT_{\Gamma.A[s].A[s][\pi_{A[s]}]}$, where $i^{\mathsf{comp}}_A : A[s][\pi_{A[s]}] \iso A[\pi_A][s.A]$; 
    \item \label{item:id-bc-rj} for each $A \in \TT_\Delta$, $C \in \TT_{\Delta.A.A[\pi_A]}$ and $d : \Delta.A.A[\pi_A] \to \Delta.A.A[\pi_A].\id_A.C$, the following diagram commutes;
    \[\begin{tikzcd}
    [row sep = large]
    {\Gamma.A[s].A[s][\pi_{A[s]}].\id_{A[s]}.C[(s.A.A[\pi_A] \circ \chi_0 i^{\mathsf{comp}}_{A}).\id_A \ \circ \chi_0 i_{\id_A , s}]} & {\Delta.A.A[\pi_A].\id_A.C} \\
    {\Gamma.A[s].A[s][\pi_{A[s]}].\id_{A[s]}} & {\Delta.A.A[\pi_A].\id_A} \\
    {\Gamma.A[s].A[s][\pi_{A[s]}]} & {\Delta.A.A[\pi_A]}
    \arrow["({(s.A.A[\pi_A] \circ \chi_0 i^{\mathsf{comp}}_{A}).\id_A \circ \chi_0 i_{\id_A , s}}).C"{yshift=1ex},from=1-1, to=1-2]
    \arrow["{\jmor{A[s]}{C[(s.A.A[\pi_A] \circ \chi_0 i^{\mathsf{comp}}_{A}).\id_A \ \circ \chi_0 i_{\id_A , s}]}{d[s.A.A[\pi_A] \circ \chi_0 i^{\mathsf{comp}}_{A}]}}"{description}, from=2-1, to=1-1]
    \arrow["{(s.A.A[\pi_A] \circ \chi_0 i^{\mathsf{comp}}_{A}).\id_A \circ \chi_0 i_{\id_A , s}}"', from=2-1, to=2-2]
    \arrow["{\jmor{A}{C}{d}}"{description}, from=2-2, to=1-2]
    \arrow["{\reflmor{A[s]}}"{description}, from=3-1, to=2-1]
    \arrow["{s.A.A[\pi_A] \circ \chi_0 i^{\mathsf{comp}}_{A}}"', from=3-1, to=3-2]
    \arrow["{\reflmor{A}}"{description}, from=3-2, to=2-2]
    \end{tikzcd}\]
  \item \label{item:id-subt} the comprehension category is equipped with a function giving for each $t : A \to B$ in $\TT_{\Gamma}$, a morphism 
  $ \subtypeid{t}  : \id_A \to \id_B \quad \text{with} \quad p (\subtypeid{t}) =  \Gamma.t.t, $;
    \item \label{item:id-subt-id} $\subtypeid{(-)}$ preserves identities, i.e. $\subtypeid{1_A}  = 1_{\id_A}$ for each $A \in \TT_{\Gamma}$;
    \item \label{item:id-subt-comp} $\subtypeid{(-)}$ preserves composition, i.e. $\subtypeid{t' \circ t} = \subtypeid{t'}  \circ \subtypeid{t}$ for each suitable $t$ and $t'$;
    \item \label{item:id-subt-chi} the following diagrams commute for each $t : A \to B$ in $\TT_{\Gamma}$,
    \[
    \begin{tabular}{cc}
    \begin{tikzcd}
    [column sep = large]
    {\Gamma.A.A.\id_A} & {\Gamma.B.B.\id_B} \\
    {\Gamma.A} & {\Gamma.B}
    \arrow["{\chi_0 \subtypeid{t}}", from=1-1, to=1-2]
    \arrow["{\reflmor{A}}", from=2-1, to=1-1]
    \arrow["{\chi_0 t}"', from=2-1, to=2-2]
    \arrow["{\reflmor{B}}"', from=2-2, to=1-2]
    \end{tikzcd}
    &
    \begin{tikzcd}
    [column sep = large]
    {\Gamma.A.A.\id_A.C[\Gamma.t.t]} & {\Gamma.B.B.\id_B.C} \\
    {\Gamma.A.A.\id_A} & {\Gamma.B.B.\id_B}
    \arrow["{(\chi_0 \subtypeid{t}).C}", from=1-1, to=1-2]
    \arrow["{\jmor{A}{C[\Gamma.t.t]}{d[\Gamma.t.t]}}", from=2-1, to=1-1]
    \arrow["{\chi_0 \subtypeid{t}}"', from=2-1, to=2-2]
    \arrow["{\jmor{B}{C}{d}}"', from=2-2, to=1-2]
    \end{tikzcd}
    \end{tabular}
    \]
  where $\Gamma.t.t$ and $d[\Gamma.t.t]$ are given by the universal property of the following pullback square:
  \[\begin{tikzcd}
    {\Gamma.A.A[\pi_A]} & {\Gamma.A} \\
    {\Gamma.A} & {\Gamma.B.B[\pi_B]} & {\Gamma.B} \\
    & {\Gamma.B} & \Gamma
    \arrow["{\pi_A.A}", from=1-1, to=1-2]
    \arrow["{\pi_{A[\pi_A]}}"', from=1-1, to=2-1]
    \arrow["\Gamma.t.t"{description}, dashed, from=1-1, to=2-2]
    \arrow["{\chi_0 t}", from=1-2, to=2-3]
    \arrow["{\chi_0 t}"', from=2-1, to=3-2]
    \arrow["{\pi_B.B}"', from=2-2, to=2-3]
    \arrow["{\pi_{B[\pi_B]}}"', from=2-2, to=3-2]
    \arrow["\lrcorner"{anchor=center, pos=0.125}, draw=none, from=2-2, to=3-3]
    \arrow["{\pi_B}", from=2-3, to=3-3]
    \arrow["{\pi_B}"', from=3-2, to=3-3]
  \end{tikzcd}
  \begin{tikzcd}
    {\Gamma.A} & {\Gamma.B} \\
    & {\Gamma.A.A.\id_A.C[\Gamma.t.t]} & {\Gamma.B.B.\id_B.C} \\
    & {\Gamma.A.A.\id_A} & {\Gamma.B.B.\id_B}
    \arrow["{\chi_0 t}", from=1-1, to=1-2]
    \arrow["{d[\Gamma.t.t]}"{description}, dashed, from=1-1, to=2-2]
    \arrow["{\reflmor{A}}"', curve={height=18pt}, from=1-1, to=3-2]
    \arrow["d", from=1-2, to=2-3]
    \arrow["{\Gamma.t.t.C}"', from=2-2, to=2-3]
    \arrow["{\pi_{C[\Gamma.t.t]}}"', from=2-2, to=3-2]
    \arrow["\lrcorner"{anchor=center, pos=0.125}, draw=none, from=2-2, to=3-3]
    \arrow["{\pi_C}", from=2-3, to=3-3]
    \arrow["{\Gamma.t.t}"', from=3-2, to=3-3]
  \end{tikzcd}\]
  \item \label{item:id-sub-congs}
    $i_{\id_A, -}$ is functorial in that it preserves $i^{\mathsf{iso}}$ and $i^{\mathsf{comp}}$ (see \cref{sec:iso-coherences} for more detail).
  \end{enumerate}
\end{definition}

\cref{item:id-bc,item:id-bc-rj} of \cref{def:comp-cat-with-id} state that $\id$, $r$ and $j$ are preserved under substitution.
\cref{item:id-subt,item:id-subt-id,item:id-subt-comp,item:id-subt-chi,item:id-sub-congs} give the functoriality conditions for expressing the interaction of $\id$-types with subtyping.
\cref{item:id-subt-chi} expresses compatibility of the type morphism structure on the category $\TT$ with the type former structure on the category $\CC$.

\begin{relatedwork}[\citet{gambino:2023}]
  \label{rel:id-gambino}
  \citet{gambino:2023} also define identity types for comprehension categories.
  Their introduction and elimination rules are the same as ours.
  The only difference is that, just like for $\Pi$- and $\Sigma$-types,
  we require identity types to be functorial and we phrase stability using isomorphisms instead of Cartesian morphisms.
\end{relatedwork}

\begin{relatedwork}[\citet{coraglia2024contextjudgementdeduction}]
  \label{rel:id-coraglia}
  \citet{coraglia2024contextjudgementdeduction} study generalized categories with families, which are equivalent to (non-full) comprehension categories.
  In that setting, they define \emph{extensional} identity types.
  Our identity types of \cref{def:comp-cat-with-id} are \emph{intensional}.
\end{relatedwork}

We now discuss examples of comprehension categories with functorial $\id$-types, including those that arise from \AWFSs (see \cref{sec:awfs}).
To interpret identity types in a comprehension category induced by an \AWFSs,
we need to assume additional structure,
namely a \textbf{stable functorial choice of path objects}~\cite[Definition 4.8]{gambino:2023}.
Path objects give us a factorization of the diagonal morphism,
and hence, we obtain an interpretation of the identity type~\cite[Propositon 4.9]{gambino:2023}.
Since the choice of the path object is required to be functorial, comprehension categories induced by weak factorization systems support the functoriality condition in \cref{def:comp-cat-with-id}.
These conditions are satisfied by the examples in \cref{sec:awfs}.
In what follows, we discuss functorial $\id$-types in the these examples.

\begin{example}[\cref{exa:grpd,exa:interp-grpd} ctd.]
  \label{exa:id-grpd}
  In the \AWFS of groupoids, a type $A$ in context $\Gamma$ is interpreted as a split isofibration $\interp{A} : \interp{\Gamma.A} \to \interp{\Gamma}$.
  This example has $\id$-types.
  For each groupoid $\Gamma$ and split isofibration $A : \Gamma.A \to \Gamma$, we have a split isofibration $\id_A : \Gamma.A.A[\pi_A].\id_A \to \Gamma.A.A[\pi_A]$, where for each $x \in \Gamma$ and $a,b \in A_x$, the set of objects in the fiber over $(x,a,b)$ is $\hom(a,b)$, i.e. isomorphisms from $a$ to $b$.
  Reflexivity is given by the functor $r_A : \Gamma.A \to \Gamma.A.A.\id_A$ mapping $(x \in \Gamma, a \in A_x)$ to $(x \in \Gamma, a \in A_x, a \in A_x, \id_a : a \iso a)$.
  For the elimination rule, for each split isofibration $A : \Gamma.A \to \Gamma$, we have a morphism in $\Gamma.A.A[\pi_A].\id_A$ as follows: 
  \[
  t (\alpha \in a, \beta \in b, p) = (\alpha \in a, p^{-1} (\beta) \in a, \id_a).
  \]
  For each split isofibration $C: \Gamma.A.A[\pi_A].\id_A.C \to \Gamma.A.A[\pi_A].\id_A$ and functor $d$ as follows,
  \begin{align*}
  d : \Gamma.A & \to \Gamma.A.A[\pi_A].\id_A.C \\
  (x \in \Gamma, a \in A_x) & \mapsto (x \in \Gamma, a \in A_x, b : A_x, \id_a : a \iso a, d(a) \in C_{a,a,\id_a}),
  \end{align*}
  we have the eliminator $j_{C,d}$: 
  \begin{align*}
    j_{C,d} : \Gamma.A.A[\pi_A].\id_A & \to \Gamma.A.A[\pi_A].\id_A.C \\
    (x \in \Gamma, a \in A_x, b \in A_x, p : a \iso b) & \mapsto (x \in \Gamma, a \in A_x, b \in A_x, p : a \iso b, t^* (d(a)) \in C_{a,b,p}).
  \end{align*}
  Since $C$ is split, we have $j(a,a,\id_a) = d(a)$ for all $x \in \Gamma$ and $a \in A_x$.

  For functoriality, let $t : A \to B$ be a morphism of split fibrations over $\Gamma$ that maps $(x \in \Gamma, a \in A_x)$ to $(x \in \Gamma, t(a) \in B_x)$. We take $\subtypeid{t} : \id_A \to \id_B$ to be the following functor over $\Gamma.t.t$:
  \begin{align*}
  \subtypeid{t} : \Gamma.A.A[\pi_A].\id_A & \to  \Gamma.B.B[\pi_B].\id_B \\
   (x \in \Gamma, a \in A_x, b \in A_x, p: a \iso b) & \mapsto (x \in \Gamma, t(a) \in B_x, t(b) \in B_x, t(p) : t(a) \iso t(b)).
  \end{align*}

  For the comprehension category built from categories, instead of from groupoids, the set of objects in the fiber over $(x,a,b)$ is $\mathsf{iso}(a,b)$ instead of $\hom(a,b)$, for each $x \in \Gamma$ and $a,b \in A_x$.
\end{example}

The identity type on the example of Heyting algebras (cf.~\cref{exa:cha-pred,exa:interp-cha-pred}) is trivial: it denotes equality between proofs, hence gives the singleton type.

\begin{figure}[!htb]
  \centering
  \tcbset{colframe=black, colback=white, width=\textwidth, boxrule=0.1mm, arc=0mm, auto outer arc}
  \begin{tcolorbox}
  \centering
  \scriptsize\
  \begin{tabular}{cc}
    $\unaryRule
    {\Gamma \vdash A \type}{\Gamma.A.A[\pi_A] \vdash \id_A \type}{\scriptsize{id-form}} $ &
    $\unaryRuleBinaryConcl
    {\Gamma \vdash A \type}{\quad \qquad \Gamma.A \vdash \reflmor{A} : \Gamma.A.A[\pi_A].\id_A \qquad \quad}{\Gamma.A \vdash \pi_{\id_A} \circ \reflmor{A} \equiv p_2(1_{\Gamma.A}) : \Gamma.A.A[\pi_A]}{\scriptsize{id-intro}} $
  \end{tabular}
\[
\ternaryRule
{\Gamma.A.A[\pi_A].\id_A \vdash C \type}{\Gamma.A \vdash d : \Gamma.A.A[\pi_A].\id_A.C}
{\Gamma.A \vdash \pi_C \circ d \equiv \reflmor{A} : \Gamma.A.A.\id_A}
{\Gamma.A.A[\pi_A].\id_A \vdash \jmor{A}{C}{d} : C}
{\scriptsize{id-elim}}
\]
\[
\ternaryRule
{\Gamma.A.A[\pi_A].\id_A \vdash C \type}{\Gamma.A \vdash d : \Gamma.A.A[\pi_A].\id_A.C}{\Gamma.A \vdash \pi_C \circ d \equiv \reflmor{A} : \Gamma.A.A.\id_A}{\Gamma.A \vdash \jmor{A}{C}{d} \circ \reflmor{A} \equiv d : \Gamma.A.A[\pi_A].\id_A.C}{\scriptsize{id-beta}}
\]
\[ 
\binaryRule{\Delta \vdash A \type}{\Gamma \vdash s : \Delta}{\Gamma.A[s].A[s][\pi_{A[s]}] ~|~ \id_{A[s]} \vdashiso i_{\id_A, s} : \id_A [s.A.A[\pi_A] \circ \Gamma.A[s].i_{s.A}]}{\scriptsize{id-sub}}
\]
\[
\binaryRuleBinaryConcl{\Delta \vdash A \type}{\Gamma \vdash s : \Delta}{\Gamma.A[s].A[s][\pi_{A[s]}] \vdash (s.A.A[\pi_A] \circ \Gamma.A[s].i_{s.A}).\id_A \circ \Gamma.A[s].A[s][\pi_{A[s]}].i_{\id_A, s} \circ \reflmor{A[s]} }{ \equiv \reflmor{A} \circ s.A.A[\pi_A] \circ \Gamma.A[s].i_{s.A} : \Delta.A.A[\pi_A].\id_A}{\scriptsize{sub-refl}}
\]
\[
\binaryRuleTernaryConcl{\Delta \vdash A \type}{\Gamma \vdash s : \Delta}
{\Gamma.A[s].A[s][\pi_{A[s]}].\id_{A[s]} \vdash ((s.A.A[\pi_A] \circ \Gamma.A[s].i_{s.A}).\id_A \circ \Gamma.A[s].A[s][\pi_{A[s]}].i_{\id_A, s}).C \circ }{\jmor{A[s]}{C[(s.A.A[\pi_A] \circ \Gamma.A[s].i_{s.A}).\id_A \circ \Gamma.A[s].A[s][\pi_{A[s]}].i_{\id_A, s}]}{d[s.A.A[\pi_A] \circ \Gamma.A[s].i_{s.A}]} \equiv  \jmor{A}{C}{d} \circ }{ (s.A.A[\pi_A] \circ \Gamma.A[s].i_{s.A}).\id_A \circ \Gamma.A[s].A[s][\pi_{A[s]}].i_{\id_A, s} : \Delta.A.A[\pi_A].\id_A.C}{\scriptsize{sub-j}}
\]
\[
\unaryRule
{\Gamma ~|~ A \vdash t : B}
{\Gamma.A.A[\pi_A] ~|~ \id_A \vdash \subtypeid{t} : \id_B[\Gamma.t.t]}
{\scriptsize{subt-id}}
\unaryRule
{\Gamma \vdash A \type}{\Gamma.A.A[\pi_A].\id_A \vdash \subtypeid{1_{\Gamma.A}} \equiv 1_{\id_A} : \Gamma.A.A[\pi_A].\id_A[1_{\Gamma.A.A[\pi_A]}]}{\scriptsize{subt-id-i}}
\]
\[
\binaryRule
{\Gamma ~|~ A \vdash t : B}{\Gamma ~|~ B \vdash t' : C}
{\Gamma.A.A[\pi_A] ~|~ \id_A \vdash \subtypeid{t'}[\Gamma.t.t] \circ \subtypeid{t} \equiv \subtypeid{t' \circ t} : \id_C [\Gamma.t'.t'][\Gamma.t.t]}{\scriptsize{subt-id-c}}
\]
\[
\unaryRule
{\Gamma ~|~ A \vdash t : B}{\Gamma.A \vdash (\Gamma.t.t).\id_B \circ \Gamma.A.A.\subtypeid{t} \circ \reflmor{A} \equiv \reflmor{B} \circ \Gamma.t : \Gamma.B.B.\id_B}{\scriptsize{subt-id-refl}}
\]
\[
\unaryRule
{\Gamma ~|~ A \vdash t : B}
{\Gamma.A.A.\id_A \vdash ((\Gamma.t.t).\id_B \circ \Gamma.A.A.\subtypeid{t}).C \circ \jmor{A}{C[\Gamma.t.t]}{d[\Gamma.t.t]} \equiv \jmor{B}{C}{d} \circ (\Gamma.t.t).\id_B \circ \Gamma.A.A.\subtypeid{t} : \Gamma.B.B.\id_B.C}{\scriptsize{subt-id-j}}
\]
\[
\unaryRule{\Gamma \vdash A \type}{\Gamma ~|~ \id_{A[1_\Gamma]} \vdash \subiso{\id_A}{1_\Gamma.A.A[\pi_A] \circ \Gamma.i_{1.A}}{1_\Gamma} \circ {i_{\id_A,1_A}} \equiv {\idisoinv{\id_A}} \circ {\subtypeid{\idiso{A}}} : \id_A [1_\Gamma]}{\scriptsize{id-sub-id}}
\]
\[
\ternaryRuleTernaryConcl{\hspace{3.4cm} \Theta \vdash A \type}{\Gamma \vdash s' : \Delta}{\Delta \vdash s : \Theta \hspace{3.4cm}}
{\Gamma ~|~ \id_{A[s \circ s']} \vdash \compiso{\id_A}{s.A.A[\pi_A]\circ \Gamma.i_{s.A}}{s'.A[s].A[s][\pi_{A[s]}]\circ \Gamma.i_{s'.A[s]}} \circ}
{ \subiso{\id_A}{s \circ s'.A.A[\pi_A] \circ \Gamma.i_{(s\circ s').A}}{{s.A.A[\pi_A]\circ \Gamma.i_{s.A}}\circ{s'.A[s].A[s][\pi_{A[s]}]\circ \Gamma.i_{s'.A[s]}}} \circ i_{\id_A, s\circ s'} \equiv}
{i_{\id_A,s} [s'.A[s].A[s][\pi_{A[s]}] \circ \Gamma.i_{s'.A[s]}] \circ  i_{\id_{A[s]},s'} \circ \subtypeid{\compiso{A}{s}{s'}}: \id_{A[s]}[s'.A[s].A[s][\pi_{A[s]}]\circ \Gamma.i_{s'.A[s]}]}{\scriptsize{id-sub-comp}}
\]
In Rules \ruleref{id-sub}, \ruleref{sub-refl}, \ruleref{sub-j},\ruleref{id-sub-id} and \ruleref{id-sub-comp}, $i_{s.A} \coloneqq \compiso{A}{\pi_A}{s.A}  \circ \subiso{A}{s \circ \pi_{A[s]}}{\pi_A \circ s.A} \circ \compisoinv{A}{s}{\pi_{A[s]}}$. In Rule \ruleref{subt-id-j}, $\Gamma.t.t \coloneqq (\Gamma.t \circ \pi_{A[\pi_A]}).B[\pi_B] \circ p_2(\Gamma.t \circ \pi_A.A)$.
\end{tcolorbox}
\caption{Rules for functorial $\id$-types. Rules \protect\ruleref{id-sub}, \protect\ruleref{sub-refl}, \protect\ruleref{sub-j}, \protect\ruleref{id-sub-id} and \protect\ruleref{id-sub-comp} use the notation introduced in \cref{lemma:s.A}. For example, in Rule \protect\ruleref{id-sub}, $s.A$ is $(s \circ \pi_{A[s]}, \Gamma.(\compiso{A}{s}{\pi_{A[s]}} [\pi_{A[s]}]) \circ p_2(1_{\Gamma.A[s]}))$.}
\label{fig:id-for-terms}
\end{figure}

\begin{definition} \label{def:syntax-with-id-for-terms}
  We define the \emphasize{extension of $\CTT$ by functorial $\id$-types} to consist of the rules in \cref{fig:id-for-terms}.
\end{definition}

\begin{theorem}[Soundness of Rules for $\id$-types] \label{thm:soundness-id-for-terms}
  Any comprehension category with functorial identities models the rules of $\CTT$ and the rules for functorial $\id$-types in \cref{fig:id-for-terms}.
\end{theorem}

\TablesOrNot{
  In \cref{tab:subt-id}, we discuss the meaning of the rules in \cref{fig:subty-for-id} from the subtyping point of view, where type morphisms are seen as witnesses of subtyping relations.

  \begin{table}[h]
    \caption{Meaning of the rules regarding subtyping for $\id$-types when the judgement $\Gamma ~|~ A \vdash t : B$ expresses subtyping.}
    \label{tab:subt-id}
    \begin{tabular}{l p{7cm} l}
    Rules in \cref{fig:subty-for-id} & Meaning under Subtyping & Rule in \cite{coraglia_et_al:LIPIcs.TYPES.2023.3} \\
    \hline
    \ruleref{subt-id} & $A \leq_t B$ gives $\id_A \leq_{\subtypeid{t}} (\id_B[\Gamma.t.t])$. & - \\
    \ruleref{subt-id-i} & $\subtypeid{-}$ preserves identity witnesses. & - \\
    \ruleref{subt-id-c} & $\subtypeid{-}$ preserves composition of witnesses. & - \\
    \end{tabular}
  \end{table}
}
{
\begin{remark}
  \label{rem:subt-id}
  Under the subtyping point of view Rule \ruleref{subt-id} states that $A \leq_t B$ gives $\id_A \leq_{\subtypeid{t}} (\id_B[\Gamma.t.t])$.
Rules \ruleref{subt-id-i} and \ruleref{subt-id-c} states that $\subtypeid{(-)}$ preserves identity and composition of subtyping witnesses, respectively. 

\citet{coraglia_et_al:LIPIcs.TYPES.2023.3} do not discuss $\id$-types.
\end{remark}
}

\section{Strictly Functorial Substitution: $\CTTsplit$}
\label{sec:split}

As discussed in \cref{rem:split-subst}, substitution in \CTT is functorial only up to \emph{isomorphism}, whereas in many type theories, substitution is functorial up to \emph{equality}.
While this makes the type theory easier to use,
it comes at a cost;
the notion of model of such type theories is more restricted,
which makes finding such models more challenging.
In comprehension categories, functoriality of substitution is interpreted as splitness of the fibration $p : \TT \to \CC$. For example, MLTT is interpreted in full \emph{split} comprehension categories.

By replacing those \emph{isomorphisms} that reflect non-splitness of the fibration in the syntax with \emph{equalities}, we obtain $\CTTsplit$, a split version of $\CTT$ with substitution that is functorial up to equality.
For this, one needs to also add a judgement $\Gamma \vdash A \equiv B$ for equality of types.

The split version of the judgements and the rules is presented in \cref{sec:split-rules}.
The judgements and rules of $\CTTsplit$ compare to $\CTT$ precisely as follows.
\begin{enumerate}
  \item There is a judgement $\Gamma \vdash A \equiv B$ for equality of types in $\CTTsplit$, whereas $\CTT$ does not feature such a judgement.
  \item Isomorphisms of types in the following rules of $\CTT$ are equality of types in $\CTTsplit$: \ruleref{sub-id}, \ruleref{sub-comp}, \ruleref{sub-cong}, \ruleref{pi-sub}, \ruleref{sigma-sub}, \ruleref{id-sub}.
  \item The isomorphism terms in the following rules of $\CTT$ are not present in $\CTTsplit$: \ruleref{sub-tm-id}, \ruleref{sub-tm-comp}, \ruleref{tm-sub-coh}, \ruleref{sub-lam}, \ruleref{subt-pi}, \ruleref{sub-pair}, \ruleref{sub-refl}, \ruleref{sub-j}.
  \item Isomorphisms in the following rules of $\CTT$ are identity morphisms in $\CTTsplit$: \ruleref{sub-proj-id}, \ruleref{sub-proj-comp}, \ruleref{subt-pi-id}, \ruleref{subt-sigma-id}
  \item The following coherence rules of $\CTT$, explained in \cref{rem:sub-cong}, are not in $\CTTsplit$: second conclusion of \ruleref{sub-cong}, \ruleref{sub-cong-id}, \ruleref{sub-cong-comp-1}, \ruleref{sub-cong-comp-2}.
  In addition, the following coherence rules are not in $\CTTsplit$: \ruleref{pi-sub-id},\ruleref{pi-sub-comp}, \ruleref{sigma-sub-id}, \ruleref{sigma-sub-comp}, \ruleref{id-sub-id} and \ruleref{id-sub-comp}.
\end{enumerate}

\begin{theorem}[Soundness for $\CTTsplit$]
  \label{thm:split-soundness}
  Every split comprehension category models $\CTTsplit$.
\end{theorem}

\begin{remark}[About Implicit Substitution]
 We could also consider a \emph{strict} syntax with \emph{implicit} substitution rather than with \emph{explicit} substitution.
 Such an implicit substition is automatically functorial up to equality.
 \citet{gambino:2023} provide a splitting construction that can be used to turn models of \CTT into models of a variant with implicit substitution.
 \citet{lumsdaine2015} also provide a splitting construction, but only for \emph{full} comprehension categories.
\end{remark}

\section{Related Work}\label{sec:rel-work}

In this section, we discuss related work and the precise relationship to our work.

\subsection{Work on Type Formers}

Comprehension categories, and semantic structures for the interpretation of type theory, were defined by Jacobs, first in a seminal paper \cite{jacobs93} and, later, in a comprehensive book \cite{jacobs99}.
Jacobs assumes comprehension categories to be full throughout.
We give precise comparisons between our work and that of Jacobs throughout this paper, in dedicated environments (\cref{rel:pi-jacobs,rel:pi-subty-jacobs,rel:sigma-jacobs,rel:sigma-subty-jacobs,rel:id-jacobs}).

Lindgren \cite{lindgren21} defines a semantic structure on non-full comprehension categories suitable for the interpretation of dependent product types; we give more detail in \cref{rel:pi-lindgren}.

Lumsdaine and Warren \cite{lumsdaine2015} discuss a splitting construction for \emph{full} comprehension categories. They define different versions of categorical structures for the interpretation of type formers suitable for full comprehension categories.
For a comparison of our and their structures for type formers, see \cref{rel:pi-lw,rel:sigma-lw}.

\citet{gambino:2023} shows that splitting construction by \citet{hofmann:1994} extends to \emph{non-full} comprehension categories.
They consider comprehension categories arising from \AWFSs, and suitable structure for type formers.
For a comparison of our and their structures for type formers, see \cref{rel:pi-gambino,rel:sig-gambino,rel:id-gambino}.

Ahrens, North, and Van der Weide \cite{ANW23} develop a syntax for comprehension \emphasize{bi}categories, with the goal of developing a notion of directed type theory.
They do not study type formers, but only structural rules.

Curien, Garner, and Hofmann \cite{curien14} develop a type theory with explicit substitution to more accurately reflect the intended categorical semantics: explicit substitution is not necessarily strict, and thus better aligns with the interpretation of substitution as pullback.
They give an interpretation of their type theory in comprehension categories; this interpretation does not make use of morphisms between types, and the authors note that ``[i]t is therefore natural to limit attention to full comprehension categories''.
Like \cite{curien14} we have an explicit substitution operation in the rules we develop; however, regarding morphisms between types, we take a different approach by extending the syntax by a corresponding judgement for such morphisms.

\subsection{Work on Subtyping}

Subtyping in type theory has been studied extensively, from both \emph{semantic} and \emph{syntactic} angles.
We discuss what seems to us the most closely related work to ours; the overview below is by no means claimed to be exhaustive.

We first discuss work studying the \emph{semantics} of subtyping.

Firstly, Zeilberger and Melliès \cite{DBLP:conf/popl/MelliesZ15} give a fibrational view of \emphasize{subsumptive} subtyping, unlike this paper where we discuss \emphasize{coercive} subtyping.
They interpret type systems as functors from a category of type derivations to a category of underlying terms. 
In this setting, subtyping derivations are interpreted as vertical morphisms, i.e. the derivations mapped to the identity morphism of the underlying term.

Secondly, \citet{coraglia_et_al:LIPIcs.TYPES.2023.3} study ``generalized categories with families'', a notion that they show is equivalent to the (non-full) comprehension categories discussed in the present paper \cite{coraglia:2024}.
Taking a semantic viewpoint on subtyping, they sketch \cite{coraglia_et_al:LIPIcs.TYPES.2023.3} how generalized categories with families interpret some rules related to \emphasize{coercive} subtyping, notably the rules of transitivity, subsumption, weakening, substitution, and rules related to the type formers $\Pi$ and $\Sigma$ --- for details, see 
\TablesOrNot
{\cref{tab:subtyping,tab:subty-pi,tab:subty-sigma}.}
{\cref{tab:subtyping,rem:subty-pi,rem:subty-sigma}.}
We develop that work further, by presenting more structural rules, and analyzing  identity types as well.

We point out one potential source of confusion when comparing the type-theoretic rules presented in the present work with the rules shown by \citet[Propositions 20 and 21]{coraglia_et_al:LIPIcs.TYPES.2023.3}.
Specifically, ``[I]n writing the rules above in Propositions 20 and 21, [\citeauthor{coraglia_et_al:LIPIcs.TYPES.2023.3}] have written the action of
reindexing as if the fibrations involved were split.''
That is, even though Coraglia and Emmenegger study comprehension categories that are not necessarily split, they present a simplified versions of their rules which can be interpreted only in \emphasize{split} comprehension categories.
In our work, we present type-theoretic rules suitable for interpretation in any comprehension category, not necessarily split.
As a consequence, some of our rules contain more coherence isomorphisms.

Next, we discuss work studying subtyping from a \emph{syntactic} point of view.


Firstly, \citet{DBLP:journals/mscs/LuoA08} study structural coercive subtyping for inductive data types. 
They propose functoriality of type formers and prove desirable properties, such as admissibility of transitivity of subtyping, of the resulting theory.

Secondly, \citet{DBLP:conf/esop/LaurentLM24} extend MLTT to a type theory with definitionally functorial type formers. 
They use this functoriality to extend MLTT to two type theories with coercive and subsumptive subtyping, respectively. 
They show that the functoriality of type formers is sufficient to establish back-and-forth translations between the two type theories, resulting in an equivalence between them. 
They also study meta-theoretic properties of their systems, e.g. showing that their functorial system is normalizing and has decidable type checking.

A rough comparison of the type theory from \citet{DBLP:conf/esop/LaurentLM24} with coercive subtyping, called \MLTTcoe, to ours is as follows.
The syntax of \MLTTcoe supports at most one coercion between any two types.
This corresponds to considering a thin category of types in our semantics. 
Furthermore, the syntax of \MLTTcoe comes with substitution which is strictly functorial.
This corresponds to considering a split fibration in our semantics.
The split version of our syntax is discussed in \cref{sec:split}.
If we assume splitness and thinness in our syntax, our rules for $\Pi$- and $\Sigma$-types imply theirs.
A more notable difference is that the identity type of \MLTTcoe does not have a counterpart to our Rule \ruleref{subt-id-j} which expresses that the eliminator of identity type is preserved by coercion.
As the authors explain in Section 3.3 of another version of their work \cite{laurent:hal-04160858}, they make a design choice to not include such a rule as it is not necessary for having their desired functorial equations.

Thirdly, \citet{adjedj:hal-05167997} develop a type theory they call \textsf{AdapTT} that captures type casting and coercive subtyping.
They show that \textsf{AdapTT} is modelled by \NatModDO, a structure equivalent to \emph{split} generalized categories with families.
Similar to $\CTTsplit$, substitution in their syntax is strictly functorial.
Furthermore, they provide a general framework, called \textsf{AdapTT$_2$}, for defining type formers that are automatically functorial, including general inductive types specified by a notion of signature.

Fourthly, \citet{DBLP:journals/tcs/AspinallC01} study syntactic properties of dependent type theory with subtyping; in particular, they prove subject reduction and decidability of type-checking for a theory with dependent types and subtyping.

Fifthly, \citet{DBLP:journals/iandc/LuoSX13} study \emphasize{coercive} subtyping, the form of subtyping analyzed semantically in this paper. They show that coercive subtyping provides a conservative extension of a type theory.

\section{Conclusion}

We have presented the judgements and rules of \CTT which reflect the structure of comprehension categories.
Specifically, we have presented structural rules, and rules for type and term formers for dependent pairs, dependent functions, and identity types.
We have also presented categorical structures on comprehension categories that are suitable for the interpretation of the type and term formers, and we have given a sound interpretation of our rules in such comprehension categories.
Furthermore, we have explained how our rules are a form of proof-relevant subtyping, extending work by Coraglia and Emmenegger \cite{coraglia_et_al:LIPIcs.TYPES.2023.3}.
We have given an interpretation of the rules in models arising from algebraic weak factorization systems.

We have not touched on the question of whether our syntax is complete for comprehension categories.
We conjecture that it is, although we have not yet established this rigorously.
One could follow Garner in his development of 2-dimensional models of type theory~\cite{garner:2009a}.
Garner constructs an equivalence between, on the one hand, a category of suitable generalized algebraic theories and, on the other hand, the category of models he is studying.

What does this work buy us?
From a semantic point of view, the rules of \CTT distill the essence of non-full comprehension categories, a semantic structure arising from \AWFSs, which in turn are frequently used for the interpretation of type theories.
From a syntactic point of view, \CTT provides a framework for theories with coercive subtyping.

Additionally, \CTT provides a basis for strictifying type theory with additional definitional equalities, in the following sense.
As discussed in \cref{sec:awfs}, in homotopy-theoretic models of Martin-Löf type theory that arise from \AWFSs (e.g., \Cref{exa:grpd}),
type morphisms are morphisms of algebras: that is, they preserve the algebra structure of types.
In other words -- using the language of homotopy type theory -- in these models, type morphisms are morphisms which preserve transport \emph{strictly}.
Thus, one could add rules to CCTT expressing that type morphisms preserve transport strictly, and these rules would be validated by such models.
Additionally, many commonly used functions in Martin-Löf type theory are algebra morphisms in these models, and thus could be asserted to be type morphisms in rules added to CCTT.
For instance,
both the first projection from a $\Sigma$-type and constant functions are type morphisms in arbitrary \AWFSs;
thus one can add rules to \CTT expressing that these functions commute strictly with transport.
The value of having such rules comes from the prevalence of calculations with transport in type theory (often that the first projection preserves transport strictly).
Indeed, while such calculations are mathematically straightforward and often omitted from accounts in published papers of formalized mathematics (as done, for instance, by \citet{AL19}), they ``pollute'' computer-checked libraries.
See \citet{DBLP:journals/tocl/Sojakova16} for an account of a piece of synthetic homotopy theory that explicitly describes the many instances of such calculations.
Additional definitional equalities simplify such calculations,
because the proof assistants can take over more work from the user.
As such, \CTT integrates into recent research in type theory aimed at justifying and implementing more definitional equalities in type theory~\cite{DBLP:conf/plpv/AltenkirchMS07,DBLP:conf/types/Cockx19,cohen:2017,DBLP:conf/esop/LaurentLM24,DBLP:conf/csl/Strub10}.

We conclude with a question we have left open:
How does the subtyping point of view carry over to the bicategorical type theory and its interpretation in comprehension \emphasize{bi}categories as studied by \citet{ANW23}?

\begin{acks}
  We thank Ambroise Lafont and Arunava Gantait for feedback on drafts of this paper, and Arthur Adjedj, Kenji Maillard, Noam Zeilberger, Thibaut Benjamin, and Tom de Jong for useful discussions of related work.
  We also thank the anonymous referees for their careful reading and helpful comments.
  This research was supported by the NWO project “The Power of Equality” OCENW.M20.380, which is financed by the Dutch Research Council (NWO).
  This work received government funding managed by the French National Research Agency under the France 2030 program, reference ``ANR-22-EXES-0013''.
\end{acks}

\bibliographystyle{ACM-Reference-Format}
\bibliography{bibliography}

\newpage
\appendix

\section{Rules of the Type Theory}
\label{sec:structural-rules}

\begin{figure}[H]
\centering
\tcbset{colframe=black, colback=white, width=\textwidth, boxrule=0.1mm, arc=0mm, auto outer arc}
\begin{tcolorbox}
\centering
\scriptsize
  \[
  \unaryRule{\Gamma \ctx}{\Gamma \vdash 1_\Gamma: \Gamma}{ctx-mor-id}
  \binaryRule
  {\Gamma \vdash s : \Delta}{\Delta \vdash s' : \Theta}{\Gamma \vdash s' \circ s : \Theta}{ctx-mor-comp}
  \]
  \[
  \unaryRuleBinaryConcl{\Gamma \vdash s : \Delta}{\Gamma \vdash s \circ 1_\Gamma \equiv s: \Delta}
  {\Gamma \vdash 1_\Delta \circ s \equiv s : \Delta}
  {ctx-id-unit}
  \ternaryRule{\Gamma \vdash s : \Delta}{\Delta \vdash s' : \Theta}{\Theta \vdash s'' : \Phi}{\Gamma \vdash s'' \circ (s' \circ s) \equiv (s'' \circ s') \circ s : \Phi}{ctx-comp-assoc}
  \]
    \[
    \unaryRule
    {\Gamma \vdash A \type}{\Gamma ~|~ A \vdash 1_A: A}{ty-mor-id}
    \binaryRule
    {\Gamma ~|~ A \vdash t : B}{\Gamma ~|~B \vdash t' : C}{\Gamma ~|~A  \vdash t' \circ t : C}{ty-mor-comp}
    \]
    \[
    \unaryRuleBinaryConcl
    {\Gamma ~|~ A \vdash t : B}{\Gamma ~|~ A \vdash t \circ 1_A \equiv t: B}
    {\Gamma ~|~ A \vdash 1_B \circ t \equiv t: B}{ty-id-unit}
    \ternaryRule
    {\Gamma ~|~ A \vdash t : B}{\Gamma ~|~ B \vdash t' : C}{\Gamma ~|~ C \vdash t'' : D}{\Gamma ~|~ A \vdash t'' \circ (t' \circ t) \equiv (t'' \circ t') \circ t : D}{ty-comp-assoc}
    \]
    \[
    \unaryRule
    {\Gamma \vdash A \type}{\Gamma.A \ctx}{ext-ty}
    \unaryRule
    {\Gamma ~|~ A \vdash t : B}{\Gamma . A \vdash \Gamma . t : \Gamma.B }{ext-tm}
    \unaryRule
    {\Gamma \vdash A \type}{\Gamma . A \vdash \Gamma . 1_A \equiv 1_{\Gamma.A} : \Gamma. A}{ext-id}
    \]
    \[
    \binaryRule
    {\Gamma ~|~ A \vdash t : B}{\Gamma ~|~ B \vdash t' : C}{\Gamma . A \vdash \Gamma . (t' \circ t) \equiv \Gamma. t' \circ \Gamma . t : \Gamma . B}{ext-comp}
    \unaryRule
    {\Gamma \vdash A \type}{\Gamma.A \vdash \pi_A : \Gamma}{ext-proj}
    \unaryRule
    {\Gamma ~|~ A \vdash t : B}{\Gamma . A \vdash \pi_B \circ \Gamma . t \equiv \pi_A : \Gamma}{ext-c}
    \]
    \[
\binaryRule
{\Gamma \vdash s : \Delta}{\Delta \vdash A \type}{\Gamma \vdash A[s] \type}{sub-ty}
\binaryRule
{\Gamma \vdash s : \Delta}{\Delta ~|~ A \vdash t: B}{\Gamma ~|~ A[s] \vdash t[s] : B[s]}{sub-tm}
\]
\[
\binaryRule
{\Gamma \vdash s : \Delta}{\Delta \vdash A \type}{\Gamma ~|~ A[s] \vdash 1_A[s] \equiv 1_{A[s]}: A[s]}{sub-prs-id}
\ternaryRule
{\Gamma \vdash s : \Delta}
{\Delta ~|~A \vdash t : B}{\Delta ~|~ B \vdash t' : C}{\Gamma ~|~ A[s] \vdash (t' \circ t)[s] \equiv t'[s] \circ t[s] : C[s]}{sub-prs-comp}
\]
\[
\unaryRule
{\Gamma \vdash A \type}{\Gamma ~|~ A[1_\Gamma] \vdashiso \idiso{A} : A}{sub-id}
\ternaryRule
{\Gamma \vdash s : \Delta}{\Delta \vdash s' : \Theta}{\Theta \vdash A \type}{\Gamma ~|~ A[s' \circ s] \vdashiso \compiso{A}{s'}{s}  : A[s'][s]}{sub-comp}
\]
\[
\unaryRule
{\Gamma ~|~ A \vdash t: B}{\Gamma ~|~ A[1_\Gamma] \vdash t[1_\Gamma] \equiv \idisoinv{B} \circ t \circ \idiso{A} : B[1_\Gamma]}{sub-tm-id}
\]
\[
\ternaryRule
{\Gamma \vdash s : \Delta}{\Delta \vdash s' : \Theta}{\Theta ~|~ A \vdash t: B}
{\Gamma ~|~ A[s' \circ s] \vdash t[s' \circ s] \equiv \compisoinv{B}{s'}{s} \circ t[s'][s] \circ \compiso{A}{s'}{s} : B[s' \circ s] }{sub-tm-comp}
\]
  \[
\ternaryRule
{\Delta \vdash A \type}
{\Gamma \vdash s : \Delta}
{\Gamma \vdash t : A[s]}
{\Gamma \vdash (s,t) : \Delta.A}{\scriptsize{sub-ext}}
\unaryRule
{\Gamma \vdash s : \Delta.A}
{\Gamma \vdash p_2(s) : A[\pi_A \circ s]}
{\scriptsize{sub-proj}}
\]
\[
\ternaryRuleBinaryConcl
{\Delta \vdash A \type}{\Gamma \vdash s : \Delta}
{\Gamma \vdash t : A[s]}
{\Gamma \vdash \pi_A \circ (s,t) \equiv s : \Delta}{\Gamma \vdash p_2(s,t) \equiv t : \Gamma.A[s]}{\scriptsize{sub-beta}}
\unaryRule
{\Gamma \vdash s : \Delta.A}{\Gamma \vdash (\pi_A \circ s ,p_2(s)) \equiv s : \Delta.A}{\scriptsize{sub-eta}}
\]
\[
\unaryRule{\Gamma \vdash A \type}{\Gamma.A \vdash \pi_A.A[1_\Gamma] \circ p_2(1_{\Gamma.A}) \equiv \Gamma.\idisoinv{A} : \Gamma.A[1_\Gamma]}{\scriptsize{sub-proj-id}}
\]
\[
\ternaryRule{\Gamma \vdash s' : \Delta}{\Delta \vdash s : \Theta}{\Theta \vdash A  \type}
{
  \Gamma.A[s][s'] \vdash \pi_{A[s][s']}.A[s \circ s'] \circ p_2(s.A \circ s'.A[s]) \equiv \Gamma.\compisoinv{A}{s}{s'} : \Gamma.A[s \circ s'] 
}
{\scriptsize{sub-proj-comp}}
\]
\[
\binaryRule
{\Delta ~|~ A \vdash t : B}{\Gamma \vdash s : \Delta}
{\Gamma.A[s] \vdash s.B \circ \Gamma.t[s] \equiv \Delta.t \circ s.A : \Delta.B}{\scriptsize{tm-sub-coh}}
\]
\end{tcolorbox}
\caption{Rules of the type theory. Note that Rule \protect\ruleref{tm-sub-coh} uses the notation introduced in \cref{lemma:s.A}.}
\label{fig:from-comp-cat}
\end{figure}

\begin{figure}[H]
\centering
\tcbset{colframe=black, colback=white, width=\textwidth, boxrule=0.1mm, arc=0mm, auto outer arc}
\begin{tcolorbox}
\centering
\scriptsize
\[
\unaryRule
  {\Gamma \vdash s : \Delta}
  {\Gamma \vdash s \equiv s : \Delta}
  {ctx-eq-refl}
\unaryRule
  {\Gamma \vdash s_1 \equiv s_2 : \Delta}
  {\Gamma \vdash s_2 \equiv s_1 : \Delta}
  {ctx-eq-sym}
\binaryRule
  {\Gamma \vdash s_1 \equiv s_2 : \Delta}
  {\Gamma \vdash s_2 \equiv s_3 : \Delta}
  {\Gamma \vdash s_1 \equiv s_3 : \Delta}
  {ctx-eq-trans}
\]
\[
\binaryRule
  {\Delta \vdash t : \Theta}
  {\Gamma \vdash s_1 \equiv s_2 : \Delta}
  {\Gamma \vdash t \circ s_1 \equiv t \circ s_2 : \Theta}
  {ctx-comp-cong-1}
\binaryRule
  {\Gamma \vdash t : \Delta}
  {\Delta \vdash s_1 \equiv s_2 : \Theta}
  {\Gamma \vdash  s_1 \circ t \equiv s_2 \circ t: \Theta}
  {ctx-comp-cong-2}
\]
\[
\unaryRule
  {\Gamma ~|~ A \vdash t : B}
  {\Gamma ~|~ A \vdash t \equiv t : B}
  {ty-eq-refl}
\unaryRule
  {\Gamma ~|~ A \vdash t_1 \equiv t_2 : B}
  {\Gamma ~|~ A \vdash t_2 \equiv t_1 : B}
  {ty-eq-sym}
\binaryRule
  {\Gamma ~|~ A \vdash t_1 \equiv t_2 : B}
  {\Gamma ~|~ A \vdash t_2 \equiv t_3 : B}
  {\Gamma ~|~ A \vdash t_1 \equiv t_3 : B}
  {ty-eq-trans}
\]
\[
\binaryRule
  {\Gamma ~|~ B \vdash t' : C}
  {\Gamma ~|~ A \vdash t_1 \equiv t_2 : B}
  {\Gamma ~|~ A \vdash t' \circ t_1 \equiv t' \circ t_2 : C}
  {ty-comp-cong-1}
\binaryRule
  {\Gamma ~|~ A \vdash t' : B}
  {\Gamma ~|~ B \vdash t_1 \equiv t_2 : C}
  {\Gamma ~|~ A \vdash t_1 \circ t' \equiv t_2 \circ t' : C}
  {ty-comp-cong-2}
\]
\[
\unaryRule
  {\Gamma ~|~ A \vdash t_1 \equiv t_2 : B}
  {\Gamma.A \vdash \Gamma.t_1 \equiv \Gamma.t_2 : \Gamma.B}
  {ext-cong}
\binaryRule
  {\Gamma \vdash s : \Delta}
  {\Delta ~|~ A \vdash t_1 \equiv t_2 : B}
  {\Gamma ~|~ A[s] \vdash t_1[s] \equiv t_2[s] : B[s]}
  {sub-cong-tm}
\]
\[
\binaryRuleBinaryConcl
{\qquad \Delta \vdash A \type}
{\Gamma \vdash s \equiv s' : \Delta \qquad}
{\Gamma ~|~ A[s] \vdashiso \subiso{A}{s}{s'} : A[s']}
{\Gamma ~|~ A[s'] \vdash \subisoinv{A}{s}{s'} \equiv \subiso{A}{s'}{s} : A[s]}
{\scriptsize{sub-cong}}
\binaryRule{\Delta \vdash A \type}{\Gamma \vdash s : \Delta}{\Gamma ~|~ A[s] \vdash \subiso{A}{s}{s} \equiv 1_{A[s]}: A[s]}{\scriptsize{sub-cong-id}}
\]
\[
\ternaryRule
{\Theta \vdash A \type}
{\Delta \vdash s' : \Theta}
{\Gamma \vdash s_1 \equiv s_2 : \Delta}
{\Gamma ~|~ A[s'][s_1] \vdash \compiso{A}{s'}{s_2} \circ {\subiso{A}{s_1}{s_2}} \equiv \subiso{A}{s' \circ s_1}{s' \circ s_2} \circ \compiso{A}{s'}{s_1} : A[s' \circ s_2]}{\scriptsize{sub-cong-comp-1}}
\]
\[
\ternaryRule{\Theta \vdash A \type}{\Gamma \vdash s' : \Delta}{\Delta \vdash s_1 \equiv s_2 : \Theta}
{\Gamma ~|~ A[s_1][s'] \vdash \compiso{A}{s_2}{s'} \circ \subiso{A}{s_1}{s_2}[s'] \equiv \subiso{A}{s_1\circ s'}{s_2 \circ s'} \circ \compiso{A}{s_1}{s'} : A[s_2 \circ s']}{\scriptsize{sub-cong-comp-2}}
\]
\[
\unaryRuleBinaryConcl
{\qquad \qquad \quad \Gamma \vdash s_1 \equiv s_2 : \Delta.A \qquad \qquad \quad}
{\Gamma \vdash p_1(s_1) \equiv p_1(s_2) : \Delta}
{\Gamma \vdash \Gamma.\subiso{A}{s_1}{s_1} \circ p_2(s_1) \equiv p_2(s_2) : \Gamma.A[s_2]}
{\scriptsize{sub-proj-cong}}
\]
\[
\quadRule
{\Gamma \vdash s_1 \equiv s_2 : \Delta}
{\Gamma \vdash t_1 : A[s_1]}
{\Gamma \vdash t_2 : A[s_2]}
{\Gamma \vdash \Gamma.\subiso{A}{s_1}{s_2} \circ t_1 \equiv t_2 : \Gamma.A[s_2]}
{\Gamma \vdash (s_1,t_1) \equiv (s_2,t_2) : \Delta.A}{\scriptsize{sub-ext-cong}}
\]
\end{tcolorbox}
\caption{Rules of the type theory regarding $\equiv$ being a congruence.}
\label{fig:from-comp-cat-congruence}
\end{figure}

\section{Interpretation of $\CTT$}
\label{sec:interp-struc-rules}

\subsection{Interpretation of Structural Rules}
Let $(\CC, \TT, p, \chi)$ be a comprehension category.

The rules introduced in \cref{sec:from-comp-cat-c-and-t} regarding context and type morphisms are interpreted as follows.
\begin{enumerate}
    \item Rule \ruleref{ctx-mor-id} is interpreted as the identity morphisms in $\CC$. This means $\interp{1_\Gamma} \coloneqq 1_\interp{\Gamma}$, for each context $\Gamma$.
    \item Rule \ruleref{ctx-mor-comp} is interpreted as the composition of morphisms in $\CC$. This means $\interp{s' \circ s} \coloneqq \interp{s'} \circ \interp{s}$ for each context $\Gamma, \Delta$ and $\Theta$ and context morphisms $s$ from $\Gamma$ to $\Delta$ and $s'$ from $\Delta$ to $\Theta$.
    \item Rule \ruleref{ctx-id-unit} is interpreted as the unit laws of identity in $\CC$.
    \item Rule \ruleref{ctx-comp-assoc} is interpreted as the associativity of composition in $\CC$.
    \item Rule \ruleref{ty-mor-id} is interpreted as the identity morphisms in $\TT$. This means $\interp{1_A} \coloneqq 1_{\interp{A}}$, where $A$ is a type in context $\Gamma$.
    \item Rule \ruleref{ty-mor-comp} is interpreted as the composition of morphisms in $\TT$. This means $\interp{t' \circ t} \coloneqq \interp{t'} \circ \interp{t}$ for types $A,B$ and $C$ in context $\Gamma$, term $t$ of type $A$ dependent on $B$ and term $t'$ of type $B$ dependent on $C$.
    \item Rule \ruleref{ty-id-unit} is the unit laws of identity in $\TT$.
    \item Rule \ruleref{ty-comp-assoc} is the associativity of composition in $\TT$.
\end{enumerate}
This means that $\Gamma \vdashiso s : \Delta$ is interpreted as $\interp{s} : \interp{\Gamma} \iso \interp{\Delta}$ in $\CC$ with the inverse $\interp{s^{-1}}$. Similarly, $\Gamma | A \vdashiso t : B$ is interpreted as $\interp{t} : \interp{A} \iso \interp{B}$ in $\TT_\interp{\Gamma}$ with the inverse $\interp{t^{-1}}$.

The rules introduced in \cref{sec:from-comp-cat-ext} regarding comprehension are interpreted as follows.
\begin{enumerate}
     \item Rule \ruleref{ext-ty} is interpreted as the action of $\chi_0$ on the objects of $\TT_\Gamma$. This means $\interp{\Gamma.A} \coloneqq \chi_0 \interp{A}$ for a type $A$ in context $\Gamma$.
    \item Rule \ruleref{ext-tm} is interpreted as the action of $\chi_0$ on the morphisms of $\TT_\Gamma$. This  means $\interp{\Gamma.t} \coloneqq \chi_0 \interp{t}$ for a term $t$ of type $B$ dependent on $A$ in context $\Gamma$.
    \item Rule \ruleref{ext-id} is interpreted as $\chi_0$ preserving identity.
    \item Rule \ruleref{ext-comp} is interpreted as $\chi_0$ preserving composition.
    \item Rule \ruleref{ext-proj} is interpreted as the action of $\chi$ on the objects of $\TT$. This means $\interp{\pi_A} \coloneqq \chi \interp{A}$ for a type $A$ in context $\Gamma$.
    \item Rule \ruleref{ext-c} is the following commuting diagram corresponding to $\chi \interp{t}$ for a term $t$ of type $B$ dependent on $A$ in context $\Gamma$.
     \[
    \begin{tikzcd}
    \interp{\Gamma . A} \arrow[rr, "\interp{\Gamma . t}"] \arrow[rd, "\interp{\pi_A}", swap] && \interp{\Gamma . B} \arrow[ld, "\interp{\pi_B}"] \\
    & \interp{\Gamma} &
    \end{tikzcd}
    \]
\end{enumerate}

The rules introduced in \cref{sec:from-comp-cat-subst} regarding substitution are interpreted as follows.

\begin{enumerate}
    \item Rules \ruleref{sub-ty} and \ruleref{sub-tm} are interpreted as the action of the reindexing functor $\interp{s}^* : \TT_\interp{\Delta} \to \TT_\interp{\Gamma}$ on objects and morphisms respectively. This means $\interp{A[s]} \coloneqq \interp{s}^* \interp{A}$ and $\interp{t[s]} \coloneqq \interp{s}^* \interp{t}$, for contexts $\Gamma$ and $\Delta$, a context morphism $s$ from $\Gamma$ to $\Delta$, types $A$ and $B$ in $\Delta$ and a term $t$ of type $A$ dependent on $B$ in context $\Delta$.
    \item Rules \ruleref{sub-prs-id} and \ruleref{sub-prs-comp} are interpreted as the reindexing functor $\interp{s}^* : \TT_\interp{\Delta} \to \TT_\interp{\Gamma}$ preserving identity and composition respectively, for contexts $\Gamma$ and $\Delta$ and a context morphism $s$ from $\Gamma$ to $\Delta$.
    \item Rule \ruleref{sub-id} is interpreted as the isomorphism $\interp{A[1_\Gamma]} \iso \interp{A}$, which is $1_\interp{\Gamma}^* \interp{A} \iso \interp{A}$, for a type $A$ in context $\Gamma$.
    \item Rule \ruleref{sub-comp} is interpreted as the isomorphism $\interp{A[s' \circ s]} \iso \interp{(A[s'])[s]} $, which is $(\interp{s'} \circ \interp{s})^* \interp{A} \iso \interp{s}^* (\interp{s'}^* \interp{A})$, for a type $A$ in context $\Theta$ and context morphisms $s$ from $\Gamma$ to $\Delta$ and  $s'$ from $\Delta$ to $\Theta$.
    \item Rules \ruleref{sub-tm-id} and \ruleref{sub-tm-comp} are the following commuting diagrams in $\TT_\Gamma$.
    \[
    \begin{tikzcd}[row sep = large, column sep = large]
    {\interp{A[1_\Gamma]}} & {\interp{A}} \\
    {\interp{B[1_\Gamma]}} & {\interp{B}}
    \arrow["{\idiso{\interp{A}}}"',"\iso", from=1-1, to=1-2]
    \arrow["{\interp{t[1_\Gamma]}}"', from=1-1, to=2-1]
    \arrow["{\interp{t}}", from=1-2, to=2-2]
    \arrow["{\idiso{\interp{B}}}"', "\iso", from=2-1, to=2-2]
    \end{tikzcd}
    \quad
    \begin{tikzcd} [row sep = large, column sep = large]
    {\interp{A[u' \circ u]}} & {\interp{A[u'][u]}} \\
    {\interp{B[u' \circ u]}} & {\interp{B[u'][u]}}
    \arrow["{\compiso{\interp{A}}{\interp{u'}}{\interp{u}}}"', "\iso", from=1-1, to=1-2]
    \arrow["{\interp{t[u' \circ u]}}"', from=1-1, to=2-1]
    \arrow["{\interp{t}}", from=1-2, to=2-2]
    \arrow["{\compiso{\interp{B}}{\interp{u'}}{\interp{u}}}"', "\iso", from=2-1, to=2-2]
    \end{tikzcd}
    \]
    \item In Rule \ruleref{sub-ext}, $(s,A)$ is interpreted as $\interp{s}.\interp{A} \circ \interp{t}$.
    \item In Rule \ruleref{sub-proj},  $p_2(s)$ is interpreted as the morphism given by the universal property of the following pullback square in $\CC$.
    \[\begin{tikzcd}[row sep = large, column sep = large]
      {\interp{\Gamma}} \\
      & {\interp{\Gamma.A[s]}} & {\interp{\Delta.A}} \\
      & {\interp{\Gamma}} & {\interp{\Delta}}
      \arrow["{\interp{p_2(s)}}"{description}, dashed, from=1-1, to=2-2]
      \arrow["\interp{s}", curve={height=-12pt}, from=1-1, to=2-3]
      \arrow[equal, curve={height=18pt}, from=1-1, to=3-2]
      \arrow[from=2-2, to=2-3]
      \arrow["\interp{\pi_{A[s]}}"{description},from=2-2, to=3-2]
      \arrow["\lrcorner"{anchor=center, pos=0.125}, draw=none, from=2-2, to=3-3]
      \arrow["\interp{\pi_A}"{description}, from=2-3, to=3-3]
      \arrow["{\interp{\pi_A \circ s}}"', from=3-2, to=3-3]
    \end{tikzcd}\]
    \item Rules \ruleref{sub-proj-id} and \ruleref{sub-proj-comp} correspond to $\chi_0$ mapping $\idisoinv{A}$ and \compisoinv{A}{s}{s'} to the morphisms given by the universal property of the following pullbacks in $\CC$.
    \[\begin{tikzcd}[ampersand replacement=\&]
      {\interp{\Gamma.A}} \\
      \& \interp{{\Gamma.A[1_\Gamma]}} \& \interp{{\Gamma.A}} \\
      \& \interp{\Gamma} \& \interp{\Gamma}
      \arrow[dashed, from=1-1, to=2-2]
      \arrow["{1_\interp{\Gamma.A}}", curve={height=-12pt}, from=1-1, to=2-3]
      \arrow["{\interp{\pi_A}}"', curve={height=12pt}, from=1-1, to=3-2]
      \arrow[from=2-2, to=2-3]
      \arrow["{\interp{\pi_{A[1_\Gamma]}}}"{description}, from=2-2, to=3-2]
      \arrow["{\interp{\pi_A}}"{description}, from=2-3, to=3-3]
      \arrow["{1_\interp{\Gamma}}"', from=3-2, to=3-3]
    \end{tikzcd}
    \begin{tikzcd}[ampersand replacement=\&]
      {\interp{\Gamma.A[s][s']}} \\
      \& {\interp{\Gamma.A[s \circ s']}} \& {\interp{\Theta.A}} \\
      \& {\interp{\Gamma}} \& {\interp{\Theta}}
      \arrow[dashed, from=1-1, to=2-2]
      \arrow["{\interp{s.A \circ s'.A[s]}}", curve={height=-12pt}, from=1-1, to=2-3]
      \arrow["{\interp{\pi_{A[s][s']}}}"', curve={height=12pt}, from=1-1, to=3-2]
      \arrow[from=2-2, to=2-3]
      \arrow["{\interp{\pi_{A[s \circ s']}}}"{description}, from=2-2, to=3-2]
      \arrow["{\interp{\pi_{A}}}"{description}, from=2-3, to=3-3]
      \arrow["{\interp{s \circ s'}}"', from=3-2, to=3-3]
    \end{tikzcd}
    \]
    \item Rule \ruleref{tm-sub-coh} is the following commuting diagram in $\CC$.
    \[\begin{tikzcd}[row sep = large, column sep = large]
	{\interp{\Gamma.A[s]}} & {\interp{\Delta.A}} \\
	{\interp{\Gamma.B[s]}} & {\interp{\Delta.B}}
	\arrow["{\interp{s}.\interp{A}}", from=1-1, to=1-2]
	\arrow["{\chi_0 \interp{t[s]}}"', from=1-1, to=2-1]
	\arrow["{\chi_0 \interp{t}}", from=1-2, to=2-2]
	\arrow["{\interp{s}.\interp{B}}"', from=2-1, to=2-2]
    \end{tikzcd}\]
\end{enumerate}

In Rule \ruleref{sub-cong}, $\subiso{A}{u}{u'}$ is interpreted as the identity morphism $1_{A[u]}$; hence, Rules \ruleref{sub-cong-id}, \ruleref{sub-cong-comp-1} and \ruleref{sub-cong-comp-2} are trivially interpreted. It is easy to see how the rest of the rules in \cref{fig:from-comp-cat-congruence} are interpreted in any comprehension category.

\subsection{Interpretation of Functorial $\Pi$-types }
\label{sec:interp-pi}
The rules presented in \cref{fig:pi-for-terms} regarding $\Pi$-types are interpreted as follows.

The type $\PiT{A}{B}$ in Rule \ruleref{pi-form} is interpreted as the object $\PiT{{\interp{A}}}{\interp{B}}$ in $\TT_{\interp{\Gamma}}$.
The context morphism $\lambda b$ in Rule \ruleref{pi-intro} is interpreted as the morphism $\lambdamor{\interp{A}}{\interp{B}} (\interp{b})$ in $\CC$.
The context morphism $\appmor{A}{B}$ in Rule \ruleref{pi-elim} is interpreted as the morphism $\appmor{\interp{A}}{\interp{B}}$ in $\CC$. The second conclusions of Rules \ruleref{pi-intro} and \ruleref{pi-elim} are validated as these morphisms are sections.
Rule \ruleref{pi-beta} is validated because of the condition $\appmor{A}{B} \circ \lambdamor{A}{B} b[\pi_A] = b$ given in \cref{def:pi-non-full}, where $\interp{p_2(\lambda b \circ \pi_A)}$ is $\interp{\lambda b}[\pi_{\interp{A}}] : \interp{\Gamma.A} \to \interp{\Gamma.A.\PiT{A}{B}[\pi_A]}$.
Rule \ruleref{pi-eta} is the requirement for strongness $\appmor{A}{B} \circ f[\pi_A] = f$ and $\interp{p_2(f \circ \pi_A)}$ is $\interp{f}[\pi_{\interp{A}}]: \interp{\Gamma.A} \to \interp{\Gamma.A.\PiT{A}{B}[\pi_A]}$.
The isomorphism $i_{\PiT{A}{B},s}$ in Rule \ruleref{pi-sub} is interpreted as the isomorphism $i_{\PiT{\interp{A}}{\interp{B}}, \interp{s}}$ in $\TT_{\interp{\Gamma}}$.
Rule \ruleref{sub-lam} is the commuting diagrams given in Point \ref{item:pi-bc-lam} of \cref{def:comp-cat-with-pi}.

In Rule \ruleref{subt-pi}, $\subtypepi{f}{g}$ is interpreted as $\interp{\pi_{\PiT{A}{B}}}.\interp{\PiT{A'}{B'}} \circ \lambda (\chi_0 \interp{g[\appmor{A}{B} [\Gamma.f]]})$, and the second conclusion is the equality given in Point \ref{item:pi-subt-chi} of \cref{def:comp-cat-with-pi}. Rules \ruleref{subt-pi-id} and \ruleref{subt-pi-comp} are the equalities given in Points \ref{item:pi-subt-id} and \ref{item:pi-subt-comp} of \cref{def:comp-cat-with-pi}.

Rules \ruleref{pi-sub-id} and \ruleref{pi-sub-comp} are the commuting diagrams in \cref{sec:iso-coherences}.

\subsection{Interpretation of Functorial $\Sigma$-types}
\label{sec:interp-sigma}
The rules presented in \cref{fig:sigma-for-terms} regarding $\Sigma$-types are interpreted as follows.

The type $\SigmaT{A}{B}$ in Rule \ruleref{sigma-form} is interpreted as the object $\SigmaT{\interp{A}}{\interp{B}}$ in $\TT_{\interp{\Gamma}}$.
The context morphism $\pairmor{A}{B}$ in Rule \ruleref{sigma-intro} is interpreted as the morphism $\pairmor{\interp{A}}{\interp{B}}$ in $\CC$.
The context morphism $\projmor{A}{B}$ in Rule \ruleref{sigma-elim} is interpreted as the morphism $\projmor{\interp{A}}{\interp{B}}$ in $\CC$.
Rules \ruleref{sigma-beta} and \ruleref{sigma-eta} are the conditions $\projmor{A}{B} \circ \pairmor{A}{B} = 1_{\Gamma.A.B}$ and $\pairmor{A}{B} \circ \projmor{A}{B} = 1_{\Gamma.\SigmaT{A}{B}}$ given in \cref{def:sigma-non-full} respectively.
The isomorphism $i_{\SigmaT{A}{B}, s}$ in Rule \ruleref{sigma-sub} is interpreted as the isomorphism $i_{\SigmaT{\interp{A}}{\interp{B}}, \interp{s}}$ in $\TT_{\interp{\Gamma}}$.
Rule \ruleref{sub-pair} is the commuting diagram given in Point \ref{item:sigma-bc-pair} of \cref{def:comp-cat-with-sigma}.

In Rule \ruleref{subt-sigma}, $\subtypesigma{f}{g}$ is interpreted as $\pairmor{\interp{A'}}{\interp{B'}} \circ \chi_0 \interp{f}.{\interp{B'}} \circ \chi_0 \interp{g} \circ \projmor{\interp{A}}{\interp{B}}$, and the second conclusion is the equality given in Point \ref{item:sigma-subt-chi} of \cref{def:comp-cat-with-sigma}. Rules \ruleref{subt-sigma-id} and \ruleref{subt-sigma-comp} are the equalities given in Points \ref{item:sigma-subt-id} and \ref{item:sigma-subt-comp} of \cref{def:comp-cat-with-sigma}.

Rules \ruleref{sigma-sub-id} and \ruleref{sigma-sub-comp} are the commuting diagrams in \cref{sec:iso-coherences}.

\subsection{Interpretation of Functorial $\id$-types}
\label{sec:interp-id}
The rules presented in \cref{fig:id-for-terms} regarding $\id$-types are interpreted as follows.

The type $\id_A$ in Rule \ruleref{id-form} is interpreted as the object $\id_{\interp{A}}$ in $\TT_{\interp{\Gamma.AA[\pi_{A}]}}$.
The context morphism $\reflmor{A}$ in Rule \ruleref{id-intro} is interpreted as the section $\reflmor{\interp{A}}$ of $\pi_{\interp{A}}$ in $\CC$.
The context morphism $\jmor{A}{C}{d}$ in Rule \ruleref{id-elim} is interpreted as the section $\jmor{\interp{A}}{\interp{C}}{\interp{d}}$  of $\pi_{\interp{\id_A}}$ in $\CC$.
Rule \ruleref{id-beta} is the condition $\jmor{A}{C}{d} \circ \reflmor{A} = d$ given in \cref{def:id-non-full}.
The isomorphism $i_{\id_A, s}$ in Rule \ruleref{id-sub} is interpreted as the isomorphism $i_{\id_{\interp{A}}, \interp{s}}$ in $\TT_{\interp{\Gamma.A[{s}].A[s][\pi_{A[s]}]}}$.
Rules \ruleref{sub-refl} and \ruleref{sub-j} are the commuting diagrams given in Point \ref{item:id-bc-rj} of \cref{def:comp-cat-with-id}.

In Rule \ruleref{subt-id}, $\subtypeid{t}$ is interpreted as $\subtypeid{\interp{t}}$. Rules \ruleref{subt-id-i} and \ruleref{subt-id-c} are the equalities given in Points \ref{item:id-subt-id} and \ref{item:id-subt-comp} of \cref{def:comp-cat-with-id}. Rules \ruleref{subt-id-refl} and \ruleref{subt-id-j} are the commuting diagrams given in point \ref{item:id-subt-chi} of \cref{def:comp-cat-with-id}.

Rules \ruleref{id-sub-id} and \ruleref{id-sub-comp} are the commuting diagrams in \cref{sec:iso-coherences}.

\section{Constructions for \cref{def:comp-cat-with-pi,def:comp-cat-with-sigma,def:comp-cat-with-id}}
\label{sec:constructions}

The following construction includes more details about the composition used in Point \ref{item:pi-subt} of \cref{def:comp-cat-with-pi}.
\begin{construction} \label{cons:pi-subtyping}
  Let $(\CC, \TT, p, \chi)$ be a comprehension category with dependent products.
  For each morphism $f : A' \to A$ in $\TT_{\Gamma}$ and $g : B[\chi_0 f] \to B'$ in $\TT_{\Gamma.A'}$, the following composition is in $\CC/\Gamma$:
  \[\begin{tikzcd}[column sep = 8em]
    {\Gamma.\PiT{A}{B}} & {\Gamma.\PiT{A}{B}.\PiT{A'}{B'} [\pi_{\PiT{A}{B}}]} & {\Gamma.\PiT{A'}{B'}}
    \arrow["{{\lambda(\chi_0 g[\appmor{A}{B}[\chi_0 f]])}}", from=1-1, to=1-2]
    \arrow["{\pi_{\PiT{A}{B}}.\PiT{A'}{B'}}", from=1-2, to=1-3]
  \end{tikzcd}\]
  where $\appmor{A}{B}[\chi_0 f]$ and $\chi_0 g[\appmor{A}{B}[\chi_0 f]]$ are given by the universal properties of the following pullback squares.
  \[\begin{tikzcd}
    [column sep = huge]
    {\Gamma.A'.\PiT{A}{B} [\pi_{A'}]} & {\Gamma.A'.\PiT{A}{B}[\pi_A][\chi_0 f]} & {\Gamma.A.\PiT{A}{B} [\pi_A]} \\
    & {\Gamma.A'.B[\chi_0 f]} & {\Gamma.A.B} \\
    & {\Gamma.A'} & {\Gamma.A}
    \arrow["\iso", from=1-1, to=1-2]
    \arrow["{\appmor{A}{B} [\chi_0 f]}"{description}, dashed, from=1-1, to=2-2]
    \arrow["{\pi_{\PiT{A}{B} [\pi_{A'}]}}"', curve={height=12pt}, from=1-1, to=3-2]
    \arrow["{\chi_0 f.\PiT{A}{B}[\pi_A]}", from=1-2, to=1-3]
    \arrow["{\appmor{A}{B}}", from=1-3, to=2-3]
    \arrow[from=2-2, to=2-3]
    \arrow["{\pi_{B[\chi_0 f]}}"', from=2-2, to=3-2]
    \arrow["\lrcorner"{anchor=center, pos=0.125}, draw=none, from=2-2, to=3-3]
    \arrow["{\pi_B}", from=2-3, to=3-3]
    \arrow["{\chi_0 f}"', from=3-2, to=3-3]
  \end{tikzcd}\]
  \[\begin{tikzcd}
    {\Gamma.\PiT{A}{B}.A'[\pi_{\PiT{A}{B}}]} & {\Gamma.A'.B[\chi_0 f]} \\
    & {\Gamma.\PiT{A}{B}.A'[\pi_{\PiT{A}{B}}].B'[\pi_{\PiT{A}{B}}.A']} & {\Gamma.A'.B'} \\
    & {\Gamma.\PiT{A}{B}.A'[\pi_{\PiT{A}{B}}]} & {\Gamma.A'}
    \arrow["{\appmor{A}{B}[\chi_0 f]}", from=1-1, to=1-2]
    \arrow["{\chi_0 g[\appmor{A}{B}[\chi_0 f]]}"{description}, dashed, from=1-1, to=2-2]
    \arrow[equal, curve={height=18pt}, from=1-1, to=3-2]
    \arrow["{\chi_0 g}", from=1-2, to=2-3]
    \arrow[from=2-2, to=2-3]
    \arrow["{\pi_{B'[\pi_{\PiT{A}{B}}.A']}}"', from=2-2, to=3-2]
    \arrow["\lrcorner"{anchor=center, pos=0.125}, draw=none, from=2-2, to=3-3]
    \arrow["{\pi_B'}", from=2-3, to=3-3]
    \arrow["{\pi_{\PiT{A}{B}}.A'}"', from=3-2, to=3-3]
    \end{tikzcd}\]
  The composition is in $\CC/\Gamma$, since $\lambda$ terms are sections.
\end{construction}

The following construction includes more details about the composition used in Point \ref{item:sigma-subt} of \cref{def:comp-cat-with-sigma}.
\begin{construction} \label{cons:sigma-subtyping}
  Let $(\CC, \TT, p, \chi)$ be a comprehension category with strong dependent sums. For each morphism $f : A \to A'$ in $\TT_{\Gamma}$ and $g : B \to B'[\chi_0 f]$ in $\TT_{\Gamma.A}$, the following composite is in $\CC/\Gamma$,
  \[\begin{tikzcd}
	{\Gamma.{\SigmaT{A}{B}}} & {\Gamma.A.B} & {\Gamma.A.B'[\chi_0 f]} & {\Gamma.A'.B'} & {\Gamma.\SigmaT{A'}{B'}}
	\arrow["{\projmor{A}{B}}", from=1-1, to=1-2]
	\arrow["{\chi_0 g}", from=1-2, to=1-3]
	\arrow["{\chi_0 f.B'}", from=1-3, to=1-4]
	\arrow["{\pairmor{A'}{B'}}", from=1-4, to=1-5]
  \end{tikzcd}\]
  since
  \begin{align*}
    \pi_{\SigmaT{A'}{B'}} \circ {\pairmor{A'}{B'}} \circ {\chi_0 f.B'} \circ  & \\
     {\chi_0 g} \circ {\projmor{A}{B}} & =  \pi_{A'} \circ \pi_{B'} \circ {\chi_0 f.B'} \circ {\chi_0 g} \circ {\projmor{A}{B}} & \text{(def. $\pairmor{A}{B}$)}  \\
    & = \pi_{A'} \circ \chi_0 f \circ \pi_{B'[\chi_0 f]} \circ \chi_0 g \circ {\projmor{A}{B}} & \text{($\chi(\chi_0 f)_{B'}$)}\\
    & = \pi_A \circ \pi_{B'[\chi_0 f]} \circ \chi_0 g \circ {\projmor{A}{B}} & \text{($f$ vertical)}\\
    & = \pi_A \circ \pi_{B} \circ {\projmor{A}{B}} & \text{($g$ vertical)} \\
    & = \pi_{\SigmaT{A}{B}}. & \text{(def. $\projmor{A}{B}$)}
  \end{align*}
\end{construction}

\section{Functoriality of $i_{X,s}$ for \cref{def:comp-cat-with-pi,def:comp-cat-with-sigma,def:comp-cat-with-id}}
\label{sec:iso-coherences}
In \cref{def:comp-cat-with-pi}, $i_{\PiT{A}{B},-}$ is functorial in that it preserves $\idiso{\PiT{A}{B}}$ and $\compiso{\PiT{A}{B}}{s}{s'}$. By this we mean that the two following diagrams commute for each suitable $A$, $B$, $s$, and $s'$.
    \[\begin{tikzcd}[ampersand replacement=\&]
      {\PiT{A[1_\Gamma]}{B[1.A]}} \& {\PiT{A}{B}[1_\Gamma]} \\
      {\PiT{A}{B}}
      \arrow["{i_{\PiT{A}{B}, 1_\Gamma}}","{\iso}"', from=1-1, to=1-2]
      \arrow["{\subtypepi{\idisoinv{A}}{i_{B,1_\Gamma.A}}}"', "{\iso}",from=1-1, to=2-1]
      \arrow["{\idiso{\PiT{A}{B}}}","{\iso}"', from=1-2, to=2-1]
    \end{tikzcd}
    \begin{tikzcd}[ampersand replacement=\&]
      {\PiT{A[s \circ s']}{B[(s \circ s').A]}} \& {\PiT{A}{B}[s \circ s']} \\
      {\PiT{A[s][s']}{B[s.A][s'.A[s]]}} \\
      {\PiT{A[s]}{B[s.A]}[s']} \& {\PiT{A}{B}[s][s']}
      \arrow["{i_{\PiT{A}{B},s \circ s'}}","{\iso}"', from=1-1, to=1-2]
      \arrow["{\subtypepi{\compisoinv{A}{s}{s'}}{i_{B,(s \circ s').A}}}"',"{\iso}", from=1-1, to=2-1]
      \arrow["{\compiso{\PiT{A}{B}}{s}{s'}}","{\iso}"', from=1-2, to=3-2]
      \arrow["{i_{\PiT{A[s']}{B[s'.A]}}}"',"{\iso}", from=2-1, to=3-1]
      \arrow["{i_{\PiT{A}{B},s}[s']}"',"{\iso}", from=3-1, to=3-2]
    \end{tikzcd}
    \]
Similarly, in \cref{def:comp-cat-with-sigma,def:comp-cat-with-id}, $i_{\SigmaT{A}{B},-}$ and $i_{\id_A,-}$  are functorial in that the following diagrams commute for each suitable $A$, $B$, $s$, and $s'$.
\[\begin{tikzcd}[ampersand replacement=\&]
      {\SigmaT{A[1_\Gamma]}{B[1.A]}} \& {\SigmaT{A}{B}[1_\Gamma]} \\
      {\SigmaT{A}{B}}
      \arrow["{i_{\SigmaT{A}{B}, 1_\Gamma}}","{\iso}"', from=1-1, to=1-2]
      \arrow["{\subtypesigma{\idiso{A}}{i_{B,1_\Gamma.A}}}"', "{\iso}",from=1-1, to=2-1]
      \arrow["{\idiso{\SigmaT{A}{B}}}","{\iso}"', from=1-2, to=2-1]
    \end{tikzcd}
    \begin{tikzcd}[ampersand replacement=\&]
      {\SigmaT{A[s \circ s']}{B[(s \circ s').A]}} \& {\SigmaT{A}{B}[s \circ s']} \\
      {\SigmaT{A[s][s']}{B[s.A][s'.A[s]]}} \\
      {\SigmaT{A[s]}{B[s.A]}[s']} \& {\SigmaT{A}{B}[s][s']}
      \arrow["{i_{\SigmaT{A}{B},s \circ s'}}","{\iso}"', from=1-1, to=1-2]
      \arrow["{\subtypesigma{\compiso{A}{s}{s'}}{i_{B,(s \circ s').A}}}"',"{\iso}", from=1-1, to=2-1]
      \arrow["{\compiso{\SigmaT{A}{B}}{s}{s'}}","{\iso}"', from=1-2, to=3-2]
      \arrow["{i_{\SigmaT{A[s']}{B[s'.A]}}}"',"{\iso}", from=2-1, to=3-1]
      \arrow["{(i_{\SigmaT{A}{B},s})[s']}"',"{\iso}", from=3-1, to=3-2]
    \end{tikzcd}
    \]
    \[
    \begin{tikzcd}[ampersand replacement=\&]
      {\id_{A[1_\Gamma]}} \& {\id_A[1_\Gamma.A.A[\pi_A] \circ \Gamma.i_{1.A}]} \\
      {\id_A} \& {\id_A [1_\Gamma]}
      \arrow["{i_{\id_A,1_A}}", from=1-1, to=1-2]
      \arrow["{\subtypeid{\idiso{A}}}"', from=1-1, to=2-1]
      \arrow["i^{\mathsf{sub}}", from=1-2, to=2-2]
      \arrow["{\idisoinv{\id_A}}"', from=2-1, to=2-2]
    \end{tikzcd}
    \begin{tikzcd}[ampersand replacement=\&,row sep=large]
      {\id_{A[s \circ s']}} \& {\id_A[s \circ s'.A.A[\pi_A] \circ \Gamma.i_{(s\circ s').A}]} \\
      \& \begin{array}{c} \id_A[s.A.A[\pi_A]\circ \Gamma.i_{s.A}] \\ {[s'.A[s].A[s][\pi_{A[s]}]\circ \Gamma.i_{s'.A[s]}]} \end{array} \\
      {\id_{A[s][s']}} \& {\id_{A[s]}[s'.A[s].A[s][\pi_{A[s]}]\circ \Gamma.i_{s'.A[s]}]}
      \arrow["{i_{\id_A, s\circ s'}}", from=1-1, to=1-2]
      \arrow["{\subtypeid{\compiso{A}{s}{s'}}}"{description}, from=1-1, to=3-1]
      \arrow["{i^{\mathsf{comp}} \circ i^{\mathsf{sub}}}"{description}, from=1-2, to=2-2]
      \arrow["{i^{-1}_{\id_A,s} [s'.A[s].A[s][\pi_{A[s]}]\circ \Gamma.i_{s'.A[s]}]}"{description}, from=2-2, to=3-2]
      \arrow["{i_{\id_{A[s]},s'}}"', from=3-1, to=3-2]
    \end{tikzcd}
    \]
    where $\idiso{}$, $i^{\mathsf{comp}}$ and $i^{\mathsf{sub}}$ are the isomorphisms from 
    \cref{sec:from-comp-cat-judgements-rules}, and we have $i_{B,1_\Gamma.A} : B[1_\Gamma.A][\Gamma.\idiso{}] \iso B$, $i_{B,(s \circ s').A}:B[(s \circ s').A][\Gamma.i^{\mathsf{comp}}] \iso  B[s.A][s'.A[s]]$ and $i_{s.A}: A[s][\pi_{A[s]}] \iso A[\pi_A][s.A]$.
\section{Rules of $\CTTsplit$} 
\label{sec:split-rules}

The judgments of $\CTTsplit$ are as follows:
\begin{enumerate}
    \item $\Gamma \ctx$
    \item $\Gamma \vdash  s : \Delta$
    \item $\Gamma \vdash s \equiv s' : \Delta$
    \item $\Gamma \vdash A \type$
    \item $\Gamma \vdash A \equiv B$ \label{judge:type_eq}
    \item $\Gamma ~|~ A \vdash t : B$
    \item $\Gamma ~|~ A \vdash t \equiv t' : B $
\end{enumerate}
Judgment \ref{judge:type_eq} is not a judgement in $\CTT$.

\begin{figure}[H]
\centering
\tcbset{colframe=black, colback=white, width=\textwidth, boxrule=0.1mm, arc=0mm, auto outer arc}
\begin{tcolorbox}
\centering
\scriptsize
  \[
    \unaryRule{\Gamma \ctx}{\Gamma \vdash 1_\Gamma: \Gamma}{\scriptsize{ctx-mor-id}}
    \binaryRule
    {\Gamma \vdash s : \Delta}{\Delta \vdash s' : \Theta}{\Gamma \vdash s' \circ s : \Theta}{\scriptsize{ctx-mor-comp}}
    \unaryRuleBinaryConcl
    {\Gamma \vdash s : \Delta}{\Gamma \vdash s \circ 1_\Gamma \equiv s: \Delta}{\Gamma \vdash 1_\Delta \circ s \equiv s : \Delta}{\scriptsize{ctx-id-unit}}
    \]
    \[
    \ternaryRule
    {\Gamma \vdash s : \Delta}{\Delta \vdash s' : \Theta}{\Theta \vdash s'' : \Phi}{\Gamma \vdash s'' \circ (s' \circ s) \equiv (s'' \circ s') \circ s : \Phi}{\scriptsize{ctx-comp-assoc}}
    \unaryRule
        {\Gamma \vdash A \type}{\Gamma ~|~ A \vdash 1_A: A}{\scriptsize{ty-mor-id}}
    \]
    \[
    \binaryRule
        {\Gamma ~|~ A \vdash t : B}{\Gamma ~|~B \vdash t' : C}{\Gamma ~|~A  \vdash t' \circ t : C}{\scriptsize{ty-mor-comp}}
    \unaryRuleBinaryConcl
        {\Gamma ~|~ A \vdash t : B}{\Gamma ~|~ A \vdash t \circ 1_A \equiv t: B}{\Gamma ~|~ A \vdash 1_B \circ t \equiv t: B}{\scriptsize{ty-id-unit}}
    \]
    \[
        \ternaryRule{\Gamma ~|~ A \vdash t : B}{\Gamma ~|~ B \vdash t' : C}{\Gamma ~|~ C \vdash t'' : D}{\Gamma ~|~ A \vdash t'' \circ (t' \circ t) \equiv (t'' \circ t') \circ t : D}{\scriptsize{ty-comp-assoc}}
    \]
   \[
       \unaryRule
       {\Gamma \vdash A \type}{\Gamma.A \ctx}{\scriptsize{ext-ty}}
       \unaryRule
       {\Gamma ~|~ A \vdash t : B}{\Gamma . A \vdash \Gamma . t : \Gamma.B }{\scriptsize{ext-tm}}
    \unaryRule
    {\Gamma \vdash A \type}{\Gamma . A \vdash \Gamma . 1_A \equiv 1_{\Gamma.A} : \Gamma. A}{\scriptsize{ext-id}}
\]
  \[
     \binaryRule
     {\Gamma ~|~ A \vdash t : B}{\Gamma ~|~ B \vdash t' : C}{\Gamma . A \vdash \Gamma . (t' \circ t) \equiv \Gamma. t' \circ \Gamma . t : \Gamma . B}{\scriptsize{ext-comp}}
      \unaryRule
        {\Gamma \vdash A \type}{\Gamma.A \vdash \pi_A : \Gamma}{\scriptsize{ext-proj}}
      \unaryRule
        {\Gamma ~|~ A \vdash t : B}{\Gamma . A \vdash \pi_B \circ \Gamma . t \equiv \pi_A : \Gamma}{\scriptsize{ext-c}}
        \]
    \[
       \binaryRule
       {\Gamma \vdash s : \Delta}{\Delta \vdash A \type}{\Gamma \vdash A[s] \type}{\scriptsize{sub-ty}}
       \binaryRule
       {\Gamma \vdash s : \Delta}{\Delta ~|~ A \vdash t: B}{\Gamma ~|~ A[s] \vdash t[s] : B[s]}{\scriptsize{sub-tm}}
    \]
    \[
    \binaryRule
    {\Gamma \vdash s : \Delta}{\Delta \vdash A \type}{\Gamma ~|~ A[s] \vdash 1_A[s] \equiv 1_{A[s]}: A[s]}{\scriptsize{sub-prs-id}}
    \ternaryRule
    {\Gamma \vdash s : \Delta}
    {\Delta ~|~ A \vdash t : B}{\Delta ~|~ B \vdash t' : C}{\Gamma ~|~ A[s] \vdash (t' \circ t)[s] \equiv t'[s] \circ t[s] : C[s]}{\scriptsize{sub-prs-comp}}
    \]
    \begin{tabular}{cc}
    $\unaryRule
    {\Gamma \vdash A \type}{\Gamma \vdash A[1_\Gamma] \equiv A}{\scriptsize{sub-id}}$ &
    $\ternaryRule
    {\Gamma \vdash s : \Delta}{\Delta \vdash s' : \Theta}{\Theta \vdash A \type}{\Gamma A[s' \circ s] \equiv A[s'][s]}{\scriptsize{sub-comp}}$
    \end{tabular}
    \[
    \unaryRule
    {\Gamma ~|~ A \vdash t: B}{\Gamma ~|~ A \vdash t[1_\Gamma] \equiv t : B}{sub-tm-id}
    \ternaryRule
    {\Gamma \vdash s : \Delta}{\Delta \vdash s' : \Theta}{\Theta ~|~ A \vdash t: B}
    {\Gamma ~|~ A[s' \circ s] \vdash t[s' \circ s] \equiv t[s'][s] : B[s' \circ s] }{sub-tm-comp}
    \]
    \[
    \binaryRule
    {\Gamma \vdash s : \Delta}{\Gamma \vdash t : A[s]}
    {\Gamma \vdash (s,t) : \Delta.A}{\scriptsize{sub-ext}}
    \unaryRule
    {\Gamma \vdash s : \Delta.A}{
      \Gamma \vdash p_2(s) : A[\pi_A \circ s]}
      {\scriptsize{sub-proj}}
    \unaryRule
    {\Gamma \vdash s : \Delta.A}{\Gamma \vdash (\pi_A \circ s ,p_2(s)) \equiv s : \Delta.A}{\scriptsize{sub-eta}}
    \]
    \[
    \binaryRuleBinaryConcl
    {\Gamma \vdash s : \Delta}{\Gamma \vdash t : A[s]}
    {\Gamma \vdash \pi_A \circ (s,t) \equiv s : \Delta}{\Gamma \vdash p_2(s,t) \equiv t : \Gamma.A[s]}{\scriptsize{sub-beta}}
\binaryRule
{\Delta ~|~ A \vdash t : B}{\Gamma \vdash s : \Delta}
{\Gamma.A[s] \vdash s.B \circ \Gamma.t[s] \equiv \Delta.t \circ s.A : \Delta.B}{\scriptsize{tm-sub-coh}}
\]
\[
\unaryRule{\Gamma \vdash A \type}{\Gamma.A \vdash \pi_A.A[1_\Gamma] \circ p_2(1_{\Gamma.A}) \equiv \Gamma.1_A : \Gamma.A}{\scriptsize{sub-proj-id}}
\]
\[
\ternaryRule{\Gamma \vdash s' : \Delta}{\Delta \vdash s : \Theta}{\Theta \vdash A  \type}
{
  \Gamma.A[s][s'] \vdash \pi_{A[s][s']}.A[s \circ s'] \circ p_2(s.A \circ s'.A[s]) \equiv \Gamma.1_A[s][s'] : \Gamma.A[s][s'] 
}
{\scriptsize{sub-proj-comp}}
\]
\[
\unaryRule
  {\Gamma \vdash s : \Delta}
  {\Gamma \vdash s \equiv s : \Delta}
  {ctx-eq-refl}
\unaryRule
  {\Gamma \vdash s_1 \equiv s_2 : \Delta}
  {\Gamma \vdash s_2 \equiv s_1 : \Delta}
  {ctx-eq-sym}
\binaryRule
  {\Gamma \vdash s_1 \equiv s_2 : \Delta}
  {\Gamma \vdash s_2 \equiv s_3 : \Delta}
  {\Gamma \vdash s_1 \equiv s_3 : \Delta}
  {ctx-eq-trans}
\]
\[
\binaryRule
  {\Delta \vdash s' : \Theta}
  {\Gamma \vdash s_1 \equiv s_2 : \Delta}
  {\Gamma \vdash s' \circ s_1 \equiv s' \circ s_2 : \Theta}
  {ctx-comp-cong-1}
\binaryRule
  {\Gamma \vdash s' : \Delta}
  {\Delta \vdash s_1 \equiv s_2 : \Theta}
  {\Gamma \vdash  s_1 \circ s' \equiv s_2 \circ s': \Theta}
  {ctx-comp-cong-2}
\]
\[
\unaryRule
  {\Gamma ~|~ A \vdash t : B}
  {\Gamma ~|~ A \vdash t \equiv t : B}
  {ty-eq-refl}
\unaryRule
  {\Gamma ~|~ A \vdash t_1 \equiv t_2 : B}
  {\Gamma ~|~ A \vdash t_2 \equiv t_1 : B}
  {ty-eq-sym}
\binaryRule
  {\Gamma ~|~ A \vdash t_1 \equiv t_2 : B}
  {\Gamma ~|~ A \vdash t_2 \equiv t_3 : B}
  {\Gamma ~|~ A \vdash t_1 \equiv t_3 : B}
  {ty-eq-trans}
\]
\[
\binaryRule
  {\Gamma ~|~ B \vdash t' : C}
  {\Gamma ~|~ A \vdash t_1 \equiv t_2 : B}
  {\Gamma ~|~ A \vdash t' \circ t_1 \equiv t' \circ t_2 : C}
  {ty-comp-cong-1}
\binaryRule
  {\Gamma ~|~ A \vdash t' : B}
  {\Gamma ~|~ B \vdash t_1 \equiv t_2 : C}
  {\Gamma ~|~ A \vdash t_1 \circ t' \equiv t_2 \circ t' : C}
  {ty-comp-cong-2}
\]
\[
\unaryRule
  {\Gamma ~|~ A \vdash t_1 \equiv t_2 : B}
  {\Gamma.A \vdash \Gamma.t_1 \equiv \Gamma.t_2 : \Gamma.B}
  {ext-cong}
\binaryRule{\Delta \vdash A \type}{\Gamma \vdash s \equiv s' : \Delta}{\Gamma \vdash A[s] \equiv A[s']}{\scriptsize{sub-cong}}
\binaryRule
  {\Gamma \vdash s : \Delta}
  {\Delta ~|~ A \vdash t_1 \equiv t_2 : B}
  {\Gamma ~|~ A[s] \vdash t_1[s] \equiv t_2[s] : B[s]}
  {sub-cong-tm}
\]
\[
\ternaryRule
{\Delta \vdash A \type}
{\Gamma \vdash s_1 \equiv s_2 : \Delta}
{\Gamma \vdash t_1 \equiv t_2 : A[s_1]}
{\Gamma \vdash (s_1,t_1) \equiv (s_2,t_2) : \Delta.A}{\scriptsize{sub-ext-cong}}
\binaryRuleBinaryConcl
{\Delta \vdash A \type}
{\Gamma \vdash s_1 \equiv s_2 : \Delta.A}
{\Gamma \vdash p_1(s_1) \equiv p_1(s_2) : \Delta}{\Gamma \vdash p_2(s_1) \equiv p_2(s_2) : A[s_1]}{\scriptsize{sub-proj-cong}}
\]
\end{tcolorbox}
\caption{Structural Rules of $\CTTsplit$. Note that Rule \protect\ruleref{tm-sub-coh} uses the notation introduced in \cref{lemma:s.A}.}
\label{fig:from-comp-cat-congruence-split}
\end{figure}

\begin{figure}[!htb]
  \centering
  \tcbset{colframe=black, colback=white, width=\textwidth, boxrule=0.1mm, arc=0mm, auto outer arc}
  \begin{tcolorbox}
  \centering
  \scriptsize
  \[
  \binaryRule{\Gamma \vdash A \type}{\Gamma.A \vdash B \type}{\Gamma \vdash \PiT{A}{B} \type}{\scriptsize{pi-form}}
  \ternaryRule{\Delta \vdash A \type}{\Delta.A \vdash B \type}{\Gamma \vdash s : \Delta}{\Gamma \vdash \PiT{A[s]}{B[s.A]} \equiv \PiT{A}{B}[s]}{\scriptsize{pi-sub}}
  \unaryRule
  {\Gamma.A \vdash b : B}
  {\Gamma \vdash \lambda b : \PiT{A}{B}}
  {\scriptsize{pi-intro}}
  \]
  \[
  \binaryRuleBinaryConcl{\qquad \qquad \quad \Gamma \vdash A \type}{\Gamma.A \vdash B \type \qquad \qquad \quad}{\Gamma.A.\PiT{A}{B}[\pi_A] \vdash \appmor{A}{B} : \Gamma.A.B}{\Gamma.A.\PiT{A}{B}[\pi_A] \vdash \pi_B \circ \appmor{A}{B} \equiv \pi_{\PiT{A}{B}[\pi_A]} : \Gamma.A}{\scriptsize{pi-elim}}
  \ternaryRule{\Gamma \vdash A \type}{\Gamma.A \vdash B \type}{\Gamma.A \vdash b : B}
  {\Gamma. A \vdash \appmor{A}{B} \circ p_2(\lambda b \circ \pi_A) \equiv b : \Gamma.A.B}{\scriptsize{pi-beta}}
  \]
  \[
  \unaryRule
  {\Gamma \vdash f : \PiT{A}{B}}
  {\Gamma \vdash \lambda (\appmor{A}{B} \circ p_2(f \circ \pi_A)) \equiv f : \Gamma.\PiT{A}{B}}{\scriptsize{pi-eta}}
  \]
  \[
  \binaryRule
  {\Gamma \vdash s : \Delta}{\Gamma \vdash b : \PiT{A}{B}}
  {\Gamma \vdash \lambdamor{A}{B}(b) \circ s \equiv s.\PiT{A}{B} \circ \lambdamor{A[s]}{B[s.A]} (p_2(b \circ s.A)): \Delta.\PiT{A}{B}}{\scriptsize{sub-lam}}
  \]
  \begin{tabular}{cc}
    $\binaryRule{\Gamma \vdash A \type}{\Gamma.A \vdash B \type}{\Gamma \vdash \SigmaT{A}{B} \type}{\scriptsize{sigma-form}}$ &
    $ \binaryRuleBinaryConcl{\Gamma \vdash A \type}{\Gamma.A \vdash B \type}{\Gamma.A.B \vdash \pairmor{A}{B} : \Gamma.\SigmaT{A}{B}}{\Gamma \vdash \pi_{\SigmaT{A}{B}} \circ \pairmor{A}{B} \equiv \pi_A \circ \pi_B : \Gamma}{\scriptsize{sigma-intro}} $
  \end{tabular}
  \begin{tabular}{cc}
    $ \binaryRule{\Gamma \vdash A \type}{\Gamma.A \vdash B \type}{\Gamma.\SigmaT{A}{B} \vdash \projmor{A}{B} : \Gamma.A.B}{\scriptsize{sigma-elim}}$ &
    $ \binaryRule{\Gamma \vdash A \type}{\Gamma.A \vdash B \type}{\Gamma.A.B \vdash \projmor{A}{B} \circ \pairmor{A}{B} \equiv 1_{\Gamma.A.B} : \Gamma.A.B}{\scriptsize{sigma-beta}} $
  \end{tabular}
  \[\binaryRule{\Gamma \vdash A \type}{\Gamma.A \vdash B \type}{\Gamma.\SigmaT{A}{B} \vdash \pairmor{A}{B} \circ \projmor{A}{B} \equiv 1_{\Gamma.\SigmaT{A}{B}} : \Gamma.\SigmaT{A}{B}}{\scriptsize{sigma-eta}}\] 
  \[ \ternaryRule{\Delta \vdash A \type}{\Delta.A \vdash B \type}{\Gamma \vdash s : \Delta}{\Gamma \vdash \SigmaT{A[s]}{B[s.A]} \equiv \SigmaT{A}{B}[s]}{\scriptsize{sigma-sub}}
  \]
  \[ \ternaryRule{\Delta \vdash A \type}{\Delta.A \vdash B \type}{\Gamma \vdash s : \Delta}{\Gamma.A[s].B[s.A] \vdash s.\SigmaT{A}{B} \circ \pairmor{A[s]}{B[s.A]} \equiv \pairmor{A}{B} \circ s.A.B : \Delta.\SigmaT{A}{B}}{\scriptsize{sub-pair}} \]
    \begin{tabular}{cc}
    $\unaryRule
    {\Gamma \vdash A \type}{\Gamma.A.A[\pi_A] \vdash \id_A \type}{\scriptsize{id-form}} $ &
    $\unaryRuleBinaryConcl
    {\Gamma \vdash A \type}{\quad \qquad \Gamma.A \vdash \reflmor{A} : \Gamma.A.A[\pi_A].\id_A \qquad \quad}{\Gamma.A \vdash \pi_{\id_A} \circ \reflmor{A} \equiv p_2(1_{\Gamma.A}) : \Gamma.A.A[\pi_A]}{\scriptsize{id-intro}} $
  \end{tabular}
    \[
    \ternaryRule
    {\Gamma.A.A[\pi_A].\id_A \vdash C \type}{\Gamma.A \vdash d : \Gamma.A.A[\pi_A].\id_A.C}
    {\Gamma.A \vdash \pi_C \circ d \equiv \reflmor{A} : \Gamma.A.A.\id_A}
    {\Gamma.A.A[\pi_A].\id_A \vdash \jmor{A}{C}{d} : C}
    {\scriptsize{id-elim}}
    \]
    \[
    \ternaryRule
    {\Gamma.A.A[\pi_A].\id_A \vdash C \type}{\Gamma.A \vdash d : \Gamma.A.A[\pi_A].\id_A.C}{\Gamma.A \vdash \pi_C \circ d \equiv \reflmor{A} : \Gamma.A.A.\id_A}{\Gamma.A \vdash \jmor{A}{C}{d} \circ \reflmor{A} \equiv d : \Gamma.A.A[\pi_A].\id_A.C}{\scriptsize{id-beta}}
    \]
    \[
    \binaryRule{\Delta \vdash A \type}{\Gamma \vdash s : \Delta}{\Gamma.A[s].A[s][\pi_{A[s]}] \vdash \id_{A[s]} \equiv \id_A [s.A.A[\pi_A] \circ \Gamma.A[s].i^{\mathsf{comp}}_A]}{\scriptsize{id-sub}}
    \]
    \[
    \binaryRule{\Delta \vdash A \type}{\Gamma \vdash s : \Delta}{\Gamma.A[s].A[s][\pi_{A[s]}] \vdash s.A.A[\pi_A].\id_A  \circ \reflmor{A[s]}  \equiv \reflmor{A} \circ s.A.A[\pi_A] : \Delta.A.A[\pi_A].\id_A}{\scriptsize{sub-refl}}
    \]
    \[
    \binaryRuleBinaryConcl{\Delta \vdash A \type}{\Gamma \vdash s : \Delta}
    {\Gamma.A[s].A[s][\pi_{A[s]}].\id_{A[s]} \vdash s.A.A[\pi_A].\id_A.C \circ \jmor{A[s]}{C[s.A.A[\pi_A].\id_A]}{d[s.A.A[\pi_A]]} \equiv}
    {\jmor{A}{C}{d} \circ s.A.A[\pi_A].\id_A : \Delta.A.A[\pi_A].\id_A.C}{\scriptsize{sub-j}}
    \]
\end{tcolorbox}
\caption{Rules for $\Pi$-, $\Sigma$- and $\id$-types in $\CTTsplit$. Rules \protect\ruleref{pi-sub}, \protect\ruleref{sub-lam}, \protect\ruleref{sigma-sub}, \protect\ruleref{sub-pair}, \protect\ruleref{id-sub}, \protect\ruleref{sub-refl} and \protect\ruleref{sub-j} use the notation introduced in \cref{lemma:s.A}. For example, in Rule \protect\ruleref{pi-sub}, $s.A$ is $(s \circ \pi_{A[s]}, p_2(1_{\Gamma.A[s]}))$.}
\label{fig:type-formers-split}
\end{figure}

\begin{figure}[!htb]
  \centering
  \tcbset{colframe=black, colback=white, width=\textwidth, boxrule=0.1mm, arc=0mm, auto outer arc}
  \begin{tcolorbox}
  \centering
  \scriptsize
  \[
  \quadRuleBinaryConcl
  {\qquad \qquad \qquad \Gamma.A \vdash B \type}{\Gamma.A' \vdash B' \type}{\Gamma ~|~ A' \vdash f : A}{\Gamma.A' ~|~ B[\Gamma.f] \vdash g : B' \qquad \qquad \qquad}
  {\Gamma ~|~ \PiT{A}{B} \vdash \subtypepi{f}{g} : \PiT{A'}{B'}}
  {\Gamma. \PiT{A}{B} \vdash \Gamma.\subtypepi{f}{g} \equiv \Gamma.g \circ \lambda(p_2(\Gamma.g \circ (\pi_{\PiT{A}{B}[\pi_{A'}]}.B[\chi_0 f] \circ p_2 (\appmor{A}{B} \circ \Gamma.f.\PiT{A}{B}[\pi_A])))) : \Gamma. \PiT{A'}{B'}}
  {\scriptsize{subt-pi}}
  \]
  \[
  \binaryRule
  {\Gamma \vdash A \type}{\Gamma.A \vdash B \type}
  {\Gamma ~|~ \PiT{A}{B} \vdash \subtypepi{1_{A}}{1_B} \equiv 1_{\PiT{A}{B}} : \PiT{A}{B}}
  {\scriptsize{subt-pi-id}}
  \]
  \[
  \senRule{\Gamma.A \vdash B \type}{\Gamma.A' \vdash B' \type}
  {\Gamma ~|~ A' \vdash f : A}
  {\Gamma.A ~|~ B[\Gamma.f] \vdash g : B'}
  {\Gamma ~|~ A'' \vdash f' : A'}
  {\Gamma.A' ~|~ B'[\Gamma.f'] \vdash g' : B''}
  {\Gamma ~|~ \PiT{A}{B} \vdash \subtypepi{f \circ f'}{g' \circ g[\Gamma.f']} = \subtypepi{f'}{g'} \circ \subtypepi{f}{g} : \PiT{A''}{B''}}
  {\scriptsize{subt-pi-comp}}
  \]
  \[
  \quadRuleBinaryConcl
  {\Gamma.A \vdash B \type}{\Gamma.A' \vdash B' \type}{\Gamma ~|~ A \vdash f : A'}{\Gamma.A ~|~ B \vdash g : B'[\Gamma.f]}
  {\Gamma ~|~ \SigmaT{A}{B} \vdash \subtypesigma{f}{g} : \SigmaT{A'}{B'}}
  {\Gamma.\SigmaT{A}{B} \vdash \Gamma.\subtypesigma{f}{g} \equiv \pairmor{A'}{B'} \circ (\Gamma.f).B' \circ \Gamma.g \circ \projmor{A}{B} : \Gamma.\SigmaT{A'}{B'}}
  {\scriptsize{subt-sigma}}
  \]
  \[
  \binaryRule
  {\Gamma \vdash A \type}{\Gamma.A \vdash B \type}
  {\Gamma ~|~ \SigmaT{A}{B} \vdash \subtypesigma{1_{A}}{1_B} \equiv 1_{\SigmaT{A}{B}} : \SigmaT{A}{B}}
  {\scriptsize{subt-sigma-id}}
  \]
  \[
  \senRule{\Gamma.A \vdash B \type}{\Gamma.A' \vdash B' \type}{\Gamma ~|~ A \vdash f : A'}{\Gamma.A ~|~ B \vdash g : B'[\Gamma.f]}{\Gamma ~|~ A' \vdash f' : A''}{\Gamma.A' ~|~ B' \vdash g' : B''[\Gamma.f']}
  {\Gamma ~|~ \SigmaT{A}{B} \vdash \subtypesigma{f' \circ f}{g'[\Gamma.f] \circ g} \equiv \subtypesigma{f'}{g'} \circ \subtypesigma{f}{g} : \SigmaT{A''}{B''}}
  {\scriptsize{subt-sigma-comp}}
  \]
  \[
\unaryRule
{\Gamma ~|~ A \vdash t : B}
{\Gamma.A.A[\pi_A] ~|~ \id_A \vdash \subtypeid{t} : \id_B[\Gamma.t.t]}
{\scriptsize{subt-id}}
\unaryRule
{\Gamma \vdash A \type}{\Gamma.A.A[\pi_A].\id_A \vdash \subtypeid{1_{\Gamma.A}} \equiv 1_{\id_A} : \Gamma.A.A[\pi_A].\id_A[1_{\Gamma.A.A[\pi_A]}]}{\scriptsize{subt-id-i}}
\]
\[
\binaryRule
{\Gamma ~|~ A \vdash t : B}{\Gamma ~|~ B \vdash t' : C}
{\Gamma.A.A[\pi_A] ~|~ \id_A \vdash \subtypeid{t'}[\Gamma.t.t] \circ \subtypeid{t} \equiv \subtypeid{t' \circ t} : \id_C [\Gamma.t'.t'][\Gamma.t.t]}{\scriptsize{subt-id-c}}
\]
\[
\unaryRule
{\Gamma ~|~ A \vdash t : B}{\Gamma.A \vdash (\Gamma.t.t).\id_B \circ \Gamma.A.A.\subtypeid{t} \circ \reflmor{A} \equiv \reflmor{B} \circ \Gamma.t : \Gamma.B.B.\id_B}{\scriptsize{subt-id-refl}}
\]
\[
\unaryRule
{\Gamma ~|~ A \vdash t : B}
{\Gamma.A.A.\id_A \vdash ((\Gamma.t.t).\id_B \circ \Gamma.A.A.\subtypeid{t}).C \circ \jmor{A}{C[\Gamma.t.t]}{d[\Gamma.t.t]} \equiv \jmor{B}{C}{d} \circ (\Gamma.t.t).\id_B \circ \Gamma.A.A.\subtypeid{t} : \Gamma.B.B.\id_B.C}{\scriptsize{subt-id-j}}
\]
where $\Gamma.t.t \coloneqq (\Gamma.t \circ \pi_{A[\pi_A]}).B[\pi_B] \circ p_2(\Gamma.t \circ \pi_A.A)$
\end{tcolorbox}
\caption{Rules for subtyping for $\Pi$-, $\Sigma$- and $\id$-types in $\CTTsplit$. Rule \protect\ruleref{subt-sigma} uses the notation introduced in \cref{lemma:s.A}. By this notation, $(\Gamma.f).B'$ is $(\Gamma.f \circ \pi_{B'[\Gamma.f]}, p_2(1_{\Gamma.B'[\Gamma.f]}))$}
\label{fig:subty-split}
\end{figure}

\end{document}